%% file: main.tex
\newif\ifTR
\TRtrue

\ifTR
  \documentclass[acmsmall,nonacm]{acmart}
  \makeatletter
  \let\@authorsaddresses\@empty
  \makeatother
\else
  \documentclass[acmsmall]{acmart}
\fi

\usepackage[linesnumbered,noend,ruled]{algorithm2e}
\usepackage[capitalize]{cleveref}
\usepackage{colortbl}
\usepackage{paralist}
\usepackage{pifont}
\usepackage{subcaption}
\usepackage{tikz}
\usepackage{xspace}
\usepackage{thm-restate}
\usepackage{wrapfig}
\usepackage[dvipsnames,table]{xcolor}

\usepackage{tabularray}
\UseTblrLibrary{booktabs}

\usetikzlibrary{arrows}
\usetikzlibrary{automata}
\usetikzlibrary{petri}
\usetikzlibrary{positioning}

\usepackage[xcolor,leftbars]{changebar}
\cbcolor{red}
\nochangebars

\input{macros.tex}

\ifTR
\else
\setcopyright{cc}
\setcctype{by}
\acmDOI{10.1145/3808281}
\acmYear{2026}
\acmJournal{PACMPL}
\acmVolume{10}
\acmNumber{PLDI}
\acmArticle{203}
\acmMonth{6}
\acmSubmissionID{pldi26main-p181-p}
\received{2025-11-14}
\received[accepted]{2026-04-03}
\fi

\ccsdesc[500]{Security and privacy~Denial-of-service attacks}
\ccsdesc[500]{Theory of computation~Automata over infinite objects}
\ccsdesc[500]{Theory of computation~Regular languages}
\ccsdesc[300]{Applied computing~Document searching}

\keywords{
regular expression matching,
backreferences,
ReDoS,
determinization,
Antimirov's derivatives,
register automata,
register set automata
}

\begin{document}

\ifTR
\title{Towards Efficient Matching of Regexes with Backreferences using Register
  Set Automata (Technical Report)}
\else
\title{\cbstart Towards Efficient Matching of Regexes with Backreferences using
  Register Set Automata\cbend}
\fi

\author{Vojtěch Havlena}
\orcid{0000-0003-4375-7954}
\affiliation{%
  \institution{Brno University of Technology}
  \city{Brno}
  \country{Czech Republic}
}
\email{ihavlena@fit.vutbr.cz}

\author{Lukáš Holík}
\orcid{0000-0001-6957-1651}
\affiliation{%
  \institution{Brno University of Technology}
  \city{Brno}
  \country{Czech Republic}
}
\affiliation{%
  \institution{Aalborg University}
  \city{Aalborg}
  \country{Denmark}
}
\email{holik@fit.vutbr.cz}

\author{Ondřej Lengál}
\orcid{0000-0002-3038-5875}
\affiliation{%
  \institution{Brno University of Technology}
  \city{Brno}
  \country{Czech Republic}
}
\email{lengal@fit.vutbr.cz}

\author{Jan Vašák}
\orcid{0009-0006-4410-0398}
\affiliation{%
  \institution{Brno University of Technology}
  \city{Brno}
  \country{Czech Republic}
}
\email{xvasak01@stud.fit.vutbr.cz}

\author{Sabína Gulčíková}
\orcid{0009-0002-1783-9390}
\affiliation{%
  \institution{Brno University of Technology}
  \city{Brno}
  \country{Czech Republic}
}
\email{xgulci00@stud.fit.vutbr.cz}

\begin{abstract}
  Matching regexes (regular expressions) is a~common problem in many areas of
  computer science, with requirements on high speed and robust performance.
  Regexes with backreferences allow one to express certain patterns (even beyond regular) concisely,
  however, since the matching is usually done by backtracking, the matching
  speed can degrade to a degree that constitutes a service failure or a security threat.
  To facilitate high-speed matching of such regexes,
  we propose \emph{register set automata} ($\rsa$s), an extension of register automata
  where registers can contain \emph{sets} of symbols (from a~potentially infinite alphabet) and the
  following operations are supported:
  adding input values to registers, merging or clearing registers, and testing
  whether a~register contains a~value.
  We show that a~large class of \emph{register automata} can be transformed
  into \emph{deterministic} $\rsa$s, which can serve as a~basis for
  fast matching of a~family of regexes with single-letter
  capture groups and backreferences.
  We also give a~derivative-based algorithm for transforming a~large class of
  regexes with backreferences to register automata and show that the time complexity of
  matching is linear and quadratic to the length of the input for finite and
  infinite alphabets respectively.
  Our prototype implementation of a~regex matcher shows that our approach can
  significantly improve the robustness of state-of-the-art regex matchers on
  regexes with backreferences.
  We also study the theoretical properties of the model and
  show that the emptiness problem for $\rsa$s is decidable and complete for the
  $\fomega$ class and that
  $\rsa$s are incomparable in expressive power to other popular automata models over data
  words.
\end{abstract}

\maketitle

\vspace{-0.0mm}
\section{Introduction}\label{sec:label}
\vspace{-0.0mm}

\emph{Regular expression} (regex) matching is a task performed routinely in
computer programs, such as in searching, data validation, parsing,
finding and replacing, data leak detection, or syntax highlighting.
Studies show that 30--40\,\% of Java, JavaScript, and Python software use regex
matching~\cite{Davis19}.
In many of these applications, performance of regex matching is crucial, such
as in various online services, where underperforming regex matching over user
input can cause the server to become unresponsive and unusable.
If such a~situation is caused intentionally by a~malicious user, we talk about
the so-called \emph{regular expression denial of service} (ReDoS)
attack~\cite{OwaspReDoS}.
ReDoS is a~real-world threat; it caused, for instance,
the 2016 outage of StackOverflow~\cite{stackoutage} or rendered
vulnerable websites that were using the Express.js framework~\cite{expressjsoutage}.

While there are matchers that have a~reasonably robust performance on basic
regexes---such as \grep~\cite{grep}, \retwo~\cite{re2}, or
\hyperscan~\cite{WangHCPLHZ19}---their performance on \emph{extended regexes}
can quickly deteriorate~\cite{TuronovaHHLVV22} (or they do not support the particular extensions at
all~\cite{re2-add-backref,hyperscan-unsupported-backref}).
Examples of such extensions are, e.g., bounded
counters~\cite{TuronovaHLSVV20,TuronovaHHLVV22}, lookaheads and lookbehinds,
or capture groups and backreferences, which are the topic of this paper.
%
%
The performance degrades mainly due to the inability of matchers to use the fastest automata-based algorithms that have constant or at most linear per-character cost.
These algorithms are based on the automata determinization,%
\footnote{In
practice, the matchers are based on Thompson's algorithm~\cite{Thompson1968},
which does not build the (possibly prohibitively large) DFA
\emph{a~priori} but, instead, \emph{on the fly}, while using cache to store already
constructed parts of the DFA.
}
which cannot easily accommodate the syntactic features of extended regexes
(the language of a~regex with backreferences may even not be regular).
\begin{changebar}
The matchers then usually fall back to \emph{backtracking}, which is much
slower on average and exponential at worst. 
\end{changebar}

\emph{Pattern matching using regular expressions with
backreferences} is performed ubiquitously, e.g., in validating user inputs on web
pages, processing text using the \texttt{grep} and \texttt{sed} tools,
transforming XML documents, or detecting network
incidents~\cite{Snort,BecchiC2008}.
Consider, for instance, the (extended) regex
\begin{center}
$\regexex = {}$
\texttt{/(.).*;.*(.).*;.*(.).*\textbackslash 3\textbackslash 2\textbackslash
1/},
\end{center}
where we use
the wildcard~\texttt{"."} to denote any symbol except the semicolon~\texttt{";"}
(i.e., it stands for the character class \texttt{"[\textasciicircum;]"}),
parentheses \texttt{"(}\ldots\texttt{)"} to denote the so-called
\emph{capture groups}, and \texttt{"\textbackslash x"} to denote
a~\emph{backreference} to the string captured by the \texttt{x}'th capture
group
\begin{changebar}
(the semicolon~\texttt{";"} serves in the regexs as a~delimiter).
\end{changebar}
Intuitively, the regex matches input strings~$w$ that can be seen as
a~concatenation of three strings $w = uvz$ such that~$u$ has the structure
$u = u_1\,\mathtt{;}\,u_2\,\mathtt{;}\,u_3$ with $u_1, u_2, u_3 \in (\Sigma
\setminus \{\mathtt{;}\})^+$, $v$ is a~string of three characters $a_3 a_2 a_1$
such that $a_i \in u_i$ for $i \in \{1,2,3\}$, and $z \in (\Sigma \setminus
\{\mathtt{;}\})^*$ (such a~regex can describe, e.g., a~simplified version of the
match rule of some XML transformation).
Trying (unsuccessfully) to find a~match of the regex in the randomly generated 42-character-long string
\begin{center}
\texttt{"ah;jk2367ash;la5akv45lwkjb9f.dj5fqkbxsfyrf"}
\end{center}
using the state-of-the-art PCRE2 regex matcher available at
\href{https://regex101.com/r/6ZwM2t/1/}{regex101.com}~\cite{regex101} takes
10,374 steps before reporting \emph{no match}.\footnote{%
Finding a~match of the regex
\texttt{/(.).*(.).*(.).*\textbackslash 3\textbackslash
2\textbackslash 1/} in the same text took 169,379 steps before \emph{no match} was
reported.
This regex is more challenging than $\regexex$---it does not use delimiters (semicolons in
$\regexex$).\label{ftn:hard_regex}}
Ideally, the matcher should take only 42 steps, one step for every
character in the string.
This $\frac{10,374}{42} = 247\times$ slowdown is caused by the so-called
\emph{catastrophic backtracking}---the PCRE2 matcher is based on backtracking
and, since the regex is nondeterministic, the backtracking algorithm needs to
try all possibilities of placing the three capture groups before concluding that
there is no match.
From a~theoretical point of view, this inadequate performance is hardly
a~surprise, since matching of regexes with backreferences is an \np-hard
problem~\cite{Aho90}.
This theoretical obstacle, however, does not need to be a~show-stopper, as for
many regexes appearing in practice, there is still hope that a~matching
algorithm with a~time complexity linear (or at most quadratic) to the size of
the input is possible.

As mentioned above, practical means of avoiding backtracking in such cases are not available 
since the automata determinization does not support backreferences.
Indeed, neither of the two currently most advanced regex
matchers in the industry, \retwo~\cite{re2} and \hyperscan~\cite{WangHCPLHZ19},
supports backreferences, due to a~missing efficient algorithm~\cite{hyperscan-unsupported-backref,re2-add-backref}.
The main obstacle is a~lack of a suitable deterministic automata model with a fast membership test that would support backreferences.

In this paper, we develop such a~formal model, particularly suited for fast matching of a~class of
these regexes where capture groups are single-letter (the single-letter backreferences constitute a~large portion of backreferences in our dataset collected for real-world applications).
The formal model are
\emph{register set
automata} (RSAs), 
an extension of \emph{register
automata}~\cite{KaminskiF94,DemriL09} (RAs) where registers can contain
sets of symbols instead of just single symbols (as for RAs).
Deterministic $\rsa$s ($\drsa$s) can be simulated in an efficient and robust regex-matching algorithm in time linear to the length of the input text, for finite alphabets, or
quadratic to the size of the input for infinite alphabets.
Our key result is a~\emph{RA to $\drsa$ determinization \mbox{(semi-)}algorithm}; we also propose a~partial-derivative-based algorithm for compiling a~regex with
single-letter backreferences to an RA that
completes the matching workflow.
%
We implemented a prototype matcher based on these results and   
experimented with a~sample of regexes with single-letter backreferences
extracted from a comprehensive real-world benchmark set. 
Although our determinisation is not complete (it may fail), 
the experiments show that it is quite reliable, 
and our matcher indeed significantly improves predictability of
matching by state-of-the-art matchers and reduces the danger of ReDoS.

\vspace{-2mm}


\paragraph{Theory of Register Set Automata.}
We also deliver a number of positive theoretical results related to our new
$\rsa$ automata model, which place it into the landscape of automata over infinite
alphabets.
First, the model strictly generalizes RAs~\cite{KaminskiF94,DemriL09} and is incomparable to
(one-way) alternating RAs, another popular powerful generalization of RAs~\cite{DemriL09}.
Besides the aforementioned RA to $\rsa$ determinization, 
the core property of $\rsa$s is that their emptiness problem is decidable,
although for a~higher price than for RAs---the complexity blows from
$\pspace$-complete for RAs to $\fomega$-complete (i.e., Ackermannian) for
$\rsa$s.

Importantly, since we can now determinize RAs to RSAs and test emptiness of RSAs, we can use them to test language inclusion of RAs by the standard approach:
by determinizing, complementing, and testing emptiness of the intersection
(complementing a~$\drsa$ is done easily by swapping final and non-final states).
Testing RA inclusion is an essential problem in many of their applications, such as
in their minimization,
learning~\cite{BolligHLM13,IsbernerHS2014,GarhewalVHSLS2020},
checking for fixpoint in regular model checking~\cite{ChenHLT20},
checking XML schema subsumption~\cite{TovaDV00},
verification of parameterized concurrent programs with shared
memory~\cite{IosifX19}, and is an essential component of the RA toolkit.
Unfortunately, inclusion of RAs is in general an undecidable problem~\cite{NevenSV04}, which
has forced researchers to either find ways to approximate the inclusion
test---e.g., via several membership
tests~\cite{IsbernerHS2014,GarhewalVHSLS2020} or by an abstraction refinement
semi-algorithm using interpolation~\cite{IosifX19}---or restrict themselves to
other models with decidable inclusion problem, such as deterministic RAs (whose
expressive power is quite limited)~\cite{MurawskiRT2018}
or session automata~\cite{BolligHLM13}.
We note that inclusion of languages described by regexes with backreferences is
also undecidable~\cite{Freydenberger13}.

\paragraph{Contribution.}
Let us summarise the main contributions of the paper:
%

\begin{enumerate}
  \item  Introduction of the model of register set automata and its theoretical analysis: 
    Showing its closure properties, placing the $\rsa$ model into the landscape of automata-with-registers
    models, and proving $\fomega$-completeness of the emptiness problem for $\rsa$s by
    showing interreducibility with the coverability problem for transfer Petri
    nets, discussing the power of several extensions of the model.
  \item  Designing a~(semi-)algorithm that can determinize an~RA into
    a~$\drsa$ and showing that it is complete for the class of languages that can
    be obtained from RAs with one register and no disequality tests by
    Boolean operations (i.e., union, intersection, and complement).
  \item  Developing a~partial-derivative based algorithm for converting a~large
    class of regexes with backreferences to RAs.
  \item  Showing that DRSA-based regex matching can be done in linear or
    quadratic time, depending on whether the size of the alphabet is finite or
    infinite respectively. 
  \item Experimentally confirming that simulation of the deterministic RSA obtained by our algorithm is a practical matching algorithm for regexes with single-letter backreferences, and that it is significantly more resilient against ReDoS than state-of-the-art matchers.
\end{enumerate}

\vspace{-0.0mm}
\section{Preliminaries}\label{sec:prelims}
\vspace{-0.0mm}

\paragraph{Sets and Functions.}
We use~$\nat$ to denote the set of natural numbers without~0, $\natzero$ to denote
$\nat \cup \{0\}$, and $\initof n$ for $n \in \nat$ to denote the set $\set{1,
\ldots, n}$ (we note that $\initof 0 = \emptyset$).
We sometimes use~``$\cdot$'' to denote an \emph{ellipsis}, i.e., a~value that
can be ignored.
\begin{changebar}
For a~(partial) function $f\colon A \partialto B$, a~set $C \subseteq B$, and
an element $d \in B$,
we use $f\subst{C}{d}$ to denote the function $f\subst{C}{d}(x) = d$ if
$f(x) \in C$ and $f\subst{C}{d}(x) = f(x)$ otherwise.
\end{changebar}
Furthermore, for a set $X\subseteq A$, we use $f_{|X}$ to denote the
restriction of $f$ to $X$ defined as $f_{|X} = f \cap (X\times B)$.

\paragraph{Data Words.}
Let us fix a~finite nonempty \emph{alphabet}~$\Sigma$ and an infinite
\emph{data domain}~$\datadom$. 
A~(finite) \emph{data word} of \emph{length}~$n$ is a~function $\word \colon \initof
n \to (\Sigma \times \datadom)$; we use $\lenof \word = n$ to denote its length
and $\wordof 1, \ldots, \wordof n$ to denote its symbols.
The \emph{empty word} of length~0 is denoted~$\epsilon$.
We use $\projsigmaof \word$ and $\projdataof \word$ to denote the
\emph{projection} of~$\word$ onto the respective domain (e.g., if $\word =
\pair a 1 \pair b 2 \pair b 3$, then $\projsigmaof \word = abc$ and
$\projdataof \word = 123$) and, given $a \in \Sigma$, we use $\projof a {\wordof i}$
as a~shortcut for $\projsigmaof{\wordof i} = a$.

\paragraph{Register Automata on Data Words
\begin{changebar}
\cite{KaminskiF94,DemriL09}.
\end{changebar}
}
A~(nondeterministic one-way) \emph{register automaton} (on data words),
abbreviated as (N)RA, is a~tuple $\aut = (Q, \regs, \Delta, I, F)$ where
$Q$~is a~finite set of \emph{states},
$\regs$~is a~finite set of \emph{registers},
$I \subseteq Q$ is a~set of \emph{initial states},
$F \subseteq Q$ is a~set of \emph{final states}, and
$\Delta \subseteq Q \times \Sigma \times 2^\regs \times 2^\regs \times (\regs \to \regs \cup
\set{\inp, \bot}) \times Q$ is a~\emph{transition relation} such that if
$\tau \colon (q, a, \guardeq, \guardneq, \update, s) \in \Delta$, then $\guardeq \cap \guardneq
= \emptyset$.
We use $\trans q a {\guardeq, \guardneq, \update}{s}$ to denote~$\tau$ (and often
drop from $\update$ mappings $r \mapsto r$ for $r \in \regs$, which we treat as
implicit).
The semantics of~$\tau$ is that~$\aut$ can move from state~$q$ to state~$s$ if the
$\Sigma$-symbol at the current position of the input word is~$a$ and the
$\datadom$-value at the current position is equal to all registers
from~$\guardeq$ and not equal to any register from~$\guardneq$; the content of
the registers is updated so that $r_i \gets \update(r_i)$ (i.e., $r_i$~can be
assigned the value of some other register, the current $\datadom$-symbol,
denoted by~$\inp$, or it can be cleared by being assigned~$\bot$).

A~\emph{configuration} of~$\aut$ is a~pair $c \in Q \times (\regs \to
\datadom \cup \set{\bot})$, i.e., it consists of a~state and an~assignment of data
to registers.
\begin{changebar}
An \emph{initial configuration} of~$\aut$ is a~pair $c_{\init} = (q_{\init},
\{r \mapsto \bot \mid r \in \regs\})$ with $q_{\init} \in I$.
\end{changebar}
Suppose $c_1 = (q_1, f_1)$ and $c_2 = (q_2, f_2)$ are two configurations of~$\aut$.
We say that~$c_1$ can make a~\emph{step} to~$c_2$ over $\pair a d \in
\Sigma\times \datadom$ using transition $\tau \colon \trans q a {\guardeq, \guardneq,
\update} s \in \Delta$, \mbox{denoted as
$c_1 \stepofusing{\pair a d} \tau c_2$, iff}
\begin{enumerate}
  \item  $d = f_1(r)$ for all $r \in \guardeq$,
  \item  $d \neq f_1(r)$ for all $r \in \guardneq$, and
    \vspace{-6mm}
  \item  for all $r \in \regs$, we have
$        f_2(r) = \begin{cases}
          f_1(r') & \text{if } \update(r) = r' \in \regs,\\
          d & \text{if } \update(r) = \inp, \text{ and} \\
          \bot & \text{if } \update(r) = \bot .
        \end{cases}$
\end{enumerate}
%
A~\emph{run}~$\rho$ of~$\aut$ over the word $w = \pair{a_1}{d_1} \ldots
\pair{a_n}{d_n}$ from a~configuration~$c$ is a~sequence of
alternating configurations and transitions $\rho = c_0 \tau_1 c_1 \tau_2 \ldots \tau_n c_n$ such that $\forall 1 \leq i \leq
n \colon c_{i-1} \stepofusing{\pair{a_i}{d_i}}{\tau_i}~c_i$ and $c_0~=~c$.
We say that~$\rho$ is \emph{accepting} if~$c$ is an initial configuration, $c_n =
(s,f)$, and $s~\in~F$.
The \emph{language} $\langof \aut$ accepted by $\aut$ is defined as
$\langof \aut = \set{w \in (\Sigma \times \datadom)^* \mid \aut \text{ has an
accepting run over } w}$.

%

We say that~$\aut$ is a~\emph{deterministic RA} ($\dra$) if for all states $q
\in Q$ and all $a \in \Sigma$, it holds that for any two distinct transitions
$\trans q a {\guardeq_1, \guardneq_1, \update_1} s_1,
\trans q a {\guardeq_2, \guardneq_2, \update_2} s_2 \in \Delta$
we have that $\guardeq_1 \cap \guardneq_2 \neq \emptyset$ or $\guardeq_2 \cap
\guardneq_1 \neq \emptyset$.
$\aut$ is \emph{complete} if for all states $q \in Q$, symbols $a \in \Sigma$,
and $\guard \subseteq \regs$, there is a~transition $\trans q a {\guardeq,
\guardneq, \update} s$ such that $\guardeq \subseteq \guard$ and $\guard
\cap \guardneq = \emptyset$.
%

\paragraph{Universal RAs.}
A~\emph{universal RA} ($\ura$) $\aut_U$ is defined exactly as an $\nra$ with
the exception of its language.
The language of~$\aut_U$ is the set $\langof{\aut_U} = \{w \in (\Sigma
\times \datadom)^* \mid $ every run of $\aut_U$ on~$w$ is
accepting$\}$ (we emphasize that if a~run cannot continue from some state over
the current input symbol, then it is not accepting).
\begin{changebar}
The concept of \emph{universality} in the name is linked to the use of universal
quantification over runs in deciding acceptance.
\end{changebar}

\medskip

There is indeed duality between $\nra$s and $\ura$s, as stated by the following
fact.

\begin{fact}\label{lem:complement_ura_nra}
For every $\nra$ $\aut_N$, there is a~$\ura$ accepting the complement
of~$\langof{\aut_N}$.\linebreak
Conversely, for every $\ura$ $\aut_U$, there is an~$\nra$ accepting the
  complement of~$\langof{\aut_U}$.
\end{fact}

\begin{proof}
For both parts, the complement automaton is obtained by
\begin{inparaenum}[(i)]
  \item  adding a~rejecting \emph{sink} state to the automaton,
  \item  completing the transition relation (i.e., adding transitions for
    undefined combinations of symbols and guards to the sink state) of the input automaton, and
  \item  swapping final and non-final states.
\end{inparaenum}
  \qedhere
\end{proof}

\begin{example}\label{ex:ra_example}
Consider the language of words over $\Sigma= \set{a}$ that contain two
occurrences of some data value, i.e., the language
\begin{equation*}
\hlmathbox{
\langexrep = \set{\word \mid \exists i,j\colon i \neq j \land
\projdataof{\wordof i} = \projdataof{\wordof j}}.
}
\end{equation*}

\noindent
An $\nra$ recognising this language is in the following figure:
%
\begin{center}
\input{figs/ra_example.tikz}
\end{center}
Formally, it is an $\nra$
$\aut = (\set{q,s,t}, \set{r}, \Delta, \set{q}, \set{t})$
where the transition relation is defined as $\Delta = \{
  \trans q a {\emptyset, \emptyset, \emptyset} q$,
  $\trans q a {\emptyset, \emptyset, \{r \mapsto \inp\}} s$,
  $\trans s a {\emptyset, \{r\}, \emptyset} s$,
  $\trans s a {\{r\}, \emptyset, \emptyset} t$,
  $\trans t a {\emptyset, \emptyset, \emptyset} t
  \}$ (actually, the guard $\inp \neq r$ on the self-loop over~$s$ is redundant).
  Recall that~$\emptyset$ for an update denotes the mapping~$\{r \mapsto r\}$.
We note that~$\langexrep$ is not expressible by any $\dra$ or $\ura$.\footnote{
\cbstart
This can be shown by contradiction, assuming that there is a~$\dra$ with~$k$
registers accepting~$\langexrep$.  One can then take a~run over some rejected word of the length $2k+2$, cut it in half, and modify the second half of the word to contain some symbol that is in the first half of the word but is not stored any register.
Because of determinism, the $\drsa$ should also reject the modified word, but then the accepted language is not~$\langexrep$. Contradiction.
\cbend
}
Intuitively, $\aut$~waits in~$q$ until it nondeterministically guesses the input
data value that should be repeated, stores it into register~$r$, and moves to
state~$s$.
In state~$s$, it is waiting to see the data value again, upon which it moves to
the accepting state~$t$ and reads out the rest of the word.

On the other hand, the complement of the language, i.e., the language of words
where no two positions have the same data value, formally,
\begin{equation*}
\hlmathbox{
\langnegexrep = \set{\word \mid \forall i,j\colon i \neq j \implies
\projdataof{\wordof i} \neq \projdataof{\wordof j}},
}
\end{equation*}

\noindent
\cbstart
is not expressible by any $\dra$ or $\nra$~\cite[Proposition~5 and its proof]{KaminskiF94}, but is expressible by a~$\ura$.
\cbend
The $\ura$ accepting $\langnegexrep$ looks similar to the $\nra$~$\aut$ above with the
exception of final states, which are $\set{q,s}$ (note that in $\ura$s, in order
to accept a~word, all runs over the word need to accept, so in order to
accept in this example, all runs of the $\ura$ need to avoid the state~$t$).
\qed
\end{example}

\medskip

We use $\nraeq$, $\uraeq$, and $\draeq$ to denote the sub-classes of $\nra$s,
$\ura$s, and $\dra$s
with \emph{no disequality} guards, i.e., automata where for every transition
$\trans q a {\guardeq, \guardneq, \update} s$ it holds that $\guardneq =
\emptyset$.
Furthermore,
for a~class~$\mathcal C$ of automata with registers and $n \in \nat$, we
use~$\mathcal{C}_n$ to denote the sub-class of~$\mathcal{C}$ containing automata
with at most~$n$ registers (e.g., $\draof 2$).
We abuse notation and use~$\mathcal C$ to also denote the class of
languages defined by~$\mathcal C$.

%
%
%
%
%
%
%

\vspace{-0.0mm}
\section{Register Set Automata}\label{sec:rsa}
\vspace{-0.0mm}

\begin{changebar}
On this section, we introduce the model of register set automata, which differ
from~RAs mainly in the ability to store into registers \emph{sets of values}
instead of just single data elements.
\end{changebar}
A~(nondeterministic) \emph{register set automaton} (on data words),
abbreviated as (N)$\rsa$ is a~tuple $\aut_S = (Q, \regs, \Delta, I, F)$ where
$Q, \regs, I, F$ are the same as for RAs and
$\Delta \subseteq Q \times \Sigma \times 2^\regs \times 2^\regs \times (\regs \to 2^{\regs \cup \set{\inp}}) \times Q$
such that
if $\trans q a {\guardin, \guardnotin, \update}{s} \in \Delta$, then
$\guardin \cap \guardnotin = \emptyset$ (as with $\nra$s, we often do not write
mappings $r \mapsto \{r\}$ for $r \in \regs$ when defining~$\update$).
The semantics of a~transition $\trans q a {\guardin, \guardnotin, \update} s$ is
that~$\aut_S$ can move from state~$q$ to state~$s$ if the $\Sigma$-symbol at
the current position of the input word is~$a$ and the $\datadom$-value at the
current position is in all registers from~$\guardin$ and in no register
from~$\guardnotin$; the content of the
registers is updated so that $r_i \gets \bigcup \set{x \mid x \in \update(r_i)}$
(i.e., $r_i$~can be assigned the union of values of several registers, possibly
including the current $\datadom$-symbol denoted by~$\inp$).


A~\emph{configuration} of~$\aut_S$ is a~pair $c \in Q \times (\regs \to
2^\datadom)$, i.e., it consists of a~state and an~assignment of sets of data values
to registers.
\begin{changebar}
An \emph{initial configuration} of~$\aut_S$ is a~pair $c_{\init} =
(q_{\init},\set{r \mapsto \emptyset \mid r \in \regs})$ with $q_\init \in I$.
\end{changebar}
Let $c_1 = (q_1, f_1)$ and $c_2 = (q_2, f_2)$ be two configurations of~$\aut_S$.
We~say that~$c_1$ can make a~\emph{step} to~$c_2$ over $\pair a d \in \Sigma
\times \datadom$ using transition $t\colon \trans q a {\guardin, \guardnotin,
\update} s \in \Delta$, \mbox{denoted as
$c_1 \stepofusing{\pair a d} t c_2$, iff}
\begin{enumerate}
  \item  $d \in f_1(r)$ for all $r \in \guardin$,
  \item  $d \notin f_1(r)$ for all $r \in \guardnotin$, and
    \vspace{-3mm}
  \item  for all $r \in \regs$, we have
      $f_2(r) = \bigcup\set{f_1(r') \mid r' \in \regs, r' \in \update(r)} \cup
      \begin{cases}
        \set{d} & \text{if } \inp \in \update(r) \text{ and} \\
        \emptyset & \text{otherwise}.
      \end{cases}$
\end{enumerate}
%
%
The definition of a~run and language of~$\aut_S$ is then the same as
for~$\nra$s.


We say that the $\rsa$~$\aut_S$ is \emph{deterministic} ($\drsa$) if for all states $q
\in Q$ and all $a \in \Sigma$, it holds that for any two distinct transitions
$\trans q a {\guardin_1, \guardnotin_1, \update_1} s_1,
\trans q a {\guardin_2, \guardnotin_2, \update_2} s_2 \in \Delta$
we have that $\guardin_1 \cap \guardnotin_2 \neq \emptyset$ or $\guardin_2 \cap
\guardnotin_1 \neq \emptyset$.

\newcommand{\figRsaExamples}[0]{
\begin{figure}[t]
\begin{minipage}[b]{4.3cm}
\hspace*{-5mm}
\input{figs/ex_rep.tikz}
\vspace{2.45mm}
\caption{$\drsaof 1$ for $\langexrep$}
\label{fig:drsa_langexrep}
\end{minipage}
\begin{minipage}[b]{4.9cm}
\centering
\input{figs/neg_ex_rep.tikz}
\vspace{2.65mm}
\caption{$\drsaof 1$ for $\langnegexrep$}
\label{fig:drsa_langnegexrep}
\end{minipage}
\begin{minipage}[b]{4.3cm}
\hspace*{-7mm}
\input{figs/neg_all_rep.tikz}
\vspace{-8mm}
\caption{$\rsaof 1$ for $\langnegallrep$}
\label{fig:rsa_langnegallrep}
\end{minipage}
\vspace*{-3mm}
\end{figure}
}

\begin{figure}[t]
\begin{minipage}[b]{4.3cm}
\hspace*{-5mm}
\input{figs/ex_rep.tikz}
\vspace{2.45mm}
\caption{$\drsaof 1$ for $\langexrep$}
\label{fig:drsa_langexrep}
\end{minipage}
\begin{minipage}[b]{4.9cm}
\centering
\input{figs/neg_ex_rep.tikz}
\vspace{2.65mm}
\caption{$\drsaof 1$ for $\langnegexrep$}
\label{fig:drsa_langnegexrep}
\end{minipage}
\begin{minipage}[b]{4.3cm}
\hspace*{-7mm}
\input{figs/neg_all_rep.tikz}
\vspace{-8mm}
\caption{$\rsaof 1$ for $\langnegallrep$}
\label{fig:rsa_langnegallrep}
\end{minipage}
\vspace*{-3mm}
\end{figure}

\begin{example}\label{ex:drsa_langexrep}
A~$\drsa$ accepting the language~$\langexrep$ from \cref{ex:ra_example}
is in \cref{fig:drsa_langexrep}.
%
%
Formally, it is a~$\drsaof 1$
$\aut = (\set{q,s}, \set{r}, \Delta, \set{q}, \set{s})$ where
$\Delta = \{
  \trans q a {\emptyset, \set r, \set{r \mapsto \set{r,\inp}}} q$,
  $\trans q a {\set{r}, \emptyset, \emptyset} s$,
  $\trans s a {\emptyset, \emptyset, \emptyset} s\}$.
Intuitively, the $\drsa$ waits in~$q$ and accumulates the so-far seen
input data values in register~$r$ (we use $r \gets r \cup \set \inp$ to denote
the update $r \mapsto \set{r,\inp}$).
Once the $\drsa$ reads a~value that is already in~$r$, it moves to~$s$ and
accepts.
\qed
\end{example}

\begin{example}\label{ex:drsa_langnegexrep}
A~$\drsaof 1$ accepting the language~$\langnegexrep$ from \cref{ex:ra_example} is in
\cref{fig:drsa_langnegexrep}.
%
Intuitively, the automaton stays in state~$q$ and accumulates input data values
in register~$r$, making sure the input data value has not been seen
previously.
\qed
\end{example}

\begin{example}\label{ex:lang_negallrep}
  Consider the following language:
  \begin{equation*}
  \hlmathbox{
  \langnegallrep = \set{\word \mid \exists i\forall j\colon i\neq j \implies
  \projdataof{\wordof i} \neq \projdataof{\wordof j}}
}
  \end{equation*}

  \noindent
  Intuitively, $\langnegallrep$~is the language of all words containing
  a~data value with exactly one occurrence.
  This language is accepted, e.g., by the $\rsaof 1$ in
  \cref{fig:rsa_langnegallrep}.
  %
  %
  The $\rsa$ stays in state~$q$, collecting the seen values into its
  register, and at some point, when it sees a~value not seen previously,
  it nondeterministically moves to~$s$, remembering the value in its
  register.
  Then, at state~$s$, the $\rsa$ just checks that it does not see the remembered
  value any more.
  \begin{changebar}
  We note that~$\langnegallrep$ cannot be accepted by a~$\drsa$
  (cf.~the proof of \cref{thm:drsa_lt_rsa}).
  \end{changebar}
%
\qed
\end{example}

\vspace{-0.0mm}
\section{Properties of Register Set Automata}\label{sec:properties}
\vspace{-0.0mm}

In this section, we establish decidability and complexity of basic decision problems
for~$\rsa$s and their closure properties.
First, we claim that $\rsa$s generalise $\nra$s.

\begin{restatable}{fact}{lemNraToRsa}\label{thm:nra-to-rsa}
For every $n \in \nat$ and $\nraof n$, there exists an $\rsaof n$ accepting the
  same language.
\end{restatable}

The next theorem shows the core property of $\rsa$s: that their emptiness
problem is decidable, however, for
a~much higher price than for~$\nra$s, for which it is
$\pspace$-complete\footnote{%
Note that for an alternative definition of $\nra$s considered
in~\cite{KaminskiF94,SakamotoI00}, where no two registers can contain the same data
value, the problem is $\np$-complete~\cite{SakamotoI00}.}~\cite{DemriL09}.
For classifying the complexity of the problem, we use the hierarchy of
fast-growing complexity classes of Schmitz~\cite{Schmitz16a}, in particular the
class $\fomega$, which, intuitively, corresponds to Ackermannian problems closed
under primitive-recursive reductions.

\begin{restatable}{theorem}{thmRsaEmptiness}
\label{thm:rsa-emptiness}
  The emptiness problem for $\rsa$ is decidable,
  in particular, $\fomega$-complete.
\end{restatable}


\begin{proof}[Sketch of proof.]
The proof is done by showing interreducibility of $\rsa$ emptiness with
coverability in \emph{transfer Petri nets} (TPNs) (often used for modelling the
so-called \emph{broadcast protocols}), which is a~known $\fomega$-complete
problem~\cite{SchmitzS13,Schmitz17,SchmitzS12}.
In the following, we briefly describe both directions
of the reduction (see \ifTR\cref{sec:proof-rsa-emptiness} \else\cite{techrep} \fi for details and examples).

\begin{enumerate}
    \item[($\rsa$ emptiness $\leq$ TPN coverability)]
    Intuitively, the conversion of an $\rsa$ $\aut$ $ = (Q, \regs, \Delta, I,
    F)$ into a~TPN $\netof \aut$ is done in the following way.
    The set of places of $\netof \aut$ will be as follows:
    \begin{inparaenum}[(i)]
      \item  one place for each state of~$\aut$,
      \item  two special places $\init$ and $\fin$, and
      \item  one place for every subset $\region \subseteq \regs$; these places
        are used to represent all possible intersections of values held in
        registers.
        E.g., if there are four tokens in the place representing $r_1
        \cap r_2$, it means that there are exactly four different data values
        stored in both~$r_1$ and~$r_2$ and in no other register.
    \end{inparaenum}
    Each transition~$t$ of~$\aut$ is simulated by one or more TPN transitions between places
    representing its source and target states.
    The number of respective TPN transitions
    depends on how specific the guard is in the original automaton, since we need to
    distinguish every possible option of~$\inp$ being in some region $\region \in 2^\regs$.
    The transitions move the token between the places corresponding to $t$'s
    source and target states and, moreover, use the \emph{broadcast} arcs to
    move tokens between the places representing regions, according to the
    manipulation of the set-registers in the update function of~$t$.
    The special place $\init$ is used to have a~single starting marking
    (it nondeterministically chooses one state from~$I$) and the place
    $\fin$ is used as the coverability test target; all places
    \mbox{
    corresponding to final states of~$\aut$ can simply transition into it.}

%
%
%
%

    \item[(TPN coverability $\leq$ $\rsa$ emptiness)]

    Given a~TPN~$\net$, the $\rsa$ $\autof \net$
    simulating it will have the following structure.
    There will be a~state~$\qmain$, which will be active before and
    after the simulation of firing each transition of~$\net$.
    Moreover, there will be one register for every place of~$\net$;
    individual tokens in the places will be simulated by unique data values
    from~$\datadom$ stored in the corresponding registers.
    For each transition of~$\net$, there will be a~\emph{gadget}, doing a~cycle
    on~$\qmain$, that represents the semantics of~$\net$'s transition.
    Each such gadget is composed of several \emph{protogadgets}, which simulate
    basic actions performed during the transition (adding a~token to a~place,
    removing a~token, moving all tokens between places).
    Implementation of adding a~token and moving tokens is relatively easy,
    the tricky part is removing a~token, since~$\rsa$s do not support removing
    a~data value from a~register.
    We solve this by using a~\emph{lossy remove}: i.e., if \emph{one} token is
    to be removed from a~place, we simulate it by removing \emph{at least one}
    token (but potentially more).
    This will not preserve \emph{reachability}, but it is enough to preserve
    \emph{coverability}.
    Moreover, there will also be an \emph{initial} part setting the contents of the
    registers to reflect the initial marking of~$\net$ (terminating in~$\qmain$)
    and a~\emph{final} part that checks the coverability by removing (again in
    a~lossy way) tokens from places, terminating in a~single final state.
    \qedhere

    %
    %
    %
    %
    %
    %
\end{enumerate}

\end{proof}


\begin{remark}
Since $\rsa$s generalise $\nra$s, their universality, equivalence, and language
inclusion problems are all undecidable.
\end{remark}

\vspace{-0.0mm}
\subsection{Closure Properties}\label{sec:closure-properties}
\vspace{-0.0mm}

The closure properties of $\rsa$s under Boolean operations are the same as for
$\nra$s (cf.\ \cite[Proposition~5, Theorem~3]{KaminskiF94}).

\begin{restatable}{theorem}{thmRsaClosure}\label{thm:rsa-closure}
The following closure properties hold for the class $\rsa$:
\begin{enumerate}
  \item  $\rsa$ is closed under union and intersection.
  \item  $\rsa$ is not closed under complement.
\end{enumerate}
\end{restatable}

\begin{proof}[Sketch of proof.]
  The proofs for closure under union and intersection are standard.
  For showing the non-closure under complement, consider the
  language~$\langnegallrep$ from \cref{ex:lang_negallrep}, which can be accepted by
  $\rsa$.
  We use a~similar technique as in the proof of Proposition~3.2
  in~\cite{Figueira12} and show that if there were an $\rsa$ accepting its
  complement, namely, the language
  \begin{equation*}
    \hlmathbox{
    \langallrep = \set{\word \mid \forall i\exists j\colon i\neq j \land
    \projdataof{\wordof i} = \projdataof{\wordof j}},
  }
  \end{equation*}

  \noindent
  $\rsa$s could be used to decide the emptiness problem of a~Minsky machine,
  which is a~known undecidable problem.
  See the full proof in \ifTR\cref{sec:proofs_closure_properties} \else \cite{techrep} \fi for more details.
\end{proof}

For $\rsa$s with a~limited number of registers, we lose the closure under
intersection.

\begin{restatable}{theorem}{thmRsanClosure}\label{thm:rsan-closure}
  For each $n \in \nat$, the class $\rsaof n$
    is closed under union and
    not closed under intersection and complement.
\end{restatable}

\begin{proof}[Sketch of proof.]
The proof of closure of $\rsaof n$ under union is standard.
For proving non-closure of $\rsaof 1$ under intersection, we consider the two
following $\rsaof 1$ languages

\begin{equation}
  \lang^A_1 = \set{w \mid \projdataof{w_1} =
  \projdataof{w_{|w|}}}
  \qquad
  \text{and}
  \qquad
  \lang^B_1 = \set{w \mid \projdataof{w_2} =
  \projdataof{w_{|w| - 1}}}
\end{equation}

\noindent
  and show that their intersection cannot be accepted by~$\rsaof 1$.
  This argument can be extended to $\rsaof n$ for $n > 1$.
  Non-closure of $\rsaof n$ under complement then follows from De Morgan's laws.
\end{proof}
\vspace{-2mm}

\begin{restatable}{theorem}{thmDrsaClosure}\label{thm:drsa-closure}
$\drsa$ is closed under union, intersection, and complement.
\end{restatable}

\begin{proof}
The proofs are standard (product construction and swapping (non)-final states).
\end{proof}
\vspace{-2mm}

\begin{restatable}{theorem}{thmDrsanClosure}\label{thm:drsan-closure}
  For each $n \in \nat$, the class $\drsaof n$ is closed under complement and
    not closed under union and intersection.
\end{restatable}

\begin{proof}
Proof of closure under complement is standard.
Non-closure under intersection is done similarly as in the proof
of \cref{thm:rsan-closure}.
Non-closure for union follows from De Morgan's laws.
\end{proof}

As with RAs, nondeterminism also allows bigger expressivity for $\rsa$s.

\begin{theorem}\label{thm:drsa_lt_rsa}
$\drsa \subsetneq \rsa$
\end{theorem}

\begin{proof}
Let us consider the language~$\langnegallrep$ from the proof of
  \cref{thm:rsa-closure}, which is expressible using~$\rsa$s, and its
  complement~$\langallrep$, which is not expressible using~$\rsa$s.
Since $\drsa$s are closed under complement (\cref{thm:drsa-closure}), if they could
accept~$\langnegallrep$, they could also accept~$\langallrep$, which is
  a~contradiction.
  Therefore, $\langnegallrep \notin \drsa$.
\end{proof}

\vspace{-1.0mm}
\subsection{Expressivity}\label{sec:label}
\vspace{-0.0mm}

The RSA model captures an interesting class of data word languages, strictly
generalizing NRAs and being incomparable to ARAs~\cite{DemriL09} or pebble
automata~\cite{NevenSV04}.
Due to the page limit, see \ifTR\cref{sec:expressivity-rsa} \else \cite{techrep} \fi for detailed positioning
of RSAs in the landscape of register automata models.

\newcommand{\detalgonew}[0]{
\begin{figure}[t]
\begin{algorithm}[H]
  \SetKwProg{Fn}{Function}{:}{}
  \SetKwInOut{Input}{\hspace*{-\algomargin}Input}
  \SetKwInOut{Output}{\hspace*{-\algomargin}Output}
  \SetKwData{worklist}{worklist}
  \caption{Determinization of an $\nra$ into a~$\drsa$}
  \label{alg:det}
  \Input{Copyless $\nra$ $\aut = (Q, \regs, \Delta, I, F)$}
  \Output{\mbox{$\drsa$ $\aut' = (\mathcal{Q}', \regs, \Delta', I', F')$
    with $\langof{\aut'} = \langof{\aut}$ or $\bot$}}

  $\mathcal{Q}' \gets \worklist \gets I' \gets
    \{(I, c_0 = \{r \mapsto 0 \mid r \in \regs\})\}$\;\label{ln:init}
  $\Delta' \gets \emptyset$\;

  \While{$\worklist \neq \emptyset$}{
    $(S,c) \gets \worklist.\mathit{pop}()$\; \label{ln:pop}
    \ForEach{$a \in \Sigma, \guard \subseteq \set{r \in \regsof S \mid c(r) \neq 0}$}{ \label{ln:guards} 
      $T \gets \big\{\trans q a {\guardeq, \guardneq, \cdot}{q'}\in\Delta \mid
      q \in S, \guardeq \subseteq \guard, \guardneq \cap \guard =
      \emptyset\big\}$\;\label{ln:transitions}
      $S' \gets \big\{q' \mid \trans \cdot \cdot {\cdot, \cdot, \cdot}{q'} \in T\big\}$\;\label{ln:macrostate}
      \lIf{$\exists \trans q \cdot {\cdot, \guardneq, \cdot} {q'} \in
        T,\exists r \in \guardneq \colon c(r) > 1$}{\label{ln:cardinality}
        \Return{$\bot$}
      }
      $T^\bullet = \{\trans q a {\guardeq, \guardneq, \update \subst{\guardeq}{\inp}} {q'} \mid \trans q a {\guardeq, \guardneq, \update} {q'} \in
        T\}$\;
      \ForEach( \tcp*[f]{update the register size classes}){$r_i \in \regs$}{\label{ln:aggregate}
         $\mathit{tmp} \gets \emptyset$\;
         \ForEach{$\trans \cdot \cdot {\guardeq, \cdot, \update} \cdot \in T^\bullet$}{\label{ln:trans-iter}
           \lIf{$\update(r_i) = y \neq \bot \land c(y) \neq 0$}{
             $\mathit{tmp} \gets \mathit{tmp} \cup \set y$
           }
         }
         $\mathit{op}_{r_i} \gets \mathit{tmp}$\;\label{ln:update_collapse}
         $c'(r_i) \gets \!\!\!\!\sum\limits_{x \in \mathit{op}_{r_i}}^{>1\leadsto \omega}\!\!
           c(x)$
         \;\label{ln:counter_update}


         %

      }
      \ForEach(\tcp*[f]{check for Cartesian overapproximation}){$q' \in S'$}{\label{ln:approx_test_loop}
        $P \gets \mathit{op}_{r_1} \times \cdots \times \mathit{op}_{r_n}$
          for $\set{r_1, \ldots, r_n} = \regsof{q'}$\;\label{ln:cartesian}
        \ForEach{$(x_1, \ldots, x_n) \in \mathit{P}$}{\label{ln:approxtest}
          \lIf{$\nexists (\trans \cdot \cdot {\cdot, \cdot, \update}{q'} )
            \in T^\bullet$  s.t.\!\!\!\! $\bigwedge\limits_{1\leq i \leq n}\!\!\!\!\!
            \update(r_i) = x_i$}{
              \Return{$\bot$}\label{ln:abort}
            \vspace*{-3.0mm}
          }
        }
      }
      $\update' \gets \set{r_i \mapsto \mathit{op}_{r_i} \mid r_i \in
        \regs}$\;\label{ln:update}
      \If{$(S', c') \notin \mathcal{Q}'$}{
        $\worklist.\mathit{push}((S', c'))$\;
        $\mathcal{Q}' \gets \mathcal{Q}' \cup \set{(S', c')}$\;
      }
      $\Delta' \gets \Delta' \cup \big\{\trans {(S, c)} a {\guard, \regs \setminus
        \guard, \update'} {(S', c')}\big\}$\;\label{ln:newtrans}
    }
  }

  \Return{$\aut' =  (\mathcal{Q}', \regs, \Delta', I', \set{(S, c) \in \mathcal{Q}' \mid S \cap F \neq \emptyset})$}\label{ln:det_return}\;
\end{algorithm}
\vspace*{-5mm}
\end{figure}
}

\vspace{-0.0mm}
\section{Determinizing Register Automata}\label{sec:dra}
\vspace{-0.0mm}


$\rsa$s have the following interesting property: a~large class of $\nra$s can
\begin{changebar}
be determinized into $\drsa$s (we emphasize that the determinization considered
  here changes the model from one storing single values in registers (NRAs) to
  one storing sets of values in registers (DRSAs)).
\end{changebar}
In this section, we give a~determinization semi-algorithm and specify
properties \mbox{of a~class of $\nra$s for which it is complete.}

\detalgonew

Let $\aut = (Q, \regs, \Delta, I, F)$ be an $\nra$.
We use $\regsof q$ for $q \in Q$ to denote the set of registers~$r$ 
\begin{changebar}
\emph{active} at $q$, 
for which 
\end{changebar}
there exists a~transition $\trans s \cdot {\guardeq, \guardneq, \update} t \in
\Delta$ with
\begin{inparaenum}[(i)]
  \item  $\update(r) \neq \bot$ and $t = q$ or
  \item  $r \in \guardeq \cup \guardneq$ and $s = q$.
\end{inparaenum}
\begin{changebar}
Besides being the basis for the register locality optimization in \cref{sec:reg-locality},
the set of active registers $\regsof q$~is also used in the basic algorithm as an
overapproximation of the set of registers with a value different from $\bot$ (it excludes those register that were just assigned $\bot$). 
Given a~set of states~$S$, we define $\regsof S = \bigcup_{q \in S} \regsof q$.
\end{changebar}
%
Furthermore,
we call~$\aut$ \emph{copyless} if there is no reachable
configuration $(q, f)$ such that $f(r_1) = f(r_2) \neq \bot$ for a~pair of distinct
registers $r_1, r_2 \in \regs$, i.e., there is at most one copy of each data
value in~$\aut$.
Again, any $\nra$ can be converted into the copyless form, however, the
number of states can increase to $B_{|\regs|} \cdot |Q|$ where $B_n$ is the $n$-th
Bell number.
Intuitively, the transformation is done by creating one copy of
each state for every possible partition of~$\regs$ (the partitions contain
registers with the same value), and modifying the transition function
correspondingly.

The determinization (semi-)algorithm for a~copyless $\nra$~$\aut$ is shown
in \cref{alg:det}.
On the high level, it is similar to the classical Rabin-Scott subset
construction for determinizing finite automata~\cite{RabinS59} with additional
treatment of registers superimposed onto~it.

During the construction, we track
\begin{inparaenum}[(i)]
  \item  all states of~$\aut$ in which the runs of $\aut$ might be at a~given
    point, represented by a~set of states $S \subseteq Q$ and
  \item  a~mapping $c\colon \regs \to \set{0,1,\omega}$ assigning to each register its \emph{size class}, 
  that records whether the size of the set in the register is 0, 1, or larger than 1 (denoted by the $\omega$).
      The size classes are needed to ensure that our simulation of a~disequality
      test $\inp \neq r$ by the non-membership test $\inp \notin r$ is precise, as explained under the item (1) below. 
\end{inparaenum}
the macrostate is then a~pair~$(S, c)$.
The initial state of the constructed $\drsa$ is the macrostate~$(I, c_0)$ where
$c_0$ is a~mapping assigning zero to each register (the run of a~$\drsa$
starts with all registers initialized to~$\emptyset$) (\lnref{ln:init}).

The main loop of the algorithm then constructs successors of reachable
macrostates for each $a\in \Sigma$ and each $g \subseteq \regs$ on
\lnref{ln:guards}; each pair $a$, $g$ corresponds to the so-called
\emph{minterm} (minterms denote combinations of guards whose semantics do not
overlap~\cite{DAntoniV14}).
\begin{changebar}
The set of active registers $\regsof S$ is used here to prune those minterms that clearly cannot be satisfied,
since they are testing a register whose value must be $\bot$.   
\end{changebar}
For each minterm, we collect all transitions of~$\aut$ compatible with this
minterm (\lnref{ln:transitions}) and generate the successor set of states~$S'$
(\lnref{ln:macrostate}).
The~$\aut'$ update function~$\update'$ for register~$r$ is then set to collect
into~$r$ all possible values that might be stored into~$r$ in~$\aut$ on any run
over the input word at the given position
(Lines~\ref{ln:aggregate}--\ref{ln:update}).

\noindent
\begin{changebar}
The algorithm uses the following three techniques to handle three sources of imprecision:
\end{changebar}
\begin{enumerate}
  \item  
\begin{changebar}
\emph{Register size classes. (Lines 10 to 15).}
\end{changebar}
Since the algorithm collects in the set-register~$r$ all possible values
    that could have been stored into the standard register~$r$ in~$\aut$, if the
    disequality tests in~$\guardneq$ were changed for non-membership tests
    in~$\guardnotin$, this could mean that~$\aut'$ might not be able to simulate
    some transition of~$\aut$ (the transition would not be enabled).
    Consider the following example:

    \vspace{-3mm}
    \begin{center}
      \begin{minipage}[b]{0.4\textwidth}
      \scalebox{0.8}{
      \input{figs/nra-problem-diseq.tikz}
      }
      \centering

      \vspace{1mm}
      \scalebox{0.9}{
        (a) An $\nra$ $\aut$ with a~disequality guard
      }
      \end{minipage}
      \hfill
      \begin{minipage}[b]{0.50\textwidth}
      \scalebox{0.8}{
      \input{figs/rsa-problem-diseq.tikz}
      }
      \centering

      \vspace{-2mm}
      \scalebox{0.9}{
        (b) A~part of the~$\drsa$ obtained for~$\aut$
      }
      \end{minipage}
    \end{center}
\noindent
    where (b) contains a~part of the~$\drsa$ obtained if \cref{alg:det} did not
    use the $c$-component of macrostates.
    The reason for this is that after reading the third symbol (i.e., $\pair a
    2$), the~$\rsa$ goes to the macrostate~$\set q$---it thinks it cannot be
    in~$s$ any more.
\begin{changebar}
    This is the reason why we augment macrostates with the~$c$-component, the register size classes.
    They allow to detect that a~disequality test is performed on a~register
    containing more than one element, in which case we terminate the algorithm
    (\lnref{ln:cardinality}).
    The sizes classes are updated on
    \lnref{ln:counter_update}
    where the sum is \emph{saturated} to~$\omega$ for values ${}> 1$ (denoted
    by $\sum\limits^{>1\leadsto \omega}$).
\end{changebar}

  \item 
\begin{changebar}
\emph{Checking for Cartesian overapproximation (Lines 16 to 19).} 
    By collecting all possible values that can occur in registers, the
    algorithm is performing the so-called \emph{Cartesian overapproximation} 
\end{changebar}
   (i.e., it
    is losing information about dependencies between components in tuples).
    This can lead to a~scenario where, for some set-register assignment~$f'$
    of~$\aut'$, we would have $d_1 \in f'(r_1)$ and $d_2 \in f'(r_2)$, but
    there would be no corresponding configuration of~$\aut$ with register
    assignment~$f$ such that $d_1 = f(r_1)$ and $d_2 = f(r_2)$.
    Consider, e.g., an~$\nra$ for the language
    $\set{uvwvz \mid u,w,z \in (\Sigma \times \datadom)^*, |v| = 2}$:

    \begin{center}
      \scalebox{0.8}{
      \input{figs/nra_problem_cartesian.tikz}
      }
    \end{center}

    \noindent
    When the algorithm computes the successor of the macrostate $(\set{q,s,t},
    \set{r_1{:}1, r_2{:}1, r_3{:}1})$ over $a\in \Sigma$ and the guard $\guard =
    \emptyset$, it would obtain the following update of registers:
    $r_1 \gets \set \inp$ (transition from $q$ to $s$),
    $r_2 \gets r_1 \cup r_2$ (transition from $s$ to $t$ and transition from $t$
    to $t$), and
    $r_3 \gets r_3 \cup \set \inp$ (transition from $s$ to $t$ and transition
    from $t$ to $t$).
    This would simulate also the update $r_2 \gets r_2$, $r_3 \gets \inp$, which is
    nowhere in the original~$\nra$.
    The algorithm detects the possibility of such an overapproximation on
    Lines~\ref{ln:approx_test_loop}--\ref{ln:abort}.
    Note that the overapproximation checking is overly conservative in a way that 
    for some automata, if we did not do the check and continued the
    construction, outputting the $\drsa$, the resulting $\drsa$ would still be
    correct (see \ifTR\cref{ex:cartesian-overapp} in
    \cref{sec:app-examples}\else \cite{techrep}\fi).\footnote{In the implementation, we actually
    postpone the overapproximation test only after the whole $\drsa$ is
    constructed.  At this point, we can check more precisely whether there is some real
    overapproximation.}

    %
    %
    %

\item
\begin{changebar}
  \emph{Choice collapse at tests (Line 9).} 
    When a~set-register has collected several
    nondeterministic choices of values for a~standard $\nra$ register, then testing membership of a concrete value, 
    which stands for testing equality of the $\nra$ register, should collapse the choices to that single value.
    Without the collapse, a subsequent second membership test with a different value might pass as if the $\nra$ register could hold two different values at once within a~single non-deterministic case.   
\end{changebar}
    Consider the following example of an~$\nra$~$\aut$ and an~$\rsa$ obtained
    from~$\aut$ by \cref{alg:det} without the substitution $\update\subst{\guardeq}{\inp}$ 
\begin{changebar}
on Line 9 
\end{changebar}
(to
    save space, we collapse all macrostates with the same set of states into one):

    \vspace{-2mm}
    \begin{center}
      \begin{minipage}[b]{0.45\textwidth}
      \scalebox{0.9}{
      \input{figs/nra-problem-write.tikz}
      }
      \centering

      \vspace{-2mm}
      \scalebox{0.9}{
        (a) An $\nra$ $\aut$
      }
      \end{minipage}
      \hfill
      \begin{minipage}[b]{0.45\textwidth}
      \scalebox{0.9}{
      \input{figs/rsa-problem-write.tikz}
      }
      \centering

      \vspace{9mm}
      \scalebox{0.9}{
        (b) A~$\drsa$ overapproximating $\aut$'s language
      }
      \end{minipage}
    \end{center}
    One can see that while $\aut$ cannot accept the word
    $\pair a 1\pair a 2 \pair b 1 \pair b 2$, the $\drsa$ can.
    This happens because the $\drsa$ did not ``\emph{collapse}'' the possible
    nondeterministic choices that are kept in the registers for the value
    of~$r_q$ after the first membership test (on the transition from $\set{q}$
    to~$\set{s}$) succeeded.
    We avoid this situation by the substitution 
\begin{changebar}
on Line~9,
\end{changebar}
    which performs the collapse of the set of nondeterministic choices given by particular values of $\guardeq$ into
    a~single value when it is positively tested.
    The update on the $\rsa$ transition from~$\set q$ to~$\set s$
    constructed by the algorithm will then become $r_s \gets \set \inp$ and the
    result will be~precise.

    One might also imagine similar scenario as the previous but with several
    registers copying a~nondeterministically chosen value (e.g., when
    a~data value is copied from~$r_1$ to~$r_2$ and, later, $r_1$ is
    positively tested for equality, we need to guarantee that the value of~$r_2$
    also collapses to the given data value).
    In order to avoid this, we require that the input $\nra$ is
    \emph{copyless}, i.e., it never happens that a~data value is in more
    than one register.
\end{enumerate}

The correctness of \cref{alg:det} is summarized by the following theorem,
proved in \ifTR\cref{sec:det-proof}\else\cite{techrep}\fi.
\begin{restatable}{theorem}{thmDetSoundness}\label{thm:alg_soundness}
  When \cref{alg:det} returns a~$\drsa$ $\aut'$, then $\langof{\aut} =
  \langof{\aut'}$.
\end{restatable}

\cbstart
The following proposition establishes the numbers of states and
transitions of the output $\drsa$ of \cref{alg:det}, which follow directly from
the structure of macrostates and transitions of~$\aut'$.

\begin{restatable}{proposition}{thmDetComplexity}\label{thm:alg_complexity}
  Let $\aut = (Q, \regs, \Delta, I, F)$ be an RA and let 
  \cref{alg:det} return a~$\drsa$ $\aut' = (\mathcal{Q}', \regs, \Delta', I', F')$.
  Then $|\mathcal{Q}'| \leq 2^{|Q| + (\log_2 3) \cdot |\regs|}$ and $|\Delta'|
  \leq |\Sigma| \cdot 2^{|Q| + (\log_2 6)\cdot |\regs|}$.
\end{restatable}

We note that the lower bound on~$|\mathcal{Q}'|$ is~$2^{|Q|}$, which comes from the
lower bounds of determinization of finite automata~\cite[Section~1.4.1]{EsparzaB23}.
In the RAs obtained from regexes, typically $|Q| \gg |\regs|$.

We can also modify \cref{alg:det} to omit the~$c$-component of macrostates
$(S,c)$, which will give us the bounds $|\mathcal{Q}'| \leq 2^{|Q|}$ (i.e., the
same as for finite automata determinization) and $|\Delta'| \leq |\Sigma| \cdot
2^{|Q| + |\regs|}$.
This simplification of the algorithm may in practice restrict the class of
input RAs on which it successfully terminates.
This is, however, not an issue for RAs \emph{without} disequality
guards~$\guardneq$, which are output by our algorithm for converting regexes to
RAs (cf.\ \cref{sec:regexes-to-ras}); for those, the modified algorithm
successfully terminates in the same cases (the size of the result might,
however, be different in both directions---on the one hand, having the
$c$-component increases the maximum possible number of states, but on the other
hand, keeping track of which registers are empty can avoid generation of
transitions and states that can never be used in a~run).


\cbend

\vspace{-0.0mm}
\subsection{Improvements}
\vspace{-0.0mm}

In this section, we propose improvements of the determinization algorithm. In particular, we 
introduce slight modifications of \cref{alg:det} as well as additional preprocessing of input 
NRAs aiming at enlarging the class of NRAs that can be determinized to DRSAs

\subsubsection{Refining the Register Size Map}

One approach to enlarge the class of NRAs that can be determinized by \cref{alg:det} is 
to keep the value of $c$ of each macrostate that overapproximate the register size as 
precise as possible. Indeed, the value of $c$ affects the condition on \lnref{ln:cardinality}.
A way how $c$ might become unnecessarily high is when $\mathit{tmp} \cap g \neq \emptyset$
and $\inp \in \mathit{tmp}$. In that case, $c(\inp) = 1$ does not need to be added to $c'$ since 
according to the transition guard we know that $\inp$ is in $g$. We hence 
modify the \lnref{ln:update_collapse} to 
$$\mathit{op}_{r_i} \gets 
\begin{cases}
  \mathit{tmp} \setminus \set \inp & \text{if } \mathit{tmp} \cap \guard \neq \emptyset \text{ and} \\
  \mathit{tmp}& \text{otherwise}
\end{cases}$$

\subsubsection{Register Locality}\label{sec:reg-locality}

Another feature limiting the determinizability of the given NRA are nondeterministic transitions  
dealing with the same registers violating the Cartesian overapproximation checked on 
\lnref{ln:abort}. In order to increase the locality of registers, we propose 
register-local NRAs. Locality of registers limiting updates 
of the same registers in different states and hence reducing 
the possibility of the condition on \lnref{ln:abort} holds true. 
Formally, we call~$\aut$ \emph{register-local} if
for all $r \in \regs$ it holds that if $r \in \regsof q$ and $r \in \regsof s$
for some states $q,s \in Q$, then $q = s$.
It is easy to see that every $\nra$ can be transformed into the register-local
form by creating a~new copy of a~register for every state that uses it,
potentially increasing the number of registers to~$|Q|\cdot|\regs|$.
Formally, given a NRA $\aut = (Q, \regs, \Delta, I, F)$, we define its \emph{register-localization} as
$\aut_{\bullet} = (Q, \regs_{\bullet}, \Delta_{\bullet}, I, F)$
where $\regs_{\bullet} = \{ r_q \mid q \in Q, r \in \regs \}$ and $\Delta_{\bullet}$ is defined as follows:
\begin{equation}
  \begin{split}
  \Delta_{\bullet} = \Big\{ \trans p {a} {\guardeq_{\bullet},\guardneq_{\bullet}, 
  \update_{\bullet}} {q}&  \mid \trans p {a} {\guardeq, \guardneq, \update} {q} \in \Delta,\quad \guardeq_{\bullet} = v_p(\guardeq),\quad \guardneq_{\bullet} = v_p(\guardneq), \\ 
  &\hspace{-8mm}\update_{\bullet} = \{ (v_q(r), v_p(t)) \mid (r, t) \in \update \} \cup \{ (v_s(r), \bot) \mid s\in Q\setminus\{q\}, r \in \regs \} \Big\}
  \end{split}
\end{equation}
where the localization function $v_q$ is defined as $v_q(x) = x_q$ if $x \in \regs$, otherwise $v_q(x) = x$ for $x \in \{ \inp, \bot \}$. 
We extend the definition of the localization function to set of values in the usual way.
As a prior step to \cref{alg:det}, we first convert input NRA to its register-local 
form. 

\subsubsection{Relaxing the Copyless Property}

Conversion of an input NRA to an equivalent copyless NRA may introduce non-equality 
guards on registers having nondeterministically chosen values (i.e., in two runs over the same string,
the registers can have different values at the same position in the string). 
Such non-equality guards can cause \cref{alg:det} to return $\bot$ on \lnref{ln:cardinality}
as it is shown in \cref{ex:preproc}.
In order to reduce the introduction of non-equality guards, we relax the copyless 
condition into \emph{relation-free} condition in a way that there is no transition where multiple registers are assigned the same value
regardless of the input. In other words, registers can hold the same value as long 
a run does not induce equality relations on the register values.
Syntactically, relation-freeness can be expressed as each register must appear on a right-hand side of an
update at most once, and at most one register can be updated by $\inp$ or a~register
in the equality guard. Note that every NRA can be converted to a relation-free form,
which is then used as an input of \cref{alg:det}.


\begin{example}\label{ex:preproc}
	Consider the NRA over $\Sigma = \set{a}$ shown in \cref{fig:preproc_ra_example1}.
	It non-deterministically selects a data value to store in $r_1$ and then stores
	the following data value in $r_2$.
	Equivalent copyless NRA is shown in
	\cref{fig:preproc_ra_example2}. The NRA non-deterministically selects the
	data value for $r_1$ as well, but before storing the next data value in $r_2$, it must check that
	the data value is not already stored in $r_1$. If the data value is already stored in $r_1$,
	then it marks that $r_1=r_2$ in its state control and does not actually store anything
	in $r_2$. With this construction, however, we necessarily introduce a non-equality
	guard on $r_1$, which eventually causes \cref{alg:det} to return $\bot$ on
	\lnref{ln:cardinality}.
\end{example}

\begin{figure}[t]
	\begin{subfigure}{0.48\textwidth}
		\centering
    \scalebox{0.8}{
		\input{figs/preproc_ra_example.tikz}
    }
		\caption{Non-copyless form}\label{fig:preproc_ra_example1}
	\end{subfigure}
	\begin{subfigure}{0.48\textwidth}
		\centering
    \scalebox{0.8}{
		\input{figs/preproc_svra_example.tikz}
    }
		\caption{Copyless form}\label{fig:preproc_ra_example2}
	\end{subfigure}
	\caption{Copyless conversion for an NRA that stores data values in two registers.}
\end{figure}

\subsection{Determinizability} 

Naturally, we wish to syntactically characterise the class of $\nra$s for which
\cref{alg:det} is complete.
We observe that when we start with an $\nraeqof 1$ and apply the register localization the algorithm always returns a~$\drsa$ (the theorem is proved in \ifTR\cref{sec:proof-determinisability}\else\cite{techrep}\fi).

\begin{restatable}{theorem}{thmNraOneToDrsa}\label{thm:nra1-to-drsa}
  For every $\nraeqof 1$, there exists a~$\drsa$ accepting the same language.
\end{restatable}

Let $\booleanof{\nraeqof 1}$ be the class of languages that can be expressed using
a~Boolean combination of $\nraeqof 1$ languages, i.e., it is the closure of
$\nraeqof 1$ languages under union, and intersection, and complement (it could
also be denoted as $\booleanof{\uraeqof 1}$).

\begin{example}\label{ex:lang_exnegexrep}
For instance, the language $\langexnegexrep$ composed as the concatenation of
$\langexrep$ and $\langnegexrep$ with a~delimiter, formally
\begin{equation*}
  \hlmathbox{
  \langexnegexrep = \langexrep \concat \set{\pair b d \mid d \in \datadom}
  \concat \langnegexrep,
}
\end{equation*}
is in $\booleanof{\nraeqof 1}$, since it is the intersection of languages
$$
  \langexrep \concat \set{\pair b d \mid d \in \datadom}
  \concat \set{\pair a d \mid d \in \datadom}^*
  ~\text{and}~
  \set{\pair a d \mid d \in \datadom}^* \concat \set{\pair b d \mid d \in
  \datadom} \concat \langnegexrep,
$$
but is expressible neither by an~$\nra$ nor by a~$\ura$ ($\ura$s cannot express
the part \emph{before} the delimiter and $\nra$s cannot express the part
\emph{after} the delimiter).
\qed
\end{example}


The determinizability of $\booleanof{\nraeqof 1}$ then follows 
directly from \cref{thm:nra1-to-drsa,thm:drsa-closure}.
\begin{restatable}{corollary}{thmBooleanDrsa}\label{thm:boolean-drsa}
For any language in $\booleanof{\nraeqof 1}$, there exists a~$\drsa$ accepting it.
\end{restatable}
%
A direct consequence of \cref{thm:boolean-drsa} and decidability of the emptiness check is 
also decidability of the inclusion problem as it is shown in the following corollary.
\begin{corollary}
The inclusion problem between $\rsa$ and $\booleanof{\nraeqof 1}$ is decidable.  \end{corollary}
\begin{proof}
  We just write $\langof{\aut_1} \subseteq \langof{\aut_2}$ as $\langof{\aut_1}
  \cap \cmplof{\langof{\aut_2}} = \emptyset$ and use
  \cref{thm:boolean-drsa,thm:rsa-closure,thm:rsa-emptiness}.
\end{proof}








%
%

\newcommand{\figAntimirov}[0]{
\begin{figure}[t]
\resizebox{\textwidth}{!}{
\begin{minipage}{15cm}
\begin{align*}
  \derivof a \epsilon ={} & \emptyset &
    \nullaof{\epsilon} \liff {}& \mytrue \\
  \derivof a S ={} & \begin{cases}
    \{\dertuple \epsilon \bot \bot\} & \text{if } a \in S\\
    \emptyset & \text{otherwise}
  \end{cases} &
    \nullaof{S} \liff {}& \myfalse \\
  \derivof a {r_1 + r_2} ={} & \derivof a {r_1} \cup \derivof a {r_2} &
    \nullaof{r_1 + r_2} \liff{}& \nullaof{r_1} \lor \nullaof{r_2} \\
  \derivof a {r_1 \concat r_2} ={} & \begin{cases}
    \alpha_a(r_1 \concat r_2) \cup \derivof a {r_2} & \text{if }\nullaof{r_1}\\ 
    \alpha_a(r_1 \concat r_2) & \text{otherwise}
  \end{cases}\hspace*{5mm} &
    \nullaof{r_1 \concat r_2} \liff {} & \nullaof{r_1} \land \nullaof{r_2} \\
  &\rlap{$\text{where } \alpha_a(r_1 \concat r_2) = \{\dertuple{r \concat r_2} {\mathit{tst}} {\mathit{up}} \mid \dertuple r {\mathit{tst}}{\mathit{up}} \in \derivof a {r_1}\}$}\\
  \derivof a {r^*} ={} & \{\dertuple{r'\concat r^*} {\mathit{tst}} {\mathit{up}} \mid \dertuple{r'}{\mathit{tst}}{\mathit{up}} \in \derivof a r\} &
    \nullaof{r^*} \liff {} & \mytrue \\
  \derivof a {\capgroup{S}_m} = {}& \{\dertuple r \bot {r_m} \mid \dertuple r \bot \bot \in \derivof a S\} &
    \nullaof{\capgroup{S}_m} \liff {} & \myfalse \\
  \derivof a {\backref m} = {}& \{\dertuple \epsilon {r_m} \bot\} &
    \nullaof{\backref m} \liff {} & \myfalse
\end{align*}
\end{minipage}
}
\caption{Antimirov derivatives for \rewbs.  The
  predicate $\nullaof{\redef}$ outputs $\mytrue$ iff $\epsilon \in
  \dwlang\lang(\redef)$.}
\label{fig:antimirov}
\end{figure}
}

\vspace{-0.0mm}
\section{Matching of Regexes with Backreferences}\label{sec:regexes}
\vspace{-0.0mm}

In this chapter, we define regexes with backreferences and show how they can be
converted into NRAs using a~construction based on Antimirov derivatives.
We then establish the complexity of regex matching with DRSAs constructed from
these regexes.

\vspace{-0.0mm}
\subsection{Regexes with Backreferences}\label{sec:regexes-to-ras}
\vspace{-0.0mm}

\paragraph{Syntax.}
A \emph{regex with backreferences} (\rewb) $\redef$ over alphabet $\Sigma$ is defined
inductively according to the following grammar:
\begin{align*}
  \redef &::= \epsilon ~\mid~ S ~\mid~ \redef + \redef ~\mid~ \redef\cdot \redef ~\mid~ \redef^* ~\mid~ \capgroup{S}_m ~\mid~ \backref{m}
\end{align*}
where $S \subseteq \Sigma$, and $m \in \nat$. 
In addition to the standard regex syntax, the definition introduces
\begin{inparaenum}[(i)]
  \item  a~\emph{capture group} $\capgroup{S}_m$ (explicitly indexed by~$m$),
    capturing a~single symbol from the set~$S$, and
  \item  a~\emph{backreference} $\backref{m}$ to the $m$-th capture group.
\end{inparaenum}
We use~$\regex$ for the set of all \rewbs.
Note that compared to the languages that can be specified using, e.g.,
\emph{Perl Compatible Regular Expressions} (PCREs)~\cite{pcre2}, our definition
does not enable expressing languages where capture groups have unbounded
length, e.g., the language defined by the PCRE
\texttt{/\textasciicircum(.*)\textbackslash{}1\$/}, which matches strings of
the form $ww$. 
Capture groups with bounded length (but longer than 1) can be expressed using \rewbs by 
splitting the longer capture groups into single-letter ones, e.g., PCRE \texttt{/(...)\textbackslash{}1/} can be 
expressed as \texttt{/(.)(.)(.)\textbackslash{}1\textbackslash{}2\textbackslash{}3/}.
%

\paragraph{Semantics.}
We will define the semantics of \rewb~$\redef$ in two phases: 1) in terms of \emph{annotated words}, with marked parts matched by capture-groups and ``pointers'' to capture groups in place of backreferences; and 2) normal words obtained by interpreting the annotations. 

In phase 1, subwords matching capture groups are annotated by the group's index $m$ and the keyword $\mathtt{in}$, and backreference annotations consist of the index of a capture group and the keyword $\mathtt{ref}$. The language of the annotated words $\reflang$ is hence a language over the alphabet $\Sigmaref = \Sigma \cup \{ \symin{s}{m}, \symref{m} \mid s\in\Sigma, 0 \leq m \leq k_\redef \}$, with $k_\redef$ being the maximum index of a~capture group or a backreference in~$\redef$, defined as follows:
%
%
\begin{center}
\begin{minipage}{0.48\textwidth}
  \[\begin{aligned}
  \reflangof \epsilon &= \{ \epsilon \} \\
  \reflangof S &= S \\
  \reflangof{\redef_1 + \redef_2} &= \reflangof{\redef_1} \cup \reflangof{\redef_2} \\
  \reflangof{\redef_1 \cdot \redef_2} &= \reflangof{\redef_1} \cdot \reflangof{\redef_2}
  \end{aligned}\]
\end{minipage}\hspace{5mm}%
  \begin{minipage}{0.48\textwidth}
  \[\begin{aligned}
  \reflangof{\redef_1^*} &= (\reflangof{\redef_1})^* \\
  \reflangof {\capgroup{S}_m} &= \{ \symin{s}{m} \mid s \in S \} \\
  \reflangof{\backref{m}} &= \{ \symref{m} \}
  \end{aligned}\]
\end{minipage}
\end{center}
Intuitively, a backreference should match the character that the corresponding capture group matched last in the prefix left of the backreference.  
Formally, for the $m$-th capture group and a~prefix~$w$, 
it is the character $\capg(w, m) = c$ 
if it is possible to write $w$ as $w_p\cdot\symin{c}{m}\cdot w_s$ 
where $w_p \in \Sigmaref^*$ and $w_s \in \Sigmaref^*$ does not contain \emph{any occurrence} of a~symbol~$\symin{a}{m}$ for arbitrary $a\in\Sigma$ (hence this split of the prefix $w$ indeed marks the last match of the capture group in it). 
Otherwise, if $w$ has no occurrence of an annotated capture group match $\symin{a}{m}$, then $\capg(w, m) = \bot$.

In order to obtain a~language over the alphabet~$\Sigma$ from $\reflang$, we define
\emph{annotation interpretation}, denoted as~$\bproj$, which 
replaces the backreference annotations with the captured characters and 
removes the capture group annotations.
Formally,
given a~word $w\in\Sigmaref^*$ and an index $1 \leq i \leq |w|$, $\bproj$ is
defined as 
\begin{equation*}
\bproj(w, i) = 
\begin{cases}
  \capg(w[1:i-1], m) & \text{ if } w[i] = \symref{m}, \\
  c & \text{ if } w[i] = \symin{c}{m}, \\
  w[i] & \text{ otherwise,} 
\end{cases}
\end{equation*}
where $w[1:i] = w_1 \ldots w_i$ for $w = w_1 \ldots w_{|w|}$.
We lift the definition to a word as $\bproj(w) = \{ w' \mid w' = \bproj(w, 1) \cdots \bproj(w, |w|), w' \text{ does not contain }\bot \}$. Notice that removing the words containing $\bot$ removes exactly cases where a backreference occurs before the corresponding capture group is matched.
The \emph{language} of a~\rewb $\redef$ is then defined as $\langof{\redef} =
\bigcup_{w \in \reflang(\redef)} \bproj(w)$.
Since our automaton model works with data words and the semantics of \rewbs are defined as languages over $\Sigma$, we define 
the extension of words/languages over $\Sigma$ to data words. For that we consider some \emph{fixed} linear ordering on $\Sigma$, meaning 
that there is an injective mapping $\dword\colon \Sigma \to \nat$. Then, for a word $w = a_1\cdots a_n \in\Sigma^*$, we define $\dwlang{w} = \pair{a_1}{\dword(a_1)}\cdots \pair{a_n}{\dword(a_n)}$. 
We lift the definition to languages $\dwlang \lang$ in the usual way.

\paragraph{Antimirov Derivatives.}

We can translate \rewbs to RAs using a~construction based on
\emph{Antimirov derivatives}~\cite{Antimirov96}.
A~\emph{partial derivative} is a~function
$\deriv\colon\Sigma \times \regex \to 2^{\regex \times \natbot \times \natbot}$
(we use $\natbot$ to denote the set $\nat \cup \{\bot\}$)
such that $\derivof a \redef$ gives us all possible \rewbs (together with
potential information about writing or testing a~register) obtained
from~$\redef$ by trying to match a~string with leading symbol~$a \in \Sigma$.
The definition of the derivatives is in \cref{fig:antimirov}.
In the definition, we also use the predicate $\nullaof \redef$, which denotes
that~$\redef$ is \emph{nullable}, i.e., it accepts the empty string.
We note that our class of \rewbs allows this relatively simple extension of
Antimirov derivatives; if one considered unbounded capture groups, the
derivatives would get more complex.



\figAntimirov  

\paragraph{NRA Construction.}

We use the partial derivatives defined above to transform a~regex~$\redef$ into
an~NRA
\begin{changebar}
in the function $\regexra$.
\end{changebar}
Let~$k_\redef$ be the maximum index of a~capture group or a~backreference
occurring in~$\redef$.
Then we define the NRA 
$\regexra(\redef) = (Q, \{r_0, \ldots, r_{k_\redef}\}, \Delta, \{\redef\}, F)$
where
\begin{itemize}
  \item  the set of states~$Q$ contains all regexes that are reachable
    from~$\redef$ by the partial derivative construction, i.e., $Q = \mu
    Z\colon \bigcup\{q' \mid \dertuple{q'}{\cdot}{\cdot} \in \derivof
    a q, q \in Z, a \in \Sigma\}  \cup \{\redef\} \cup Z$,
  \item  $\Delta$ is the smallest set such that if
    $\dertuple{q'}{t}{u} \in \derivof a q$ for
    $a \in \Sigma$, then $\trans{q}{a}{\guardeq, \emptyset,
    \update}{q'} \in \Delta$ where
    \begin{itemize}
      \item  $\guardeq = \{r_i\}$ if $u = i \neq \bot$, else $\guardeq = \emptyset$ and
      \item  $\update = \{r_i \mapsto r_i \mid 0 \leq i \leq k_\redef, i \neq
        j\} \cup \{r_j \mapsto \inp\}$ if $u = j \neq \bot$, else
             $\update = \{r_i \mapsto r_i \mid 0 \leq i \leq k_\redef\}$, and
    \end{itemize}

  \item  $F = \{ q \in Q \mid \nullaof q\}$.
\end{itemize}
\begin{changebar}
Note that the number of used registers matches the number of capture
groups/backreferences in~$\redef$.
\end{changebar}
The argument that the number of states of~$\regexra(\redef)$ is finite is
similar as the one in~\cite{Antimirov96};
\begin{changebar}
more precisely, the number of states of the output NRA is in $\bigO(|\redef|)$.
\end{changebar}
Then, using the properties of Antimirov derivatives 
and the construction described above, we have the following theorem.

\begin{theorem}
  For a~regex $\redef$ we have $\langof{\regexra(\redef)} = \dwlang{\semof{\redef}}$.
\end{theorem}



%
%


\vspace{-0.0mm}
\subsection{Regex Matching}\label{sec:label}
\vspace{-0.0mm}

In order to match a~word $w\in\Sigma^*$ w.r.t.\ a given \rewb $\redef$, we first 
construct an RA $\aut$ corresponding to $\redef$ (using $\regexra$). Then, after register localization, we use \cref{alg:det} to determinize $\aut$.\footnote{
Regarding the determinizability of RAs obtained from \rewbs, from \cref{thm:nra1-to-drsa} and the properties of the RA construction 
we have that we are able to process arbitrary \rewbs containing a single capture group (but possibly multiple backreferences).}
If the algorithm does not return $\bot$, we have a~DRSA~$\aut'$ representing~$\redef$.
When testing regex membership, we check that $\dwlang{w} \in \langof{\aut'}$ using an algorithm
tracking configurations obtained during reading of~$\dwlang{w}$. 
The following lemma establishing the complexity of this algorithm.

\begin{lemma}\label{lem:memb-compl}
  Let $w$ be a data word and $\aut = (Q, \regs, \Delta, I, F)$ be a DRSA over a finite data domain $\datadom$. 
  Checking that $w\in\langof{\aut}$ can be done in time $\bigO(|\datadom|\cdot|w|\cdot|\regs|^2)$.
\end{lemma}
%
\begin{proof}[Proof Sketch.]
  Consider a data word $w$ and a DRSA $\aut = (Q, \regs, \Delta, I, F)$. 
  Further let $\mathcal{D} = \{ v_i \mid v_1\dots v_n = \datadom[w] \}$ (note that 
  $|\mathcal{D}|\leq |\datadom|$). 
  In this proof and the following text we assume a unit cost of comparison between elements of $\datadom$.
  The membership checking algorithm tracks configurations 
  of the form $(q, r)$ where $q\in Q$ and $r\colon \regs \to 2^{\mathcal{D}}$.
  We assume that set-register values are stored in a~sparse-set data structure with constant insertion and membership and with the union of two 
  sets linear to the size of one of the two sets (elements of one sets are inserted to the other, creation of a~new sparse-set is not needed as the number of registers is fixed). 
  When constructing a~next-state configuration, we first need to 
  check transition guards to get a successor transition. Here, we assume that the successor 
  transitions are kept in a~BDD-like structure. Hence, for each register, it suffices 
  to check whether the input symbol is in the register, and then use this bitvector to find an appropriate 
  transition (recall the automaton is deterministic). This can be done in $\bigO(|\regs|)$.
  Then, we need to compute the transition update, which can be done in $\bigO(|\datadom|\cdot|\regs|^2)$.
  Therefore the overall complexity is $\bigO(|\datadom|\cdot|w|\cdot|\regs|^2)$. 
\end{proof}

Note that if the data domain is infinite, testing $w \in \langof \aut$ can be
done in time $\bigO(|w|^2\cdot|\regs|^2)$. 
Directly from \cref{lem:memb-compl} we have the following theorem.

\begin{theorem}\label{thm:linearmatching}
  For a fixed \rewb~$\redef$ over a fixed alphabet, the worst-case time
  complexity of $\drsa$-based regex matching is \emph{linear} to the length
  of the input word (provided \cref{alg:det} finishes on the RA obtained
  from~$\redef$).
\end{theorem}

\subsection{Finer Complexity Analysis}
The constant character cost complexity of matching in \cref{thm:linearmatching} hides a quadratic dependency on the number of registers $|\regs|$ and the linear dependency on the size of the data domain $|\datadom|$. 
When the DRSA is obtained from a \rewb of the size $m$ and with $r$ backreferences by
\begin{inparaenum}[(i)]
  \item  our derivative-based construction,
  \item  register localization, and
  \item  our determinization (\cref{alg:det}),
\end{inparaenum}
then we have $|\regs| \in \bigO(m \cdot r)$. 
The $|\datadom|$ particularly can be large in practice and problematic. 

We will argue that the factor $|\datadom|$, coming from the need to unite and copy registers, can be avoided if 
the automaton never copies registers. 
Copying is thus never done and the cost of uniting registers is covered by the maximum cost of adding elements to the registers. 
The quadratic dependence on $|\regs| = rm$ also decreases to linear.  
Formally, we define an RSA as 
\emph{copy-free} if it has no transition where it assigns the same non-singleton register to two registers. Practical relevance of this class is witnessed by the fact that in our experiments, the copy-free DRSA were constructed in 488 cases out of 1,335 (our main benchmark set of single-letter regexes that were found ReDoS prone).  
%
\begin{theorem}
For a copy-free DRSA,
the complexity of testing membership is $\bigO(|\regs|  \cdot |w|)$.
\end{theorem}
\begin{proof}[Proof Sketch.]
Let us measure the age of a data value by when it was read form the input, and let the age of a register be the age of its oldest value. 
Let us union two registers by inserting values from the younger one to the older register. With sparse-sets, the complexity would be linear to the size of the younger register, as every insertion is constant-time.
That is, within one union operation, each element in the younger register contributes $\bigO(1)$.
Note that since the automaton is copy-free, a value taken from a position in the input data word can be this way (within the union operation) inserted into older registers at most $|\regs|$ times (after $|\regs|$ unions, the value must be in the oldest register). Hence, inserting the value from a particular position in the input word within unions may take at most $\bigO(|\regs|)$ time, and we have $|w|$ positions, which amounts to $\bigO(|\regs| \cdot |w|)$ for uniting registers overall.
Initialising the sparse-sets would take $\bigO(|\regs|\cdot|\datadom|)$ time, but $|\datadom|$ can be w.l.o.g. capped at $|w|$. 
\end{proof}

Note that also testing $w \in \langof \aut$ with an infinite $\datadom$ with copy-free DRSA would be in $\bigO(|\regs| \cdot |w|)$.
\newcommand{\tableRedos}[0]{
\begin{table}[t]
  \caption{
  \begin{changebar}
  A frequency table of run times and statistics for \tool and the other matchers on the
  \end{changebar}
  1,246 supported regexes with the \rengar-generated inputs.
  The column \textbf{TOs} denotes timeout (\ensuremath{>}~100\,s),
  column~\textbf{\ensuremath{\boldsymbol{>}}~1\,s} is the sum of the values in all columns over 1\,s, and
  the \textbf{errors} column shows the numbers of unsupported regexes.
  The numbers in the right-hand part of the table are for successful runs only and are in seconds.
  }
  \label{table:nodet}
  \centering
  \vspace{-3mm}
  \resizebox{\textwidth}{!}{
  \input{table-redos-pldi26.tex}
  }
\end{table}
}

\newcommand{\tableStats}[0]{
\begin{wraptable}[11]{r}{0.4\textwidth}
  \vspace*{-3mm}
  \caption{Statistics of run times (in seconds) on finished instances}
  \label{tab:stats}
  \centering
  \vspace*{-4mm}
  \hspace*{-1.6mm}
  \scalebox{0.9}{
  \input{table-stats.tex}
  }
\end{wraptable}
}

\newcommand{\figScatter}[0]{
\begin{figure}[t!]%
\centering
\begin{subfigure}{0.32\textwidth}
  \centering
  \includegraphics[width=\linewidth]{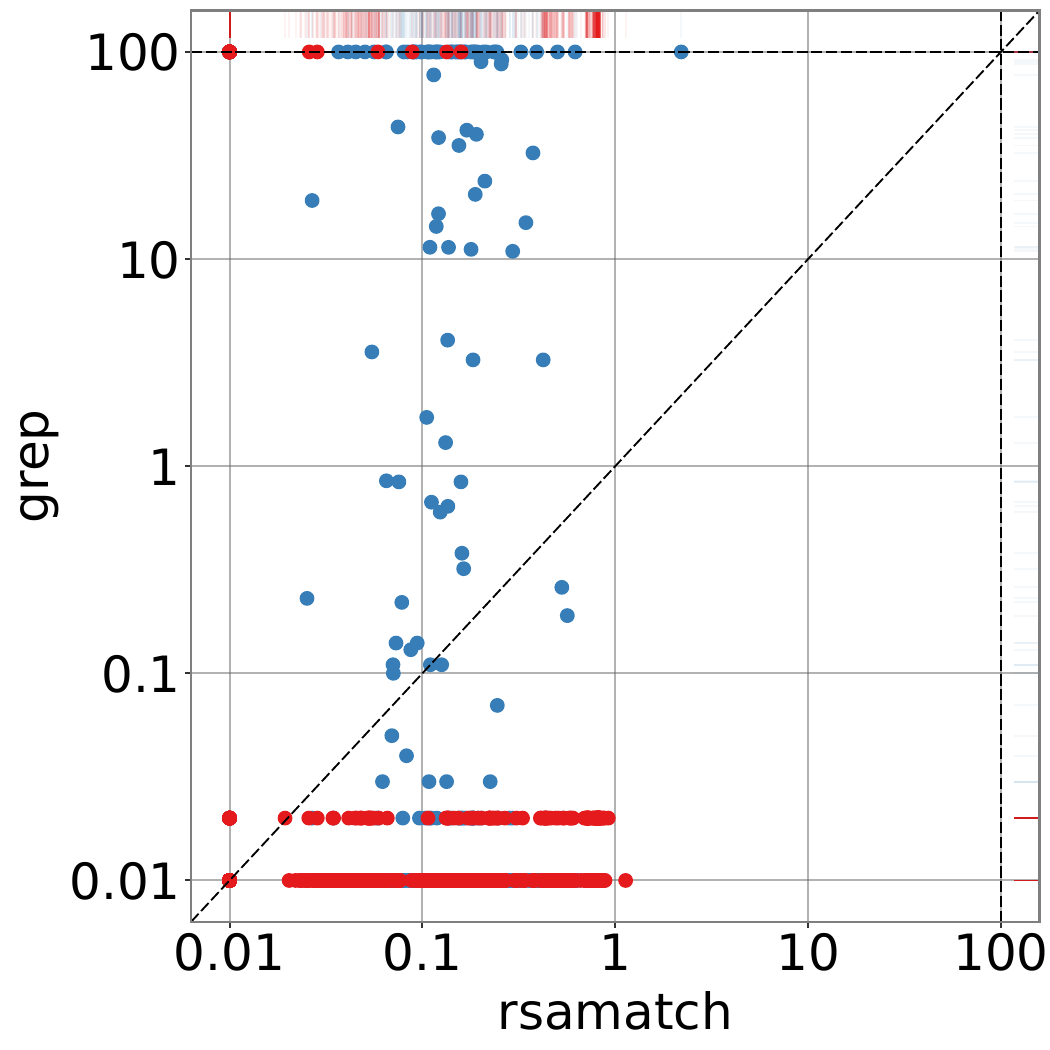}
  \caption{\tool vs. \grep}
  \label{fig:scatter_grep}
\end{subfigure}
\begin{subfigure}{0.32\textwidth}
    \centering
    \includegraphics[width=\linewidth]{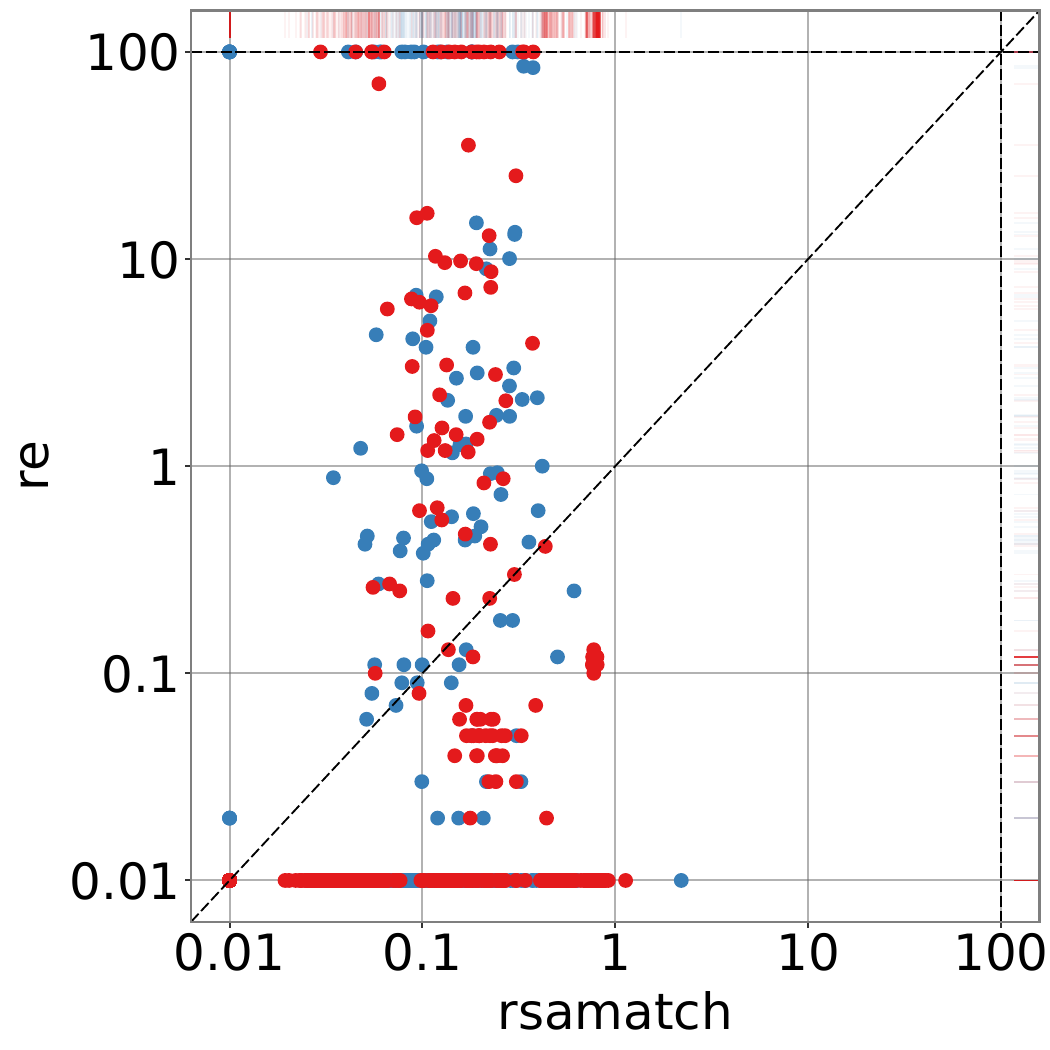}
    \caption{\tool vs. \re}
    \label{fig:scatter_re}
\end{subfigure}
\begin{subfigure}{0.32\textwidth}
    \centering
    \includegraphics[width=\linewidth]{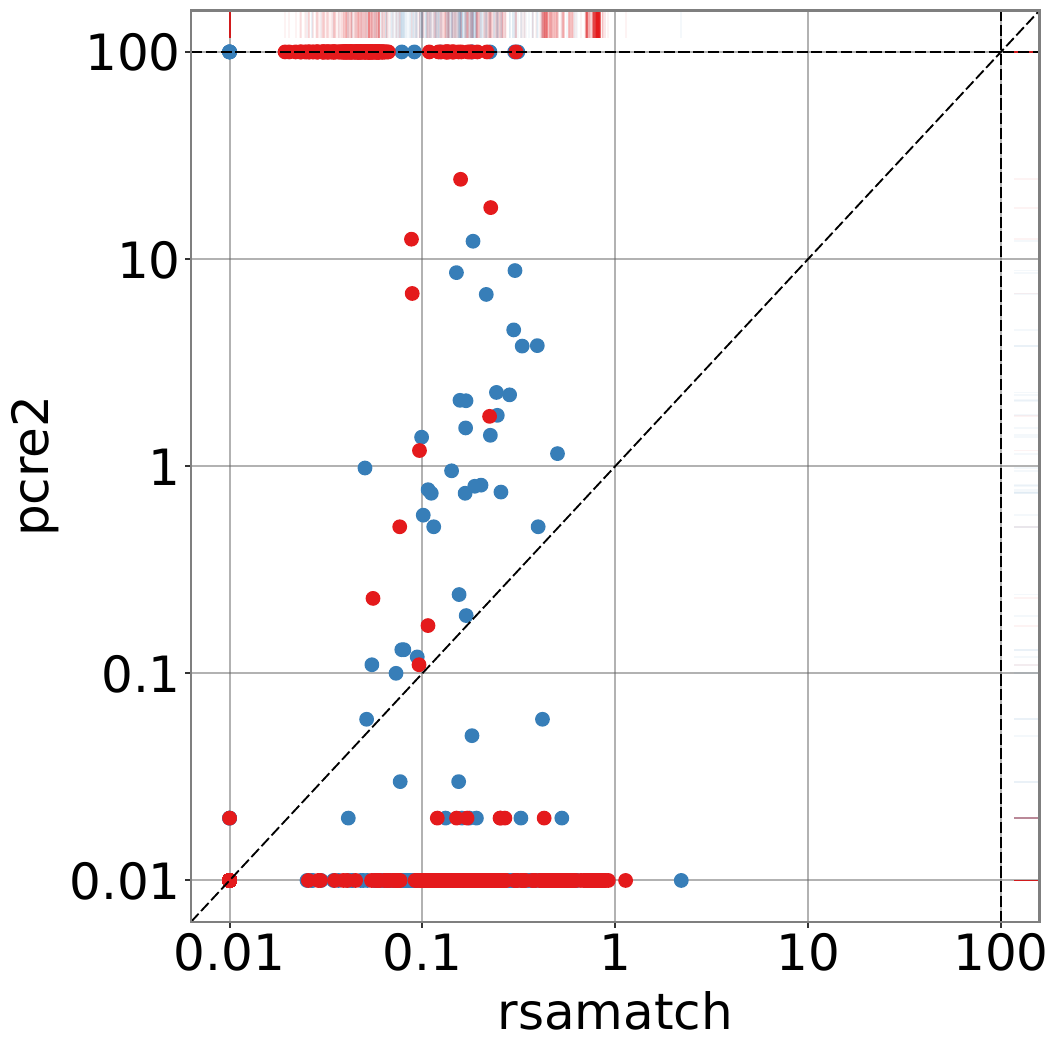}
    \caption{\tool vs. \pcre}
    \label{fig:scatter_pcre2}
\end{subfigure}

\vspace*{-3mm}
\caption{
  Comparison of \tool with \grep, \re, and \pcre. 
  Times are in seconds, axes are logarithmic.
  Dashed lines represent timeouts (100\,s).
    Colours distinguish datasets:
    \RGBcircle{55,129,192}\,Lingua Franca, 
    \RGBcircle{223,38,37}\,\rengar.
}
\label{fig:scatter}
\end{figure}
}

\vspace{-0.0mm}
\section{Experimental Evaluation}\label{sec:experiments}
\vspace{-0.0mm}

We have implemented a prototype RSA-based regex matcher \tool~\cite{rsamatch} in Python and
evaluated its performance against other state-of-the-art regex matchers
on realistic ReDoS attack vectors involving backreferences.
Our experiments are designed to evaluate the following hypotheses:
\begin{enumerate}
\item
Existing regexes matchers are vulnerable when matching regexes with backreferences.
\item
A large portion of regexes with backreferences contain only single-letter backreferences.
\item
Most of the regexes with single-letter backreferences that are vulnerable to
ReDoS can be converted to DRSAs and matched with a~predictable performance.
\end{enumerate}

%
%

\vspace{-0.0mm}
\subsection{Experimental Setup}
\vspace{-0.0mm}

\paragraph{Regex Matchers and Data Sets.}

We compared \tool against a representative sample of state-of-the-art matchers,
namely, \grep~\cite{grep}, \pcre~\cite{pcre2}, Python RE~\cite{python_re}
(denoted as \re), standard library matchers of 
JavaScript~\cite{js_v8_regexp}, Java~\cite{java}, and .NET~\cite{dotnet}
(denoted as \js, \java, and \nettool respectively).
These are mostly backtracking-based, possibly in combination with other techniques, but they all mainly use their backtracking core when confronted with backreferences.
We note that matchers for some scripting languages (Python's \re and \pcre) are
actually implemented in~C and not in the scripting language itself.
Some very fast matchers such as \retwo~\cite{re2} and especially
\hyperscan~\cite{WangHCPLHZ19}, based on automata and other techniques, are
missing in the comparison since they do not support
backreferences~\cite{hyperscan-unsupported-backref,re2-add-backref}. 
All experiments were run on a~Ubuntu GNU/Linux machine with the AMD EPYC 9124
CPU with 16 cores at 3.0\,GHz and 94\,GiB of memory (we note that the memory
consumption of all matchers was far below the limits imposed by the hardware).

As our regexes, we started with 3,252 real-world regexes with backreferences
obtained from the comprehensive data set Lingua Franca~\cite{regexLinguaFranca}
and 8,794 regexes from the benchmarks of the ReDoS generator \rengar~\cite{rengar},
for a~total of 12,046 regexes.\footnote{
\begin{changebar}
In the Lingua Franca data set, 52\,\% of the regexes had capture groups of
fixed lengths and 57\,\% had capture groups of bounded lengths (the
fixed-length are a~subset of the bounded-length).
In the \rengar data set, 22\,\% regexes had capture groups of fixed lengths and
24\,\% had capture groups of bounded lengths.
In total, 30\,\% regexes had capture groups of fixed lengths and 33\,\% had capture groups of bounded lengths.
\end{changebar}
}
From these, we picked regexes with single-letter backreferences because
although multi-letter capture groups can be rewritten as single-letter ones
(for a~bounded length, cf.\ \cref{sec:regexes}), the determinization often
fails for them (\cref{alg:det} returns~$\bot$ due to overapproximation), which
gave us 3,299 regexes (27\,\%).

\paragraph{ReDoS Inputs.}
In general, there are several flavours of ReDoS considered in the literature.
In our evaluation, we consider the scenario where there is a~regex (known to
the attacker) running on a~server (or a~network intrusion detection system or
a~packet filter etc.) and the attacker provides an input that
exhibits a~costly matching, which is, to the best of our knowledge, the most
realistic setting, since it has already led to several real-world
ReDoSes~\cite{stackoutage,expressjsoutage}.\footnote{Some of the other settings
considered in the literature is, e.g., when the attacker can provide both the
regex and the input.
While such a~situation may occur in the real-world, it is less common and much
easier to exploit when backreferences are enabled due to the \np-hardness of
their matching~\cite{Aho90}.
}
The input of every experiment is a~pair of a~regex and an
input string, which was created using the ReDoS generator
\rengar~\cite{rengar}.
\rengar was able to provide an attack vector for 1,335 regexes from the data
set~(40\,\%).
Since in the considered ReDoS setting, one regex is used to match a~large
quantity of user inputs, we perform the determinization in \tool offline and
then run the matching with the precomputed DRSA.
%
Out of the 1,335 regexes with ReDoS inputs, our determinization algorithm
returned~$\bot$ on 43 regexes (3\,\%) and timed out (because of a~regex
\begin{changebar}
complexity other than backreferences; e.g., most of the regexes start with an
implicit \texttt{".*"}, which already creates a nondeterministic choice at
the beginning) on 41 of them (3\,\%).
\end{changebar}
The average runtime of the determinization algorithm on successful instances
was 1.150\,s, the median was 0.260\,s, and the standard deviation was 3.064\,s.
Since the determinization needs to be performed only once for many runs of the matching
algorithm, the cost of determinization can be amortized.
%
\rengar generated strings of the length between 55 and 90,124 characters, with
the average of 34,521 characters.
We set the timeout to quite generous 100\,s.

\cbstart
\paragraph{Determinizability.}
In the previous paragraph, we analyzed on how many of the ReDoS-prone regexes
we managed to construct a corresponding $\drsa$.
Let us now discuss on how many of all input regexes (12,046) we managed to
build the corresponding $\drsa$.
From the 12,046 regexes, 2,944 contained unsupported syntax (we use Python
frontend and its regex dialect), which gave us 9,102 regexes to start with.
Out of these, our algorithm outputs $\drsa$ on 3,138 of them (34.5\,\%). When we
restrict ourselves to the single-letter ones, these were 3,434, making our
algorithm successful on 91\,\% of them.
\cbend

\vspace{-0.0mm}
\subsection{Results}\label{sec:results}
\vspace{-0.0mm}

\begin{table}[t]
  \caption{
  \begin{changebar}
  A frequency table of run times and statistics for \tool and the other matchers on the
  \end{changebar}
  1,246 supported regexes with the \rengar-generated inputs.
  The column \textbf{TOs} denotes timeout (\ensuremath{>}~100\,s),
  column~\textbf{\ensuremath{\boldsymbol{>}}~1\,s} is the sum of the values in all columns over 1\,s, and
  the \textbf{errors} column shows the numbers of unsupported regexes.
  The numbers in the right-hand part of the table are for successful runs only and are in seconds.
  }
  \label{table:nodet}
  \centering
  \vspace{-3mm}
  \resizebox{\textwidth}{!}{
  \input{table-redos-pldi26.tex}
  }
\end{table}

\begin{changebar}
In \cref{table:nodet}, we show a~frequency table (the table form of
a~histogram) of run times of \tool and the other
\end{changebar}
matchers on the 1,246 supported input regexes.
We note that some of the regex matchers failed on a~considerable number of the
input regexes---this is due to different dialects of regexes used by the
matchers and should not be considered their failure (\tool uses Python's
dialect, which is why \re has 0~errors).
The use of different dialects means that the semantics of one regex differs
based on the matcher~\cite{regexLinguaFranca}; 
while this would be an issue if we were are cross-comparing the performance of the
matchers, in this study, we are focusing on ReDoS vulnerability, where the
concrete semantics of a~regex is not so important.

The results show that \tool's run time for the vast majority of benchmarks (all
except two) is below one second (the run times for the two hardest benchmarks
were 1.13 and 2.2 seconds).
The other matchers, on the other hand, indeed struggled on these attack vectors
considerably, often taking tens of seconds or even exceeding the 100\,s
timeout.
Considering that the normal matching time on such strings is counted in
milliseconds, they can indeed be considered manifestations of vulnerabilities,
even for the relatively simple class of regexes with single-letter
backreferences.
If we fix a~ReDoS threshold at~1\,s, the column \textbf{$\boldsymbol{>}$ 1\,s} counts the
number of successful ReDoS attacks.
Here, the second best matcher (after \tool) was \pcre with 42 attacks (although
it could not run on 203 benchmarks) and then \grep with 41 attacks (with
74 benchmarks excluded due to an error).


In the right-hand part of the table, we give statistics about the run times of
the matchers on the instances where they did not timeout (or fail):
\begin{changebar}
the arithmetic mean, the median, and the standard deviation.
\end{changebar}
Although \tool is a~prototype written in Python, its average time is comparable
to the other matchers (also taking into account the fact that they timed out on
a~number of inputs).
The median run time of \tool is higher than for the other matchers, but it should
be significantly decreased if reimplemented in a~more efficient programming
language, such as C/C++.
The last column of the table also exhibits the robustness and predictability of
the performance of \tool, since the standard deviation of the run times is
(despite \tool's prototype nature) an order of magnitude lower than for the
other matchers.

We include a more detailed comparison with the best competing tools, \grep, \re, and
\pcre, in the form of scatter plots shown in \cref{fig:scatter}.
One can clearly see the contrast in performance predictability, with our
matcher almost always staying below the 1\,s mark, on many being
instances significantly faster than the competitors.
Note that \grep either managed to solve the matching problem almost instantly
or timed out.

\figScatter

\vspace{-0.0mm}
\section{Related Work}\label{sec:related}
\vspace{-0.0mm}

\paragraph{Regex Matching.}

The literature on efficiently matching regexes is rich, with works ranging from
matching feature-rich regexes in mainstream languages~\cite{BarriereP24,VarataluVE25}, to
instruction-level optimized regex matching in \hyperscan~\cite{WangHCPLHZ19}, down to 
accelerating regex matching using FPGAs~\cite{CeskaHHKLMMSV19,MatousekMK18}.
Here we focus mainly on automata-based regex matching and regexes with backreferences.

Automata-based regex matching can be traced back to
Thompson~\cite{Thompson1968}, who introduced a~practical algorithm based on
on-demand NFA determinization with caching.
There have been several attempts of automata models for matching regexes with
backreferences, but they have generally failed expectations of the community.
Memory automata of Schmid~\cite{Schmid16} provide a~natural extension of finite
automata with a~register bank, each register able to store a~word, but do not
allow determinization. The language-theoretic aspects of regexes with backreferences 
were further studied in~\cite{FREYDENBERGER20191}.
Chida and Terauchi~\cite{ChidaT23} show that lookaheads increase the expressive
power of regexes with backreferences.
Becchi and Crowley~\cite{BecchiC2008} encode backreferences into an extended
model of NFAs using backtracking.
Nogami and Terauchi~\cite{NogamiT25} devise a~quadratic-time (w.r.t.\ the input
length) algorithm for matching a~certain class of regexes with backreferences,
which is incomparable to our class---they allow capture groups with unbounded
length, but only one capture group and one reference to it; our work, on the
other hand, allows multiple capture groups, multiple references to them, and
our constructed automata often process the input in linear time (with
a~quadratic time worst-case guarantee).
Namjoshi and Narlikar~\cite{NamjoshiN10} extend Thompson's
algorithm~\cite{Thompson1968} to a~model similar to the models
of~\cite{Schmid16,BecchiC2008} in the obvious way by tracking the set of all
possible configurations, which is inefficient (for each symbol, the number of
operations proportional to the (potentially exponential) size of the set of
configurations is required).
Varatalu \emph{et al.}~\cite{VarataluVE25} present a~derivative-based regex
matcher supporting many extensions of classical regexes.
They do not support backreferences, but it seems that our derivatives
(cf.~\cref{sec:regexes}) could be combined with it to extend its usability.
Moreover, there are also some proprietary \emph{ad hoc} solutions
that allow matching of regexes with backreferences~\cite{Nvidia21}, but
they are mostly incomplete and with no guarantees.

In our work, we were inspired by \emph{counting-set automata} introduced by Turoňová \emph{et
al.}~\cite{TuronovaHLSVV20}, which use \emph{sets of counter values} to compactly
represent configurations of \emph{counting automata}~\cite{HolikLSTVV19} (a~restricted
version of counter automata~\cite{Minsky61} with a~bound on the value of
counters for compact representation of finite automata), to obtain
a~deterministic model for efficient matching of regular expressions with
repetitions.

\paragraph{ReDoS}
While initially neglected, the potential of exploiting regex matching for
denial of service was popularized by the successful attacks on
StackOverflow~\cite{stackoutage} and
\texttt{Express.js}~\cite{expressjsoutage}.
While many industrial matchers handle basic regexes reasonably
well~\cite{BarriereP24,VarataluVE25,re2,grep,WangHCPLHZ19}, using extensions can
make the matchers vulnerable to ReDoS; in~\cite{TuronovaHHLVV22}, this was
demonstrated even for highly optimized matchers such as \retwo~\cite{re2} or
\hyperscan~\cite{WangHCPLHZ19} (these two have backreferences explicitly
disabled~\cite{hyperscan-unsupported-backref,re2-add-backref}).
Over the years, there have been many works dealing with the analysis of regexes for
ReDoS vulnerability, sometimes including ReDoS generation, and proposals for
how to avoid this attack vector (see, e.g., the works
in~\cite{HassanALDS23,BarlasDD22,Davis20a,rengar,regexLinguaFranca,LiuZM21,SuHLCG24,ParoliniM22,KirrageRT13,WustholzOHD17,WeidemanMBW16,Shen000ML18,McLaughlinPSKV22}
or the systematic literature review in~\cite{BhuiyanCB0S25}).
On a~more theoretical level,
Terauchi~\cite{Terauchi25} studies conditions for ReDoS vulnerability of
regexes with backreferences and proposes a~technique for transforming a~class
of memory automata to invulnerable regexes.

\paragraph{Automata Models for Data Words.}

The literature on automata over infinite alphabets is rich, see, e.g., the
excellent survey by Segoufin~\cite{Segoufin06} and the paper by Neven \emph{et
al.}~\cite{NevenSV04}.
Register automata were introduced (under the name \emph{finite memory
automata}) by Kaminski and Francez in~\cite{KaminskiF94} and their basic
properties further studied by Sakamoto and Ikeda in~\cite{SakamotoI00}.
Demri and Lazić study in~\cite{DemriL09}
(non-)deterministic, universal, and alternating one-way and two-way register
automata, and their relation to the linear temporal logic with the \emph{freeze}
quantifier, which can store the current data value into a~register.
In particular, they show that $\ltlfreezeofn 1 {\ltlnext, \ltluntil}$, i.e., the
fragment of the logic with one freeze register and the \emph{next} ($\ltlnext$)
and \emph{until} ($\ltluntil$) temporal operators captures the class of
languages accepted by one-way alternating register automata with one register
($\araof 1$).
Figueira~\cite{Figueira12} introduces alternating register automata with one
register and two extra operations: $\guess$ and $\spread$ ($\araguessspreadshortof 1$).
Intuitively, $\guess$ and $\spread$ can be used for existential and universal quantification over future and past data values respectively.


Set-augmented finite automata, introduced in~\cite{benerjee-safa},
resemble $\rsa$s in a way that they allow to store a set of values in a register, 
however, with noticeable restrictions:
\begin{inparaenum}[(i)]
  \item set-augmented automata allow equality/disequality guard to be a singleton only, and
  \item the update function is restricted to storing the currently read data value to a 
    single register (or no update is applied as a second option). It is not possible, e.g., 
    to union the contents of two set-registers (which is a~crucial operation in
    our determinization algorithm), or to empty a register.
\end{inparaenum}
Due to these restrictions, $\rsa$s have greater expression power compared to set-augmented automata (e.g., the 
parametric language defined in \cite[Theorem 1]{benerjee-safa} can be accepted by $\rsaof{1}$ but not by 
a corresponding set-augmented automaton).

History-register automata, introduced in~\cite{tzevelekos-hra}, are a similar 
model to $\rsarem$ in a way that they also allow sets of values to be stored in registers.
The key difference is the assignment of values to registers. The only way how 
to change a value in a register is add/remove the \emph{currently processed} data value.
It is not possible e.g., to union the content of two registers (which is, as noted above, crucial in our determinization), which is possible 
in $\rsarem$ making our model more expressive. In particular, in \ifTR\cref{sec:app-hra} \else \cite{techrep} \fi
we show that $\rsarem$ are expressive at least as history-register automata and 
moreover, deterministic history-register automata are strictly less expressive 
than deterministic~$\rsarem$.

\vspace{-0.0mm}
\section{Conclusion and Future Work}\label{sec:label}
\vspace{-0.0mm}

We have introduced register set automata, a~class of automata over data words
providing an underlying formal model for efficient matching of
a~subclass of regexes with backreferences.
There are many challenges that we wish to address in the future:
\begin{inparaenum}[(i)]
  \item  improvement of our determinization algorithm to work on a~larger
    class of input $\nra$s,
  \item  explore other, more expressive, formal models that would
    allow deterministic automata models for a~large class of regexes with
    backreferences (e.g., the hard regex in
    footnote\footref{ftn:hard_regex} is beyond the power of~$\drsa$s),
  \item  develop efficient toolbox for working with $\rsa$s occurring in
    practice, and
  \cbstart
  \item  explore other approaches of how to efficiently deal with
    nondeterminism in practice, e.g., in a~similar way as considered
    in~\cite{KongYCGHM022} (which does not consider backreferences).
  \cbend
\end{inparaenum}


\vspace{-0.0mm}
\section*{Data Availability Statement}\label{sec:label}
\vspace{-0.0mm}

An environment with the tools and data used for the experimental evaluation in
the current study is available at~\cite{artifact}.

\newcommand{\ackVassal}[0]{
\noindent
\raisebox{-1pt}{\protect\includegraphics[height=8pt]{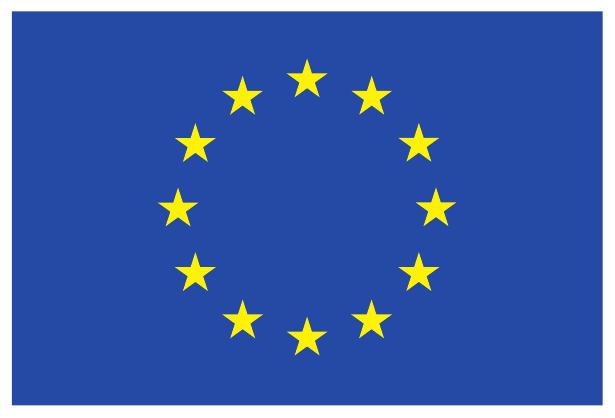}}
This work has been executed under the project VASSAL: ``Verification and Analysis for Safety and Security of Applications in Life'' funded by the European Union under Horizon Europe WIDERA Coordination and Support Action/Grant Agreement No.~101160022.\xspace
}

\begin{acks}
We thank the anonymous reviewers for their feedback that improved the quality of the paper.
This work was supported by
the Czech Science Foundation projects 26-22640S and 25-17934S and
the FIT BUT internal project FIT-S-26-9011.
\ackVassal
\end{acks}

\bibliographystyle{ACM-Reference-Format}
\bibliography{literature.bib}

\ifTR
\newpage

\appendix

\crefalias{section}{appendix}
\crefalias{subsection}{appendix}

\input{appendix.tex}

\fi

\end{document}

%% file: macros.tex
\newcommand{\ol}[1]{\textcolor{blue}{\ifmmode \text{[OL: #1]}\else [OL: #1] \fi}}
\newcommand{\sg}[1]{\textcolor{cyan}{\ifmmode \text{[SG: #1]}\else [SG: #1] \fi}}
\newcommand{\vh}[1]{\textcolor{green!50!black}{\ifmmode \text{[VH: #1]}\else [VH: #1] \fi}}
\newcommand{\jv}[1]{\textcolor{magenta}{\ifmmode \text{[JV: #1]}\else [JV: #1] \fi}}
\newcommand{\lh}[1]{\textcolor{pink}{\ifmmode \text{[LH: #1]}\else [LH: #1] \fi}}

\newcommand{\regexex}[0]{R_{\texttt{\textbackslash 3\textbackslash 2\textbackslash 1}}}

\newcommand{\lnref}[1]{Line~\ref{#1}}

\newtheorem{fact}{Fact}
\newtheorem{remark}{Remark}
\newtheorem{claim}{Claim}
\newtheorem{claimproof}{ClaimProof}

\newcommand{\claimqed}[0]{\hfill \ensuremath{\vartriangleleft}}

\newcommand{\hlbox}[1]{\fcolorbox{blue!15}{blue!15}{#1}}
\newcommand{\hlmathbox}[1]{\hlbox{\ensuremath{#1}}}

\newcommand{\nat}[0]{\mathbb{N}}
\newcommand{\natzero}[0]{\nat_{0}}
\newcommand{\natbot}[0]{\nat_\bot}
\newcommand{\set}[1]{\{#1\}}
\newcommand{\partialto}[0]{\mathrel{\rightharpoonup}}
\newcommand{\initof}[1]{[#1]}
\newcommand{\datadom}[0]{\mathbb{D}}
\newcommand{\lenof}[1]{|#1|}
\newcommand{\word}[0]{w}
\newcommand{\wordof}[1]{\word_{#1}}
\newcommand{\projsigmaof}[1]{\Sigma[#1]}
\newcommand{\projdataof}[1]{\datadom[#1]}
\newcommand{\projof}[2]{#1[#2]}

\newcommand{\img}[0]{\mathit{img}}
\newcommand{\dom}[0]{\mathit{dom}}

\newcommand{\liff}[0]{\mathrel{\Leftrightarrow}}

\newcommand{\confs}[0]{\mathcal{C}}

\newcommand{\confsenc}[0]{\mathsf{Confs}}
\newcommand{\sngconf}[0]{\mathbf{S}}
\newcommand{\post}[0]{\mathsf{Post}}

\newcommand{\minskyinc}[0]{\mathtt{inc}}
\newcommand{\minskydec}[0]{\mathtt{dec}}
\newcommand{\minskyifzero}[0]{\mathtt{ifzero}}
\newcommand{\minskyincof}[1]{\minskyinc_{#1}}
\newcommand{\minskydecof}[1]{\minskydec_{#1}}
\newcommand{\minskyifzeroof}[1]{\minskyifzero_{#1}}
\newcommand{\minsky}[0]{\mathcal{M}}

\newcommand{\pair}[2]{\langle #1, #2\rangle}

\newcommand{\aut}[0]{\mathcal{A}}
\newcommand{\autof}[1]{\aut_{#1}}

\newcommand{\regs}[0]{\mathbf{R}}
\newcommand{\regsof}[1]{\regs[#1]}
\newcommand{\inp}[0]{\mathit{in}}
\newcommand{\guard}[0]{g}
\newcommand{\guardeq}[0]{g^{=}}
\newcommand{\guardneq}[0]{g^{\neq}}
\newcommand{\guardin}[0]{g^{\in}}
\newcommand{\guardnotin}[0]{g^{\notin}}
\newcommand{\guardempty}[0]{g^{=\emptyset}}
\newcommand{\update}[0]{\mathit{up}}
\newcommand{\updateof}[1]{\update(#1)}
\newcommand{\remove}[0]{\mathit{rm}}
\newcommand{\step}[0]{\vdash}

\newcommand{\stepofusing}[2]{\mathrel{\step^{#1}_{#2}}}
\newcommand{\lang}[0]{\mathcal{L}}
\newcommand{\langof}[1]{\lang(#1)}
\newcommand{\reflang}[0]{\lang_{\mathtt{ref}}}
\newcommand{\reflangof}[1]{\reflang(#1)}
\newcommand{\Deltaempty}[0]{\Delta^{=\emptyset}}
\newcommand{\Deltarem}[0]{\Delta^{\remove}}
\newcommand{\Deltaremempty}[0]{\Delta^{\remove}_{=\emptyset}}
\newcommand{\Deltaeqreg}[0]{\Delta^{\mathit{=r}}}

\newcommand{\vectof}[1]{\mathbf{v}_{#1}}

\newcommand{\guess}[0]{\mathsf{guess}}
\newcommand{\spread}[0]{\mathsf{spread}}

\newcommand{\dra}[0]{\mathrm{DRA}}
\newcommand{\draof}[1]{\dra_{#1}}
\newcommand{\nra}[0]{\mathrm{NRA}}
\newcommand{\nraof}[1]{\nra_{#1}}
\newcommand{\ura}[0]{\mathrm{URA}}
\newcommand{\uraof}[1]{\ura_{#1}}
\newcommand{\nraeq}[0]{\nra^{=}}
\newcommand{\uraeq}[0]{\ura^{=}}
\newcommand{\draeq}[0]{\dra^{=}}
\newcommand{\nraeqof}[1]{\nraeq_{#1}}
\newcommand{\uraeqof}[1]{\uraeq_{#1}}
\newcommand{\draeqof}[1]{\draeq_{#1}}
\newcommand{\ara}[0]{\mathrm{ARA}}
\newcommand{\araof}[1]{\ara_{#1}}

\newcommand{\araguessspread}[0]{\ara(\guess, \spread)}

\newcommand{\araguessspreadshort}[0]{\ara(\mathsf{g,s})}
\newcommand{\araguessspreadshortof}[1]{\araof{#1}(\mathsf{g,s})}
\newcommand{\rsa}[0]{\mathrm{RSA}}
\newcommand{\rsaof}[1]{\rsa_{#1}}

\newcommand{\drsa}[0]{\mathrm{DRSA}}
\newcommand{\drsaof}[1]{\drsa_{#1}}
\newcommand{\rsaempty}[0]{\rsa^{=\emptyset}}

\newcommand{\rsaeqreg}[0]{\rsa^{=r}}

\newcommand{\rsaremempty}[0]{\rsa^{\mathit{rm}}_{=\emptyset}}
\newcommand{\rsarem}[0]{\rsa^{\mathit{rm}}}

\newcommand{\booleanof}[1]{\mathcal{B}(#1)}
\newcommand{\pa}[0]{\mathrm{PA}}

\newcommand{\fomega}[0]{\mathbf{F}_{\omega}}
\newcommand{\pspace}[0]{\textsc{Pspace}}
\newcommand{\np}[0]{\textsc{NP}}

\newcommand{\rep}[0]{\mathit{repeat}}
\newcommand{\langexrep}[0]{L_{\exists \rep}}
\newcommand{\langnegexrep}[0]{L_{\neg\exists \rep}}
\newcommand{\langexnegexrep}[0]{L_{\exists,\neg\exists \rep}}
\newcommand{\langallrep}[0]{L_{\forall \rep}}
\newcommand{\langnegallrep}[0]{L_{\neg\forall \rep}}
\newcommand{\langdisjoint}[0]{L_{\mathit{disjoint}}}

\newcommand{\langallaexb}[0]{L_{\forall a\exists b}}

\newcommand{\langexanob}[0]{L_{\exists\mathit{a\text{-}no\text{-}b}}}
\newcommand{\langnegexanob}[0]{L_{\neg\exists\mathit{a\text{-}no\text{-}b}}}

\newcommand{\langtransact}[0]{\lang_\textit{transact}}

\newcommand{\ltl}[0]{\mathrm{LTL}}
\newcommand{\ltlfreeze}[0]{\ltl^{\downarrow}}
\newcommand{\ltlfreezeofn}[2]{\ltlfreeze_{#1}(#2)}
\newcommand{\ltlnext}[0]{\mathtt{X}}
\newcommand{\ltluntil}[0]{\mathtt{U}}

\newcommand{\expl}[1]{\lbag#1\rbag}

\newcommand{\cmark}{\ding{51}}
\newcommand{\xmark}{\ding{55}}
\newcommand{\tabyes}[0]{\textcolor{green!70!black}{\cmark}}
\newcommand{\tabno}[0]{\textcolor{red!70!black}{\xmark}}

\newcommand{\net}[0]{\mathcal{N}}
\newcommand{\netof}[1]{\net_{#1}}
\newcommand{\tuple}[1]{\langle #1 \rangle}
\newcommand{\intrans}[0]{\mathit{In}}
\newcommand{\outrans}[0]{\mathit{Out}}
\newcommand{\transfer}[0]{\mathit{Transfer}}

\newcommand{\unitof}[1]{\mathbf{u}_{#1}}
\newcommand{\identity}[0]{\mathbf{id}}
\newcommand{\fire}[1]{[#1\rangle}

\newcommand{\region}[0]{\mathit{rgn}}
\newcommand{\regionof}[1]{\region_{#1}}
\newcommand{\cmplof}[1]{\overline{#1\vphantom{t}}}
\newcommand{\dst}[0]{\mathit{dst}}
\newcommand{\posit}[0]{\mathit{pst}}
\newcommand{\positsop}[0]{\posit_\Sigma}
\newcommand{\positpos}[0]{\posit_\Pi}
\newcommand{\negat}[0]{\mathit{ngt}}
\newcommand{\krugel}[0]{\coprod}

\newcommand{\init}[0]{\mathit{init}}
\newcommand{\qinit}[0]{q_{\init}}
\newcommand{\qmain}[0]{q_{\mathit{main}}}
\newcommand{\fin}[0]{\mathit{fin}}
\newcommand{\qfin}[0]{q_{\fin}}
\newcommand{\lossyremgadget}[0]{\textsc{LossyRm}}
\newcommand{\lossyremgadgetof}[1]{\lossyremgadget(#1)}

\newcommand{\movegadget}[0]{\textsc{Move}}
\newcommand{\movegadgetof}[1]{\movegadget(#1)}
\newcommand{\newtokengadget}[0]{\textsc{NewToken}}
\newcommand{\newtokengadgetof}[1]{\newtokengadget(#1)}
\newcommand{\regof}[1]{r_{#1}}
\newcommand{\tmp}[0]{\mathit{tmp}}
\newcommand{\concat}[0]{\mathbin{\cdot}}
\newcommand{\iterof}[2]{#1^{[#2]}}
\newcommand{\unprimeof}[1]{\mathit{unprime}(#1)}

\newcommand{\emptygadget}[0]{\textsc{Empty}}
\newcommand{\emptygadgetof}[1]{\emptygadget(#1)}
\newcommand{\tokenremgadget}[0]{\textsc{TokenRem}}
\newcommand{\tokenremgadgetof}[1]{\tokenremgadget(#1)}

\newcommand{\emptyeqgadget}[0]{\textsc{EmptyEq}}
\newcommand{\emptyeqgadgetof}[1]{\emptyeqgadget(#1)}
\newcommand{\nonlossyrmgadget}[0]{\textsc{NonLossyRm}}
\newcommand{\nonlossyrmgadgetof}[1]{\nonlossyrmgadget(#1)}

\newcommand{\greycell}[0]{\cellcolor{black!15}}

\newcommand{\pictrans}[3]{\footnotesize\begin{array}{c} \framebox{$#1$} \begin{array}{l} #2 \end{array}\\ \begin{array}{l}\hline #3 \end{array}\end{array}}

\tikzstyle{state}=[draw,circle,inner sep=1mm,minimum width=6mm]

\makeatletter
\DeclareRobustCommand{\shortto}{%
  \mathrel{\mathpalette\short@to\relax}%
}

\DeclareRobustCommand{\shortminus}{%
  \mathrel{\mathpalette\short@minus\relax}%
}

\newcommand{\short@to}[2]{%
  \mkern2mu
  \clipbox{{.5\width} 0 0 0}{$\m@th#1\vphantom{+}{\rightarrow}$}%
}

\newcommand{\short@minus}[2]{%
  \mkern2mu
  \clipbox{{.5\width} 0 0 0}{$\m@th#1\vphantom{+}{-}$}%
}
\makeatother

\newcommand{\transeps}[2]{\trans{#1} a {\emptyset, \emptyset, \emptyset}{#2}}

\tikzset{st/.style={font=\ttfamily,shape=rectangle,rounded corners=.5em,draw=black,fill=gray!10,inner xsep=.3em,inner ysep=0em,text height=2.3ex,text depth=1.0ex}}
\newcommand{\translab}[1]{\text{\small {\protect\tikz[baseline,node distance=2mm]{%
  \protect\node[st] (nd) at (0,.64ex) {\hspace{.1em}\texttt{\upshape\strut$#1$}\hspace{.1em}\strut};%
  \protect\node[inner sep=0mm, left=of nd.west] (lhs){};%
  \protect\node[inner sep=0mm, right=of nd.east] (rhs){};%
  \protect\draw (lhs) edge (nd);%
  \protect\draw[>=stealth',->] (nd) edge (rhs);%
  }}}}
\newcommand{\trans}[4]{#1 \translab{#2 \mid #3} #4}
\newcommand{\transsimple}[3]{#1 \translab{#2} #3}

\newcommand\tool{\textsf{rsamatch}\xspace}

\newcommand\grep{\textsf{grep}\xspace}
\newcommand\re{\textsf{re}\xspace}
\newcommand{\retwo}{\textsf{RE2}\xspace}
\newcommand\pcre{\textsf{pcre2}\xspace}
\newcommand\js{\textsf{js}\xspace}
\newcommand\java{\textsf{java}\xspace}
\newcommand\nettool{\textsf{.net}\xspace}
\newcommand{\hyperscan}{\textsf{HyperScan}\xspace}
\newcommand{\rengar}[0]{\textsc{Rengar}\xspace}

\newcommand{\RGBcircle}[1]{\ensuremath{{\color[RGB]{#1}\bullet}}}

\newcommand{\redef}{\mathcal{R}}
\newcommand{\backref}[1]{\textbackslash#1}
\newcommand{\capgroup}[1]{(#1)}
\newcommand{\semof}[1]{\llbracket #1 \rrbracket}
\newcommand{\bproj}[0]{\mathsf{interp}}
\newcommand{\capg}[0]{\mathsf{match}}
\newcommand{\symin}[2]{\langle #1, #2, \mathtt{in} \rangle}
\newcommand{\symref}[1]{\langle #1, \mathtt{ref} \rangle}
\newcommand{\Sigmaref}[0]{\Sigma_\mathtt{ref}}

\newcommand{\regexra}[0]{\mathsf{rebr2ra}}
\newcommand{\dwlang}[1]{#1^{\bullet}}
\newcommand{\dword}[0]{\mathsf{ord}}

\newcommand{\deriv}[0]{\partial}
\newcommand{\derivof}[2]{\deriv_{#1}(#2)}
\newcommand{\dertuple}[3]{\langle #1, #2, #3\rangle}

\newcommand{\rewb}[0]{REBR\xspace}
\newcommand{\rewbs}[0]{REBRs\xspace}

\newcommand{\regex}[0]{\mathit{RE}}
\newcommand{\nulla}[0]{\mathit{nullable}}
\newcommand{\nullaof}[1]{\nulla(#1)}

\newcommand{\mytrue}[0]{\mathit{true}}
\newcommand{\myfalse}[0]{\mathit{false}}

\newcommand{\bigO}[0]{\mathcal{O}}

\newcommand{\subst}[2]{[#1 \mapsto #2]}

%% file: figs/ra_example.tikz
\begin{tikzpicture}
  \tikzset{
    ->, 
    >=stealth',
    initial text=$ $, 
    node distance=25mm,
  }

  \node[initial,state] (q) {$q$};
  \node[state] (s) [right of=q] {$s$};
  \node[state, accepting] (t) [right of=s] {$t$};


 \draw (q) edge[loop above] node[yshift=-4mm] {$\pictrans {a} {} {}$} (q);
 \draw (q) edge node[above,yshift=-1mm] {$\pictrans a {} {r \gets \inp}$} (s);
 \draw (s) edge[loop above] node[yshift=-4mm] {$\pictrans {a} {\inp \neq r} {}$} (s);
 \draw (s) edge node[above,yshift=-4mm] {$\pictrans {a} {\inp = r} {}$} (t);
 \draw (t) edge[loop above] node[yshift=-4mm] {$\pictrans {a} {} {}$} (t);

\end{tikzpicture}

%% file: figs/ex_rep.tikz
\begin{tikzpicture}
  \tikzset{
    ->, 
    >=stealth',
    initial text=$ $, 
    node distance=25mm,
  }

  \node[initial,state] (q) {$q$};
  \node[state, accepting] (s) [right of=q] {$s$};


 \draw (q) edge[loop above] node[yshift=-1mm] {$\pictrans a {\inp \notin r} {r \gets r \cup \{\inp\}}$} (q);
 \draw (q) edge node[above,yshift=-4mm] {$\pictrans a {\inp \in r} {}$} (s);
 \draw (s) edge[loop above] node[yshift=-4mm] {$\pictrans a {} {}$} (s);

\end{tikzpicture}

%% file: figs/neg_ex_rep.tikz
\begin{tikzpicture}
  \tikzset{
    ->, 
    >=stealth',
    initial text=$ $, 
    node distance=30mm,
  }

  \node[initial,state, accepting] (q) {$q$};


 \draw (q) edge[loop above] node[yshift=-1mm] {$\pictrans a {\inp \notin r} {r \gets r \cup \{\inp\}}$} (q);

\end{tikzpicture}

%% file: figs/neg_all_rep.tikz
\begin{tikzpicture}
  \tikzset{
    ->, 
    >=stealth',
    initial text=$ $, 
    node distance=30mm,
  }

  \node[initial,state] (q) {$q$};
  \node[state, accepting] (s) [right of=q] {$s$};


 \draw (q) edge[loop above] node[yshift=-1mm] {$\pictrans a {} {r_1 \gets r_1 \cup \{\inp\}}$} (q);
 \draw (q) edge node[below] {$\pictrans a {\inp \notin r_1} {r_1 \gets\{\inp\}}$} (s);
 \draw (s) edge[loop above] node[yshift=-4mm] {$\pictrans a {\inp \notin r_1} {}$} (s);

\end{tikzpicture}

%% file: figs/nra-problem-diseq.tikz
\begin{tikzpicture}
  \tikzset{
    ->, 
    >=stealth',
    initial text=$ $, 
    node distance=25mm,
  }

  \node[initial,state] (q) {$q$};
  \node[state] (s) [right of=q] {$s$};
  \node[state, accepting] (t) [right of=s] {$t$};


 \draw (q) edge[loop above] node[yshift=0mm] {$\pictrans {a} {} {r_q \gets \inp}$} (q);
 \draw (q) edge[loop below] node[yshift=0mm] {$\pictrans {a} {} {r_q \gets r_q}$} (q);
 \draw (q) edge node[above,yshift=-1mm] {$\pictrans a {\inp \neq r_q} {r_s \gets r_q}$} (s);
 \draw (s) edge node[above,yshift=-4mm] {$\pictrans a {\inp = r_s} {}$} (t);

\end{tikzpicture}

%% file: figs/rsa-problem-diseq.tikz
\begin{tikzpicture}
  \tikzset{
    ->, 
    >=stealth',
    initial text=$ $, 
    node distance=25mm,
  }

  \node[initial,draw,rectangle,rounded corners] (q) {$\set q$};
  \node[draw,rectangle,rounded corners,xshift=10mm] (qs) [right of=q] {$\set{q,s}$};
  \node[] (t) [right of=qs,node distance=15mm] {$\cdots$};


 \draw (q) edge[loop above] node[xshift=-8mm,yshift=0mm] {$\pictrans {a} {\inp \in r_q} {r_q \gets r_q \cup \set \inp}$} (q);
 \draw (q) edge node[above,yshift=-1mm] {$\pictrans a {\inp \notin r_q} {r_q \gets r_q \cup \set \inp \\ r_s \gets r_q}$} (qs);
 \draw (qs) edge[loop above] node[above right,xshift=-2mm,yshift=-2mm] {$\pictrans {a} {\inp \notin r_q \\ \inp \notin r_s} {r_q \gets r_q \cup \set \inp \\ r_s \gets r_q}$} (qs);
 \draw (qs) edge[bend left] node[below,yshift=-1mm] {$\pictrans a {\inp \in r_q \\ \inp \notin r_s} {r_q \gets r_q \cup \set \inp}$} (q);

 \draw (qs) edge (t);

\end{tikzpicture}

%% file: figs/nra_problem_cartesian.tikz
\begin{tikzpicture}
  \tikzset{
    ->, 
    >=stealth',
    initial text=$ $, 
    node distance=25mm,
  }

  \node[initial,state] (q) {$q$};
  \node[state] (s) [right of=q] {$s$};
  \node[state] (t) [right of=s] {$t$};
  \node[state] (u) [right of=t] {$u$};
  \node[state, accepting] (v) [right of=u] {$v$};


 \draw (q) edge[loop above] node[yshift=-4mm] {$\pictrans {a} {} {}$} (q);
 \draw (q) edge node[above,yshift=-1mm] {$\pictrans a {} {r_1 \gets \inp}$} (s);
 \draw (s) edge node[above,yshift=-1mm] {$\pictrans {a} {} {r_2 \gets r_1\\ r_3 \gets \inp}$} (t);
 \draw (t) edge[loop above] node[above,yshift=-1mm] {$\pictrans {a} {} {r_2 \gets r_2 \\ r_3 \gets r_3}$} (g);
 \draw (t) edge node[above,yshift=-1mm] {$\pictrans {a} {\inp = r_2} {r_4 \gets r_3}$} (u);
 \draw (u) edge node[above,yshift=-4mm] {$\pictrans {a} {\inp = r_4} {}$} (v);
 \draw (v) edge[loop above] node[above,yshift=-4mm] {$\pictrans {a} {} {}$} (v);

\end{tikzpicture}

%% file: figs/nra-problem-write.tikz
\begin{tikzpicture}
  \tikzset{
    ->, 
    >=stealth',
    initial text=$ $, 
    node distance=25mm,
  }

  \node[initial,state] (q) {$q$};
  \node[state] (s) [right of=q] {$s$};
  \node[state, accepting] (t) [right of=s] {$t$};


 \draw (q) edge[loop above] node[yshift=0mm] {$\pictrans {a} {} {r_q \gets \inp}$} (q);
 \draw (q) edge[loop below] node[yshift=0mm] {$\pictrans {a} {} {r_q \gets r_q}$} (q);
 \draw (q) edge node[above,yshift=-1mm] {$\pictrans b {\inp = r_q} {r_s \gets r_q}$} (s);
 \draw (s) edge node[above,yshift=-4mm] {$\pictrans {b} {\inp = r_s} {}$} (t);

\end{tikzpicture}

%% file: figs/rsa-problem-write.tikz
\begin{tikzpicture}
  \tikzset{
    ->, 
    >=stealth',
    initial text=$ $, 
    node distance=25mm,
    inner sep=0.3mm,
  }

  \node[initial,state] (q) {$\set{q}$};
  \node[state] (s) [right of=q] {$\set{s}$};
  \node[state, accepting] (t) [right of=s] {$\set{t}$};

  \draw (q) edge[loop above] node[yshift=1mm] {$\pictrans {a} {} {r_q \gets r_q \cup \set{\inp}}$} (q);
  \draw (q) edge node[above,yshift=0mm] {$\pictrans b {\inp \in r_q} {r_s \gets r_q}$} (s);
  \draw (s) edge node[above,yshift=-3mm] {$\pictrans {b} {\inp \in r_s} {}$} (t);

\end{tikzpicture}

%% file: figs/preproc_ra_example.tikz
\begin{tikzpicture}
  \tikzset{
    ->, 
    >=stealth',
    initial text=$ $, 
    node distance=30mm,
  }
  \useasboundingbox (-8mm,13mm) rectangle (66mm,-13mm);

  \node[state,initial] (q) {$q$};
  \node[state] (s) [right of=q] {$s$};
  \node[state,accepting](t) [right of=s] {$t$};
  \path[->]

  (q) edge[loop above] node[yshift=-3mm]  {$\pictrans{a}{}{}$} (q)
  (q) edge[above] node {$\pictrans{a}{}{r_1 \gets \inp}$} (s)
  (s) edge[above] node {$\pictrans{a}{}{r_2 \gets \inp}$} (t);
\end{tikzpicture}

%% file: figs/preproc_svra_example.tikz
\begin{tikzpicture}
  \tikzset{
    ->, 
    >=stealth',
    initial text=$ $, 
    node distance=30mm,
  }

  \useasboundingbox (-8mm,13mm) rectangle (66mm,-13mm);
  \node[state,initial] (q) {$q$};
  \node[state] (s) [right of=q] {$s$};
  \node[state,accepting](t_2) [right of=s, yshift=7mm] {$t_{r_1=r_2}$};
  \node[state,accepting](t) [right of=s, yshift=-7mm] {$t_{r_1 \neq r_2}$};
  \path[->]

  (q) edge[loop above] node[yshift=-3mm]  {$\pictrans{a}{}{}$} (q)
  (q) edge[above] node {$\pictrans{a}{}{r_1 \gets \inp}$} (s)
  (s) edge[below] node {$\pictrans{a}{in \neq r_1}{r_2 \gets \inp}$} (t)
  (s) edge[above] node[yshift=-3mm] {$\pictrans{a}{in = r_1}{}$} (t_2);
\end{tikzpicture}
%
%

%% file: table-redos-pldi26.tex
\begin{booktabs}{colspec={lrrrrrrrr|rrr},rowsep=0pt,column{8}={Gray!20},row{2}={GreenYellow},row{1}={font=\bfseries,halign=c}}
\toprule
  tool     & \ensuremath{<} 1\,s & 1--5\,s & 5--10\,s &   10--50\,s &   50--100\,s &   TOs &  \ensuremath{\boldsymbol{>}} 1\,s &  errors & mean &   median &   std.\ dev \\
\midrule
 \tool      &   1,244 &        2 &         0 &          0 &           0 &              0 & 2   & \SetCell{c}  ---& 0.299 & 0.16     & 0.303 \\
 \grep      &   1,131 &        6 &         0 &         16 &           4 &              15 & 41   &    74         & 0.664 & 0.01     & 5.973 \\
 \re        &   1,139 &       36 &        14 &         11 &           3 &              43 & 107   &     0        & 0.532 & 0.01     & 4.416 \\
 \pcre      &   1,001 &       14 &         4 &          4 &           0 &               3 & 42   &   203         & 0.144 & 0.01     & 1.219 \\
 \js        &     660 &        8 &         1 &         11 &           1 &              33 & 54   &   532         & 0.531 & 0.06     & 3.702 \\
 \java      &     595 &       18 &         5 &          7 &           5 &              19 & 54   &   597         & 0.996 & 0.05     & 6.438 \\
 \nettool   &     629 &       24 &        13 &          9 &           3 &              40 & 89   &   528         & 0.813 & 0.04     & 5.165 \\
\bottomrule
\end{booktabs}

%% file: appendix.tex
\vspace{-0.0mm}
\section{Proofs for \cref{sec:properties}}\label{sec:label}
\vspace{-0.0mm}

\input{proofs-properties.tex}



\input{app-expressivity.tex}
\input{app_proof_det.tex}

\vspace{-0.0mm}
\section{Proof of \cref{thm:nra1-to-drsa}}\label{sec:proof-determinisability}
\vspace{-0.0mm}

\begin{figure}[t]
  \centering
  \input{figs/ra_nodet.tikz}
  \caption{Counterexample for determinization of $\nraeqof 1$ without localization.}
  \label{fig:ra_nodet}
\end{figure}

\thmNraOneToDrsa*

\begin{proof}
  Consider $\nraeqof 1$ $\aut = (Q, \regs, \Delta, I, F)$ with a single register $r$ (for $\aut$ having no register 
  the theorem coincides with determinisation of NFAs).
  Further, consider $\aut_{\bullet} = (Q, \regs_\bullet, \Delta_\bullet, I, F)$ to
  be the register localized automaton using the procedure from \cref{sec:reg-locality}. 
  Note that the localization step is necessary for the proof as it is shown on the following example. Consider NRA 
  from \cref{fig:ra_nodet} over singleton alphabet 
  with a single register $r$ as input of the determinisation (single valued is trivially satisfied). When generating 
  successors of the macrostate $\{ 1, 2 \}$ with $g = \emptyset$, we 
  have $S'= \{ 3,5 \}$ and $\mathit{op}_r = \{ r, \inp \}$. However, for the state $5$, there is 
  no transition with the target $5$ and update $\update(r) = r$. Therefore, the condition \cref{ln:abort} is 
  satisfied and the algorithm returns $\bot$.
  
  Continuing with the proof, we first prove the auxiliary claims that are used in the proof of the main theorem:
  \begin{claim}\label{claim:loc-reg}
   $\regs_\bullet[q] \subseteq v_q(\regs)$ where $v_q$ is defined in \cref{sec:reg-locality}.
  \end{claim}
  \begin{claimproof}
    Follows directly from the definition of $\update_\bullet$ in the localization procedure.
  \end{claimproof}

  \begin{claim}\label{claim:single-val}
    $\aut_{\bullet}$ is copyless.
  \end{claim}
  \begin{claimproof}
    Assume that there is a configuration $(s,f)$ s.t. $f(r_p) = f(r_q) \neq\bot$ and $r_p \neq r_q$. 
    From the localization procedure we have that for registers that are not active 
    in the particular state, its value is $\bot$ in every reachable configuration.
    Hence, from \cref{claim:loc-reg} we have that the only possibility when 
    $f(r_p) = f(r_q) \neq\bot$ could happen is when $r_p = r_q = r_s$, which is a contradiction.
  \end{claimproof}

  For the proof of the main theorem, we apply determinization on the localized automaton $\aut_{\bullet}$.
  From \cref{claim:single-val} we have that $\aut_{\bullet}$ is copyless.
  We prove that the algorithm does not return $\bot$. 
  Since the automaton has no disequality guards, the only way how \cref{alg:det} could fail is at \lnref{ln:abort}.
  Due to the fact that $\regs_\bullet = \{ r_q \mid q\in Q, \regs = \{ r \} \}$, 
  in the loop on \cref{ln:aggregate} we can assume that $r_i = r_q$ for some $q \in Q$.
  Further, from the definition of register localization, we have that if
  $x \in \mathit{op}_{r_q}$, then $x$ was added to $\mathit{op}_{r_q}$ using a transition with the 
  target state $q$ on \cref{ln:trans-iter}. Now consider the approximation-test loop on \cref{ln:approx_test_loop}.
  From \cref{claim:loc-reg} we have that $\regs_\bullet[q'] \subseteq \{ r_{q'} \}$. If this 
  set is empty, we are done. Otherwise, consider some $x \in \mathit{op}_{r_{q'}}$. From the reasoning above,
  we have that $x$ was added to $\mathit{op}_{r_{q'}}$ using a transition with the target state $q'$ satisfying 
  $\update(r_{q'}) = x$. Therefore, the condition on \cref{ln:abort} is always false.
\end{proof}

\vspace{-0.0mm}
\section{Other Examples}\label{sec:app-examples}
\vspace{-0.0mm}

\begin{example}
Consider the language $\langdisjoint$~\cite[Example~2.5]{NevenSV04} of words $w
= uv$ where~$u$ and~$v$ are non-empty and the data values in~$u$ and~$v$ are
disjoint, i.e.,
\begin{equation*}
  \langdisjoint = \big\{w \mid \exists u,v \in \Sigma \times \datadom\colon w = uv \land u \neq \epsilon
  \land v \neq \epsilon \land
  \forall i,j\colon \projdataof{u_i} \neq \projdataof{v_j}\big\}
\end{equation*}
For instance, the language contains the word $ddee$, but does not contain the
word $dede$ (we only consider the data values here).
An $\rsaof 1$ accepting this language looks as follows:
\begin{center}
\input{figs/disjoint.tikz}
\end{center}
Intuitively, the $\rsa$ stays in state~$q$, accumulating the so-far seen values
in register~$r$, and at some point, it nondeterministically moves to
state~$s$, where it checks that no data value from the first part of the word
appears in the second part.

We note that~$\langdisjoint$ can be accepted neither by any~$\ara$ (intuitively,
the different threads in an $\ara$ cannot synchronize on the transition from
the first part to the second part of the word), nor by any $\araguessspread$
(which are strictly more expressive than $\ara$s).
It also cannot be accepted by a~$\drsa$.
\qed
\end{example}

\begin{example}\label{ex:cartesian-overapp}
	Consider the NRA shown in \cref{fig:postproc_example}. It is in
	the register-local form, having multiple copies for registers $r_1$, $r_2$
	(updates of the type $r \gets \bot$ are implicit here). During determinisation
	of this NRA, when generating transitions from the macrostate $\sigma = (\set{t,u},
	\set{r_1^s \colon 0, r_1^t \colon 1, r_1^u \colon 1, r_2^u \colon 1})$,
	registers $r_1^u, r_2^u$ get the updates $r_1^u \gets r_1^t \cup r_1^u$
	and $r_2^u \gets r_2^u \cup \set{\inp}$. This would cause \cref{alg:det}
	to return $\bot$ on Line~\ref{ln:abort} when attempting to find one of the update
	combinations (i) $r_1^u \gets r_1^t, r_2^u \gets r_2^u$,
	(ii) $r_1^u \gets r_1^u, r_2^u \gets \inp$, neither of which exists
	in the original NRA. However, in the macrostate $\sigma$, registers $r_1^t$
	and $r_1^u$ hold the same value, and update combination (i) is thus equivalent
	to $r_1^u \gets r_1^u, r_2^u \gets r_2^u$, and update combination (ii) is
	equivalent to $r_1^u \gets r_1^t, r_2^u \gets \inp$. Both of these combinations
	can be found in the original NRA and, therefore, the detected
	overapproximation was a false~positive.
\end{example}

\begin{figure}[t]
	\centering
	\input{figs/postproc_example.tikz}
	\caption{A register-local NRA that falsely triggers the
	overapproximation detection of \cref{alg:det}}\label{fig:postproc_example}
\end{figure}

\vspace{-0.0mm}
\section{Extensions of Register Set Automata}\label{sec:app_extensions}
\vspace{-0.0mm}

\input{app_extensions.tex}

\vspace{-0.0mm}
\section{Relationship with HRAs}\label{sec:app-hra}
\vspace{-0.0mm}

In this section, we address the relationship with history register automata~\cite{tzevelekos-hra}.
First, we give a definition of a HRA. We slightly adjust the definition 
of~\cite{tzevelekos-hra} in order to better match the settings of the paper. 
In particular, we modify the HRA to run over data words, we do not consider 
initial assignment of registers, and use only set-registers (instead of RA-like 
registers and histories as it is defined in~\cite{tzevelekos-hra}). Note that 
the absence of RA-like registers can be simulated by set-registers with an additional 
removal to keep the set singleton.

A \emph{history register automaton} (HRA)
over $\datadom \times \Sigma$ is a tuple
$\aut_H = (Q, \regs, \Delta, I, F)$, where $Q, \regs, I, F$ is
the same as in RsAs, and $\Delta \subseteq Q \times \Sigma \times ((2^\regs \times
2^\regs) \cup 2^\regs) \times Q$. In this model, there are two types of transitions:
\begin{enumerate}
  \item Updating transitions $t_{up} = \trans q {a} {R_g, R_{up}} {q'}$.
    $\aut_H$ can use such a $t_{up}$ if the current state is $q$, 
    the current $\Sigma$-symbol is $a$ and $\inp$ is in all the registers in $R_g$
    and in none of the registers in $\regs \setminus R_g$. The contents of 
    the registers are then updated so that $\inp$ is in all the registers in 
    $R_{up}$ and in none of the registers in $\regs \setminus R_{up}$.
  \item Resetting transitions, $t_{res} = \transsimple q {R_{clr}} {q'}$.
    $\aut_H$ can use such a $t_{res}$ if the $q$ is the current state.
    They do not consume any input symbols (they are $\varepsilon$-transitions).
    The register contents are then updated so that all the registers in $R_{clr}$
    are emptied (the other registers are left as they are).  
\end{enumerate}
Definitions for a \emph{configuration} and an \emph{initial configuration}
are the same as for an RsA. Let $c_1 = (q_1, f_1), c_2 = (q_2, f_2)$
be two configurations of $\aut_H$. We say that $\aut_H$ can make a step
from $c_1$ to $c_2$ over $(a,d)$ using an update transition
$t_{up} = \trans {q_1} {a} {R_g, R_{up}} {q_2}$,
denoted as $c_1 \stepofusing{(a, d)}{t_{up}} c_2$, iff
\begin{enumerate}
  \item $\forall r \in R_g\colon d \in f_1(r) $, 
  \item $\forall r \in \regs \setminus R_g \colon d \notin f_1(r) $, 
  \item $\forall r \in R_{up} \colon f_2(r) = f_1(r) \cup \set{d} $, and
  \item $\forall r \in \regs \setminus R_{up} \colon f_2(r) = f_1(r)
    \setminus \set{d} $. 
\end{enumerate}
$\aut_H$ can also make a step from $c_1$ to $c_2$ using a reset transition
$t_{res} = \transsimple {q_1} {R_{clr}} {q_2}$, consuming no input, denoted as
$c_1 \stepofusing{\varepsilon}{t_{res}} c_2$, iff
\begin{enumerate}
  \item $\forall r \in R_{clr} \colon f_2(r) = \emptyset $, and
  \item $\forall r \in \regs \setminus R_{clr} \colon f_2(r) = f_1(r) $.
\end{enumerate}
The run on $\aut_H$ and the language accepted by $\aut_H$ have the same definitions as for RsAs.
We say that $\aut_H$ is \emph{deterministic} if for all states $q_1 \in Q$
it holds that either
\begin{enumerate}
  \item there is only one transition originating in $q_1$, or
  \item there are no reset transitions originating in $q_1$, and for all $a \in \Sigma, q_2 \in Q$
  	it holds that there are no distinct transitions $\trans {q_1} a {R^1_{g}, R^1_{up}} {q_2}, 
    \trans {q_1} a {R^2_{g}, R^2_{up}} {q_2} \in \Delta$
    such that $R^1_{g} = R^2_{g}$.
\end{enumerate}

Since the HRAs are using removal semantics, it is natural to compare them with RsA$^{rm}$.

\begin{proposition}\label{proposition:hra_subseteq_rsa}
  HRA $\subseteq$ RsA$^{rm}$
\end{proposition}
\begin{proof}
	We show that any HRA $\aut_H = (Q_H, \regs_H, \Delta_H, I_H, F_H)$
	can be converted to an RsA$^{rm}$ $\aut_R = (Q_R, \regs_R, \Delta_R, I_R, F_R)$ such that
	$\langof{\aut_H} = \langof{\aut_R}$.
	We keep the states, the initial states, and registers the same, i.e., $Q_H=Q_R, I_H = I_R,
	\regs_H = \regs_R$. We convert all update
	transitions of $\aut_H$ $\trans {q_1} {a} {R_g, R_{up}} {q_2} \in \Delta_H$ to
	$\trans {q_1} a {\guardin, \guardnotin, \update, \remove} {q_2} \in \Delta_R$, where
	$\guardin = R_g, \guardnotin = \regs_H \setminus R_g, \remove = R_g \setminus R_{up}$, and
	for all $r \in \regs_H$ if $r \in R_{up}$, then $\updateof{r} = \set{r, \inp}$, otherwise
	$\updateof{r} = \set{r}$.

	To deal with reset transitions, we find will sequences of transitions in $\aut_H$,
	such that the last transition is an update transition, and all the previous transitions
	are reset transitions. We also make sure the transitions form a path in $\aut_H$,
	and that no transition is in the sequence more than once. We then create a transition in
	$\aut_R$ that executes the transition sequence in one step.

	We do so by first finding all sequences of transitions
	$t_1 = \transsimple {q_1} {R^1_{clr}} {q_2}, t_2 =\transsimple {q_2} {R^2_{clr}} {q_3},
	\ldots, t_n = \trans {q_n} a {R_g, R_{up}} {q_{n+1}}$, where
	$\forall 1 \leq i \leq n \colon t_i \in \Delta_H$, and $\forall 1\leq i < j \leq n \colon
	t_i \neq t_j$. For each such sequence, where
	$R_g \cap \bigcup_{i=1}^{n-1} R^i_{clr} = \emptyset$, we add a transition
	$t = \trans {q_1} a {\guardin, \guardnotin, \update, \remove} {q_n}$ to $\Delta_R$, where
	$\guardin = R_g, \guardnotin = \regs \setminus R_g, \remove = R_g \setminus R_{up}$, and
	for all $r \in \regs_H$ if $r \in \bigcup_{i=1}^{n-1} R^i_{clr}$, then $\updateof{r} = y$,
	otherwise $\updateof{r} = \set{r} \cup y$, where $y = \set{\inp}$ if $r \in R_{up}$ and
	$y = \emptyset$ otherwise. 

	The final states $F_R$ of $\aut_R$ will then be the states $F_H$ along with states that can
	reach any state $q_f \in F_H$ by a sequence of reset transitions.
\end{proof}

Notice that the same can be done for deterministic HRAs without introducing any non-determinism,
which will be useful when relating the models' deterministic variants.
\begin{corollary}\label{corollary:dhra_subseteq_drsa}
  DHRA $\subseteq$ DRsA$^{rm}$ 
\end{corollary}

The other direction of \cref{proposition:hra_subseteq_rsa} was left as an open problem.

\subsection{Relating DRsAs and DHRAs}

To compare the expressive power of DRsA$^{rm}$s and DHRAs we will use the language
$\langtransact$ of
the DRsA$^{rm}$ $\aut$ shown in \cref{fig:DRsA_no_DHRA}. $\langtransact$ is a language over
the alphabet $\Sigma = \{a, co, rb, \#\}$ and the data domain $\datadom$.
The semantics of $\langtransact$ is as follows --- data values of the symbol $a$ are
\emph{committed} if the next non-$a$ symbol is $co$ or \emph{rolled back}
if the next non-$a$ symbol is $rb$. Words are only part of $\langtransact$
if their last symbol is $\#$ and its data value was committed earlier in the word
($\#$ may only appear as the last symbol of a word).

\begin{figure}
	\centering
	\input{figs/DRsA_no_DHRA.tikz}
	\caption{A DRsA$^{rm}$ $\aut$ accepting the language $\langtransact$}\label{fig:DRsA_no_DHRA}
\end{figure}

\begin{lemma}\label{lemma:drsa_notsubseteq_dhra}
	DRsA$^{rm}$ $\not\subseteq$ DHRA
\end{lemma}

\begin{proof}\sloppy
	Let us take the language $\langtransact$ as defined above and assume there is
	a deterministic HRA $\aut_H$ with $n$ registers accepting it.
	Now let us look at the word $w = \pair a
	{d_1} \pair {co} {\cdot} \pair a {d_2} \pair {co} {\cdot} \ldots
	\pair a {d_n} \pair {co} {\cdot} \pair a {d_{n+1}} \pair {co} {\cdot} \pair \# d$,
	where $\forall i,j \colon i \neq j \implies d_i~\neq~d_j$.
	The word $w$ belongs in $\langtransact$ iff $\exists i \in \{1, \ldots ,n+1\} \colon d = d_i$.
	
	Because $\aut_H$ is deterministic (therefore there is only one possible configuration
	at any given point in the input word),
	$\aut_H$ must store every data value of $a$ in $w$ until a $\#$ appears.
	If the data value of the $k$-th $a$ in $w$ were not stored in
	some register of $\aut_H$ (or the register was emptied before $\#$ was reached),
	there is no way for $\aut_H$ to distinguish
	between $w$ where $d = d_k$, which belongs in $\langtransact$,
	and $w$ where $d$ is a value not equal to any other
	data value in $w$, which does not belong in $\langtransact$.

	Using the pigeonhole principle we can then deduce that at
	least two data values of $a$~in $w$ must be stored in the same register.
	Let $l,m \in \mathbb{N}$, such that $l<m$ and $d_l$ and $d_m$ are the first
	two data values of $a$ stored
	in one register. We then look at the word $w'$ that
	is the same as $w$, except the pair $\pair {co} \cdot$ following
	the $m$-th $a$ has been replaced with the pair $\pair {rb} \cdot$.
	As $\aut_H$ is deterministic and words $w$ and $w'$ are the same until
	after $\pair a {d_m}$ appears, we know that when reading $w'$, $d_l$ and $d_m$
	will be stored in the same register, but $d_l$ was committed, whereas $d_m$ was rolled back.
	This means that $\aut_H$ loses the distinction between $d_l$ and $d_m$ and would either
	accept $w'$ where $d = d_m$ or reject $w'$ where $d = d_l$. This is a~contradiction
	with the assumption that $\aut_H$ accepts $\langtransact$ and we can conclude that
	no DHRA can accept $\langtransact$.
\end{proof}

Thus, we have shown that the deterministic variant of RsA$^{rm}$ is strictly more
expressive than the deterministic variant of HRA.
\begin{proposition}
	DHRA $\subsetneq$ DRsA$^{rm}$
\end{proposition}
\begin{proof}
	Follows from \cref{corollary:dhra_subseteq_drsa} and
	\cref{lemma:drsa_notsubseteq_dhra}.
\end{proof}

%% file: proofs-properties.tex
\lemNraToRsa*

\begin{proof}
We transform every $\nra$ transition $\trans q a {\guardeq, \guardneq, \update} s$ into
the $\rsa$ transition $\trans q a {\guardin, \guardnotin, \update'} s$ such that
$\guardin = \guardeq$, $\guardnotin = \guardneq$, and for every register~$r_i$
and $\update(r_i) = x$,
$$\update'(r_i) = \begin{cases}
  \{x\} & \text{for } x \in \regs \cup \set{\inp} \\
  \emptyset & \text{for } x = \bot .
\end{cases}$$
Intuitively, every simple register of $\nra$ will be represented by a~set
register of $\rsa$ that will always hold the value of either an empty or
a~singleton set.
\end{proof}

\vspace{-0.0mm}
\subsection{Proof of \cref{thm:rsa-emptiness}}\label{sec:proof-rsa-emptiness}
\vspace{-0.0mm}

\thmRsaEmptiness*

The proof is done by showing interreducibility of $\rsa$ emptiness with
coverability in \emph{transfer Petri nets} (often used for modelling the
so-called \emph{broadcast protocols}), which is a~known $\fomega$-complete
problem~\cite{SchmitzS13,Schmitz17,SchmitzS12}.
We start with defining transfer Petri nets and then give the reductions.

\vspace{-0.0mm}
\subsubsection{Transfer Petri Nets}
\vspace{-0.0mm}

Intuitively, transfer Petri net is an extension of Petri nets where transitions can
\emph{transfer} all tokens from one place to another place at once.
They are closely related to \emph{broadcast protocols}~\cite{EsparzaFM99}.

Formally,
a~\emph{transfer Petri net} (TPN) is a~triple $\net = (P, T, M_0),$ s.t.\
$P$~is a~finite set of \emph{places},
$T$ is a~finite set of \emph{transitions}, and
$M_0\colon P \to \nat$ is an \emph{initial marking}.
The set of transitions~$T$ is such that $P \cap T = \emptyset$ and every transition $t \in T$ is of the form
$t = \tuple{\intrans, \outrans, \transfer}$ where $\intrans,\outrans\colon P \to
\nat$ define $t$'s \emph{input} and \emph{output places} respectively and
$\transfer\colon P \to P$ is a~(total) \emph{transfer function}.

A~\emph{marking} of~$\net$ is a~function $M\colon P \to \nat$ assigning
a~particular number of tokens to each state.
Given a~pair of markings~$M$ and~$M'$, we use $M \leq M'$ to denote that
for all $p \in P$ it holds that $M(p) \leq M'(p)$.
Moreover, we use $M + M'$ and $M - M'$ (for $M' \leq M$) to denote the
pointwise addition and subtraction of markings and we use $\unitof p$ to denote
the marking such that $\unitof p (p') = 1$ if $p = p'$ and $\unitof p (p') = 0$
otherwise, and $\vec 0$ to denote the marking $\set{p \mapsto 0 \mid p \in
P}$.
The \emph{identity} function (over an arbitrary set that is clear from the
context) is denoted as~$\identity$.
Given a~marking~$M$, a~transition~$t = \tuple{\intrans, \outrans, \transfer}$ is
\emph{enabled} if $\intrans \leq M$,
i.e., there is a~sufficient number of tokens in each of its input places.
We use $M \fire t M'$ to denote that
\begin{enumerate}
  \item  $t$ is enabled in~$M$ and
  \item  $M'$ is the marking such that for every $p\in P$ we have
        \begin{equation}
          M'(p) = \outrans(p) + \sum \set{M_{\mathit{aux}}(p') \mid \transfer(p')
          = p} \text{ where } M_{\mathit{aux}} = M - \intrans.
        \end{equation}
\end{enumerate}
That is, the successor marking~$M'$ is obtained by
\begin{inparaenum}[(i)]
  \item  removing $\intrans$ tokens from inputs of~$t$,
  \item  transferring tokens according to~$\transfer$, and
  \item  adding $\outrans$ tokens to $t$'s outputs.
\end{inparaenum}

We say that a~marking~$M$ is \emph{reachable} if there is a~(possibly empty)
sequence $t_1, t_2, \ldots, t_n$ of transitions such that it holds that
$M_0 \fire{t_1} M_1 \fire{t_2} \ldots \fire{t_n} M$, where~$M_0$ is the initial
marking.
%
A~marking~$M$ is \emph{coverable}, if there exists a reachable marking $M'$,
such that $M \leq M'$.
The \emph{Coverability} problem for TPNs asks, given a~TPN~$\net$ and
a~marking~$M$, whether~$M$ is coverable in~$\net$.

\begin{proposition}[\cite{SchmitzS13}]\label{prop:tpn-coverability-fomega-complete}
The \emph{Coverability} problem for TPNs is $\fomega$-complete.
\end{proposition}

\noindent
Let us now prove the two directions of the proof of \cref{thm:rsa-emptiness}.

\begin{lemma}\label{lem:rsa-emptiness-in-fomega}
The emptiness problem for $\rsa$ is in $\fomega$.
\end{lemma}

\begin{proof}
The proof of the lemma is based on reducing the $\rsa$ emptiness problem to
coverability in TPNs, which is $\fomega$-complete
(\cref{prop:tpn-coverability-fomega-complete}).
Intuitively, the reduction consists of creating a TPN with
places representing both individual states of $\rsa$ and
individual \emph{regions} of the Venn diagram of $\regs$.
Transitions of $\rsa$ are represented by one or more transitions
of TPN, distinguishing every possible option of $\inp$ being
in some region $\region \in 2^\regs$.
The set of arcs leading to and from each transition is calculated in a way which preserves the semantics and position of values defined by the \emph{guard} and \emph{update} formulae.
Finally, the marking to be covered requires one token
to be present in the places representing the final states of the $\rsa$.

Formally, let $\aut = (Q, \regs, \Delta, I, F)$ be an $\rsa$.
In the following, we will construct a~TPN $\netof \aut = (P, T, M_0)$ and
a~marking~$M_F$ such that $\langof \aut \neq \emptyset$ iff~$M_F$ is coverable
in~$\netof \aut$.
We set the components of~$\netof \aut$ as follows:
\begin{itemize}
  \item  $P = Q \uplus \set{\init, \fin} \uplus 2^\regs$ where
    $\init$ and $\fin$ are two new places,
  \item  $T = \set{\tuple{\unitof{\init}, \unitof{q_i}, \identity} \mid
    q_i \in I} \cup \set{\tuple{\unitof {q_f}, \unitof{\fin}, \identity}
    \mid q_f \in F} \cup T'$ with $T'$ defined below, and
  \item  $M_0 = \unitof{\init}$.
\end{itemize}
Intuitively, the set of places contains the states of~$\aut$ (there will always
be at most one token in those places), two new places $\init$ and
$\fin$ that are used for the initial nondeterministic choice of some
initial state of~$\aut$ and for a~unique \emph{final place} (whose coverability
will be checked) respectively, and, finally, a~new place for every
\emph{region} of the Venn diagram of~$\regs$, which will track the number of
data values that two or more registers share (e.g., for $\regs = \set{r_1,
r_2, r_3}$, the subset $\set{r_1, r_3}$ denotes the region $r_1 \cap
\cmplof{r_2} \cap r_3$, i.e., the data values that are stored in~$r_1$
and~$r_3$ but are not stored in~$r_2$. The region $\cmplof{r_1} \cap \cmplof{r_2}
\cap \cmplof{r_3}$ is denoted as $\{\emptyset\}$.
Regions $\regionof 1, \regionof 2$ are \emph{distinct},
if $\regionof 1 \triangle \regionof 2 \neq \emptyset$, i.e.
$\exists r \colon (r \in \regionof 1 \lor r \in \regionof 2)
\land r \notin \regionof 1 \cap \regionof 2$.
The set of TPN transitions~$T'$ is defined below.

We now proceed to the definition of $T'$.
Let $t = \trans q a {\guardin, \guardnotin, \update} s \in \Delta$ be
a~transition in~$\aut$.
Then, we create a~TPN transition for every possible option of~$\inp$
being in some region $\regionof g \in 2^\regs$ (e.g., for $\inp \in r_1 \cap
\cmplof{r_2} \cap r_3$ or $\inp \in \cmplof{r_1} \cap \cmplof{r_2} \cap r_3$).
For~$t$ and~$\regionof g$, we define
\begin{equation}
\gamma(t, \regionof g) =
\begin{cases}
  \set{\tuple{\intrans, \outrans, \transfer}} & \text{if } (\guardin \subseteq
    \regionof g) \land (\guardnotin \cap \regionof g = \emptyset)  \text{ and} \\
  \emptyset & \text{otherwise}.
\end{cases}
\end{equation}
Then $T' = \bigcup\set{\gamma(t, \regionof g) \mid t \in \Delta, \regionof g \in
2^\regs} \cup \set{\tuple{\vec 0, \unitof{\emptyset}, \identity}}$.
Note that the task of the TPN transition $\tuple{\vec 0, \unitof{\emptyset},
\identity}$ is to provide an arbitrary number of tokens in the place for the
region
$\set \emptyset = \cmplof{r_1} \cap \ldots \cap \cmplof{r_n}$, which corresponds to an unlimited
number of input data values.

In the definition of $\gamma(t, \regionof g)$ above, $\intrans$, $\outrans$, and
$\transfer$ are defined as follows:
\begin{itemize}
  \item  $\intrans = \unitof{\regionof g} + \unitof q$ and
  \item  $\outrans = \unitof{\dst} + \unitof s$ where
    $\dst = \set{r_i \in \regs \mid (\set{\inp} \cup \regionof g) \cap \update(r_i) \neq \emptyset}$.
  \item  Before we give a formal definition of $\transfer$, let us start with an
    intuition given in the following example.

    \begin{example}\label{ex:rsa_emptiness_example_transfer}
      Let us consider the $\rsa$ in \cref{fig:rsa_emptiness_example_rsa} and its
      transition $\trans q a {\emptyset, \emptyset, \update} q$ with
      $\update(r_1) = \set{r_1, \inp}$ and $\update(r_2) = \set{r_1, r_2}$.
      We need to update the following four regions of the Venn diagram of~$r_1$
      and~$r_2$: $r_1 \cap r_2$, $r_1 \cap \cmplof{r_2}$, $\cmplof{r_1} \cap
      r_2$, and $\cmplof{r_1} \cap \cmplof{r_2}$.
      From the update function~$\update$, we see that the new values stored
      in~$r_1$ and $r_2$ will be (we used primed versions of register names to
      denote their value after update) $r_1' = r_1$ (we do not consider
      $\set{\inp}$ here because it has been discharged within $\outrans$ in the
      previous step) and $r_2' = r_1 \cup r_2$.
      The values of the regions will therefore be updated as follows:

      \noindent
      \hspace*{-12mm}
      \scalebox{0.9}{
      \begin{minipage}{1.2\textwidth}
      \begin{equation}\label{eq:rsa_emptiness_example_transfer}
      \begin{aligned}
        r_1' \cap r_2'                   &{}= r_1 \cap (r_1 \cup r_2) &
        \hspace*{-5mm}r_1' \cap \cmplof{r_2'}          &{}= r_1 \cap \cmplof{(r_1 \cup r_2)} &
        \cmplof{r_1'} \cap r_2'          &{}= \cmplof{r_1} \cap (r_1 \cup r_2) &
        \cmplof{r_1'} \cap \cmplof{r_2'} &{}= \cmplof{r_1} \cap \cmplof{(r_1 \cup r_2)}\hspace*{-5mm} \\
        &{}= (r_1 \cap r_1) \cup (r_1 \cap r_2) &
        &{}= r_1 \cap \cmplof{r_1} \cap \cmplof{r_2} &
        &{}= (\cmplof{r_1} \cap r_1) \cup (\cmplof{r_1} \cap r_2)\hspace*{-5mm} &
        &{}= \cmplof{r_1} \cap \cmplof{r_1} \cap \cmplof{r_2} \\
        &{}= r_1 \cup (r_1 \cap r_2) &
        &{}= \emptyset &
        &{}= \cmplof{r_1} \cap r_2 &
        &{}= \cmplof{r_1} \cap \cmplof{r_2} \\
        &{}= (r_1 \cap \cmplof{r_2}) \cup (r_1 \cap r_2) &
      \end{aligned}
      \end{equation}
      \end{minipage}
      }

      \medskip
      \noindent
      Note that in the last step of the calculation of $r_1' \cap r_2'$, we used
      the fact that $r_1 = (r_1 \cap r_2) \cup (r_1 \cap \cmplof{r_2})$ in order
      to obtain a union of regions.
      From the previous calculation, we see that $\transfer$ should be set as
      follows:
      $\transfer(\set{r_1, r_2}) = \set{r_1, r_2}$,
      $\transfer(\set{r_1}) = \set{r_1, r_2}$,
      $\transfer(\set{r_2}) = \set{r_2}$, and
      $\transfer(\{\emptyset\}) = \{\emptyset\}$.
      \claimqed
    \end{example}

\begin{figure}[t]
\begin{center}
\input{figs/rsa_emptiness_example_rsa.tikz}
\end{center}
\caption{An example $\rsa$}
\label{fig:rsa_emptiness_example_rsa}
\end{figure}

\newcommand{\figTpnExample}[0]{
\begin{figure}[t]
\begin{center}
  \input{figs/tpn_example_other.tikz}
\end{center}
\caption{TPN for the transition
  $\trans q a {\emptyset, \emptyset,\set{r_1 \mapsto \set{r_1,\inp}, r_2
  \mapsto \set{r_1, r_2}}} q$
  of the $\rsa$ from
  \cref{fig:rsa_emptiness_example_rsa} (corresponding colours represent the
  source position of $\inp$ for the given transition).}
\label{fig:tpn_example}
\end{figure}
}


    Formally, $\transfer$ is computed as follows.
    For every $\regionof o \in 2^\regs$, let us compute the sets of sets of registers
    \begin{equation}
    \positpos(\regionof o) = \set{\update(r_i) \cap \regs \mid r_i \in \regionof o}\quad\text{and}\quad
    \negat(\regionof o) = \bigcup\set{\update(r_i) \cap \regs \mid r_i \notin \regionof o}
    \end{equation}
    ($\posit$ is for ``\emph{positive}'' and $\negat$ is for
    ``\emph{negative}'', which represent registers that occur positively and
    negatively, respectively, in the specification of the region of the Venn
    diagram $\regionof o$).
    The intuition is that $\positpos(\region_o)$ represents the update
    of~$\regionof o$ as the \emph{product of sums} (intersection of unions),
    cf.\ the first line in \cref{eq:rsa_emptiness_example_transfer} in
    \cref{ex:rsa_emptiness_example_transfer}.
    Next, we convert the product of sums $\positpos(\regionof o)$ into
    a~sum of products
    \begin{equation}
    \positsop(\regionof o) = \krugel \positpos(\regionof o)
    \end{equation}
    where $\krugel \set{D_1, \dots, D_n}$ is the \emph{unordered Cartesian
    product} of sets $D_1, \ldots, D_n$, i.e.,
    \begin{equation}
    \krugel \set{D_1, \ldots, D_n} = \bigl\{\{d_1, \dots, d_n\} \bigm| (d_1, \dots,
      d_n) \in D_1 \times \cdots \times D_n\bigr\}.
    \end{equation}

    \begin{example}
    In the transition considered in \cref{ex:rsa_emptiness_example_transfer}, we
    obtain the following:

    \hspace*{-12mm}
    \scalebox{0.9}{
    \begin{minipage}{1.2\textwidth}
    \begin{align*}
      \positpos(\set{r_1, r_2}) &{} = \set{\set{r_1}, \set{r_1, r_2}} &
      \positpos(\set{r_1}) &{} = \set{\set{r_1}} &
      \positpos(\set{r_2}) &{} = \set{\set{r_1, r_2}} &
      \positpos(\{\emptyset\}) &{} = \set{\{\emptyset\}} \\
      \positsop(\set{r_1, r_2}) &{} = \set{\set{r_1}, \set{r_1, r_2}} &
      \positsop(\set{r_1}) &{} = \set{\set{r_1}} &
      \positsop(\set{r_2}) &{} = \set{\set{r_1}, \set{r_2}} &
      \positsop(\{\emptyset\}) &{} = \set{\set{\emptyset}} \\
      \negat(\set{r_1, r_2}) &{} = \emptyset &
      \negat(\set{r_1}) &{} = \set{r_1,r_2} &
      \negat(\set{r_2}) &{} = \set{r_1} &
      \negat(\{\emptyset\}) &{} = \set{r_1, r_2}
      ~\claimqed
    \end{align*}
    \end{minipage}
    }
    \end{example}

    Next, we modify $\positsop$ into $\positsop'$ by removing regions that are
    incompatible with $\negat$ to obtain
    \begin{equation}
    \positsop'(\regionof o) = \set{x \in
    \positsop(\regionof o) \mid x \cap \negat(\regionof o) = \emptyset} .
    \end{equation}

    \begin{example}
    In the running example, we would obtain the following values of $\positsop'$:
    \begin{align*}
      \positsop'(\set{r_1, r_2}) &{} = \set{\set{r_1}, \set{r_1, r_2}} &
      \positsop'(\set{r_1}) &{} = \emptyset \\
      \positsop'(\set{r_2}) &{} = \set{\set{r_2}} &
      \positsop'(\{\emptyset\}) &{} = \set{\set{\emptyset}}
    \end{align*}
    Compare the results with the calculation in \cref{eq:rsa_emptiness_example_transfer}.
    \claimqed
    \end{example}

    Lastly, for every~$\regionof i \in 2^\regs$ such that $\positsop'(\regionof
    o) \in \regionof i$, we set $\transfer(\regionof i) = \regionof o$.

\begin{figure}[t]
\begin{center}
  \input{figs/tpn_example_other.tikz}
\end{center}
\caption{TPN for the transition
  $\trans q a {\emptyset, \emptyset,\set{r_1 \mapsto \set{r_1,\inp}, r_2
  \mapsto \set{r_1, r_2}}} q$
  of the $\rsa$ from
  \cref{fig:rsa_emptiness_example_rsa} (corresponding colours represent the
  source position of $\inp$ for the given transition).}
\label{fig:tpn_example}
\end{figure}

    \begin{example}
    Continuing in the running example, we obtain
    \begin{align*}
      \transfer(\set{r_1, r_2}) &{} = \set{r_1, r_2} &
      \transfer(\set{r_1}) &{} = \set{r_1, r_2} \\
      \transfer(\set{r_2}) &{} = \set{r_2} &
      \transfer(\{\emptyset\}) &{} = \{\emptyset\} ,
    \end{align*}
    which is the same result as in \cref{ex:rsa_emptiness_example_transfer}.
    \cref{fig:tpn_example} contains the TPN fragment for all TPN transitions
      constructed from~$\aut$'s transition $\trans q a {\emptyset,
      \emptyset,\set{r_1 \mapsto \set{r_1,\inp}, r_2 \mapsto \set{r_1, r_2}}}
      q$.
    \claimqed
    \end{example}

    The following claim shows that this construction is indeed well defined.

    \begin{claim}
    The function $\transfer$ is well defined.
    \end{claim}
    \begin{claimproof}
        It is necessary to show that no set of values will be
        duplicated and assigned to two distinct regions when $\transfer$ is calculated.
        According to the definition of $\transfer$,
        for all regions $\region' \in 2^\regs$ such that
        $\positsop'(\regionof
        o) \in \region'$, the value of $\transfer(\region')$ is set to be
        $\regionof o$, therefore we need to prove that for each pair of distinct
        regions $\regionof 1$ and $\regionof 2 \in 2^\regs$ it holds that
        $\positsop'(\regionof 1) \cap \positsop'(\regionof 2) = \emptyset$.
        We prove this by contradiction.

        Assume that there are two
        distinct regions $\regionof 1$ and $\regionof 2$ such that
        there exists a~region $\regionof 3 \in \positsop'(\regionof 1)
        \cap \positsop'(\regionof 2)$.
        According to the construction of $\positsop'$, it holds that
    \begin{align*}
    \regionof 3 \in \positsop(\regionof 1)\quad\land\quad&
    \regionof 3 \cap \negat(\regionof 1) = \emptyset \quad \text{and} \\
    \regionof 3 \in \positsop(\regionof 2)\quad\land\quad &
    \regionof 3 \cap \negat(\regionof 2) = \emptyset.
    \end{align*}

    Then, according to the construction of $\positsop$,
    \begin{align*}
    \regionof 3 \in \krugel \positpos(\regionof 1)\quad\land\quad&
    \regionof 3 \cap \negat(\regionof 1) = \emptyset \quad \text{and} \\
    \regionof 3 \in \krugel \positpos(\regionof 2)\quad\land\quad &
    \regionof 3 \cap \negat(\regionof 2) = \emptyset.
    \end{align*}

    Following the construction of $\positpos$:

    \hspace*{-10mm}
    \noindent
    \scalebox{0.95}{
    \begin{minipage}{1.2\textwidth}
    \begin{gather*}
    (\forall r \in \regionof 3 \exists P \in \positpos(\regionof 1)\colon r\in P) \land
    (\forall P' \in \positpos(\regionof 1)\exists r' \in \regionof 3\colon r' \in P') \land
    (\regionof 3 \cap \negat(\regionof 1) = \emptyset) \land {}\\
    (\forall r \in \regionof 3 \exists P \in \positpos(\regionof 2)\colon r\in P) \land
    (\forall P' \in \positpos(\regionof 2)\exists r' \in \regionof 3\colon r' \in P') \land
    (\regionof 3 \cap \negat(\regionof 2) = \emptyset).
    \end{gather*}
    \end{minipage}
    }

    \bigskip

    According to the construction of $\positpos$, it holds that
    if $P \in \positpos(\region)$ then there exists a~register $r \in \region$
      such that $P = \update(r_i)$.
    Therefore, for each $P \in \positpos(\region)$ there exists a~register $r'
      \in \region'$ such that $r \in P$.
    Then, continuing in the proof, we obtain

    \hspace*{-13mm}
    \noindent
    \scalebox{0.95}{
      \begin{minipage}{1.2\textwidth}
        \begin{gather*}
        (\forall r \in \regionof 3 \exists \update(r_i)\colon r_i \in \regionof 1 \land r\in \update(r_i)) \land
        (\forall r^\bullet \in \regionof 1 \exists r' \in \regionof 3 \colon r' \in \update(r^\bullet)) \land
        (\regionof 3 \cap \negat(\regionof 1) = \emptyset) \land {} \\
        (\forall r \in \regionof 3 \exists \update(r_i)\colon r_i \in \regionof 2 \land r\in \update(r_i)) \land
        (\forall r^\bullet \in \regionof 2 \exists r' \in \regionof 3 \colon r' \in \update(r^\bullet)) \land
        (\regionof 3 \cap \negat(\regionof 2) = \emptyset)
        \end{gather*}
      \end{minipage}
    }

    \bigskip

    From the construction of $\negat(\region)$, the formula $\regionof i \cap \negat(\regionof j) = \emptyset$
    is equivalent to the formula $\forall r \in \regionof i \neg \exists r^* \notin \regionof j\colon r \in \update(r^*)$.
    Therefore:
    \begin{gather*}
    (\forall r \in \regionof 3 \exists \update(r_i)\colon r_i \in \regionof 1 \land r\in \update(r_i)) \land
    (\forall r^\bullet \in \regionof 1 \exists r' \in \regionof 3 \colon r' \in
      \update(r^\bullet)) \land {}\\
    (\forall r \in \regionof 3 \neg \exists r^* \notin \regionof 1\colon r \in
      \update(r^*)) \land {}\\
    (\forall r \in \regionof 3 \exists \update(r_i)\colon r_i \in \regionof 2 \land r\in \update(r_i)) \land
    (\forall r^\bullet \in \regionof 2 \exists r' \in \regionof 3 \colon r' \in
      \update(r^\bullet)) \land {}\\
    (\forall r \in \regionof 3 \neg \exists r^* \notin \regionof 2\colon r \in \update(r^*))
    \end{gather*}

    Further, we only make use of
    \begin{align*}
    (\forall r \in \regionof 3 \neg \exists r^* \notin \regionof 1\colon r \in \update(r^*))\quad\land\quad&
    (\forall r^\bullet \in \regionof 2 \exists r' \in \regionof 3 \colon r' \in \update(r^\bullet)),
    \end{align*}
    and the fact that $\regionof 1$ and $\regionof 2$ are distinct. Therefore, $\exists r^\mathit{dist} \colon
    r^\mathit{dist} \in \regionof 2 \land r^\mathit{dist} \notin \regionof 1$ (or vice versa).

    By simplifying
    \begin{gather*}
    (\exists r^\mathit{dist} \colon r^\mathit{dist} \in \regionof 2 \land
      r^\mathit{dist} \notin \regionof 1) \land {} \\
    (\forall r \in \regionof 3 \neg \exists r^* \notin \regionof 1\colon r \in
      \update(r^*)) \land {} \\
    (\forall r^\bullet \in \regionof 2 \exists r' \in \regionof 3 \colon r' \in \update(r^\bullet)),
    \end{gather*}
    we obtain
    \begin{gather*}
    (\exists r^\mathit{dist} \colon r^\mathit{dist} \in \regionof 2 \land
      r^\mathit{dist} \notin \regionof 1) \land {} \\
    (\forall r \in \regionof 3 \forall r^* \colon r^* \notin \regionof 1
      \rightarrow r \notin \update(r^*)) \land {} \\
    (\exists r' \in \regionof 3 \colon r' \in \update(r^\mathit{dist})),
    \end{gather*}
    which is clearly a contradiction, since $r^\mathit{dist} \notin \regionof 1
    \land \exists r' \in \regionof 3 \colon r' \in \update(r^\mathit{dist})$.
    \end{claimproof}
\end{itemize}

Finally, the marking~$M_F$ to be covered is constructed as $M_F = \unitof{\fin}$.
We have finished the construction of~$\netof \aut$, now we need to show that it
preserves the answer.

\begin{claim}\label{lem:rsa_emptiness_preserves_answer}
$\langof \aut \neq \emptyset$ iff the marking $M_F$ is coverable in $\netof
\aut$.
\end{claim}

\begin{claimproof}

\begin{enumerate}
    \item[$(\Rightarrow)$]
      Let $w \in (\Sigma \times \datadom)^*$ such that $w \in \langof\aut$.
      Moreover, assume that
      \begin{equation*}
        \rho\colon c_0 \stepofusing{\wordof 1} {t_1} c_1 \stepofusing{\wordof 2} {t_2}
    \ldots \stepofusing{\wordof n} {t_n} c_n
      \end{equation*}
      is an accepting run of~$\aut$ on~$w$.
      We will show that there exists a~sequence of firings
      \begin{equation*}
        \rho'\colon
        M_{\init} \fire{t'_{\init}}~~
        M_0 \fire{t'_1}
        M_1 \fire{t'_2} \ldots \fire{t'_n}
        M_n ~~\fire{t'_{\fin}} M_{\fin}
      \end{equation*}
      in $\netof \aut$ such that~$M_{\fin}$ covers~$M_F$.
      In particular, we construct the markings and transitions as follows:
      \begin{itemize}
        \item  $M_{\init} = \unitof{\init}$ and $t'_{\init} =
          \tuple{\unitof{\init}, \unitof{q_0}, \identity}$ for $c_0 =
          (q_0, f_0)$.
        \item  For all $0 \leq i \leq n$ with $c_i = (q_i, f_i)$, we set
          $M_i$ as follows:
          \begin{align*}
            M_i = {}&\set{\init \mapsto 0, \fin \mapsto 0, q_i \mapsto 1} \cup
            \set{q \mapsto 0 \mid q \in Q \setminus \set{q_i}} \cup {} \\
            &\Big\{\region \mapsto x \mid \region \subseteq \regs, x =
            \big|\cap_{r \in \region} f_i(r)\big|\Big\}
          \end{align*}
          Furthermore, $t'_i = \gamma(t_i, \regionof \guard)$ where $\regionof
          \guard = \set{r \mid d_i \in f_{i-1}(r)}$.

        \item  $M_{\fin}$ is as follows:
          \begin{align*}
            M_{\fin} = {}&\set{\init \mapsto 0, \fin \mapsto 1} \cup
            \set{q \mapsto 0 \mid q \in Q} \cup
            \set{\region \mapsto M_n(\region) \mid \region \subseteq \regs}
          \end{align*}
          and $t'_{\fin} =
          \tuple{\unitof{q_n}, \unitof{\fin}, \identity}$ for $c_n =
          (q_n, f_n)$.
      \end{itemize}
      Note that $M_F$ is covered by~$M_{\fin}$.

      We can now show by induction that $\rho'$ is valid, i.e., all firings are
      enabled and respect the transition relation.

    \item[$(\Leftarrow)$]
    Let $\rho \colon M_\init \fire{t_0}~~M_0 \fire{t_1} M_1 \fire{t_2} \ldots
    \fire{t_n} M_n~~\fire{t_\fin} M_\fin$ be a~run of~$\netof \aut$
    such that $M_\fin \geq M_F$, where $M_F$ is the final marking.
    We show that there exists a~sequence of transitions
    \begin{equation*}
    \rho'\colon
    c_0 \stepofusing{\wordof 1} {t'_1} c_1 \stepofusing{\wordof 2} {t'_2}
    \ldots \stepofusing{\wordof n} {t'_n} c_n
    \end{equation*}
    in $\aut$ on $w = \wordof{1}\wordof{2}\ldots\wordof{n}$, such that $c_n$ is a final
    configuration, and, therefore, $w \in \langof{\aut}$.
    First, we notice the following easy to see invariant of~$\netof \aut$, which
    holds for every $M_i$:
    \begin{equation}\label{eq:net_sim_invariant}
      \sum_{p \in Q \cup \set{\init,\fin}}\hspace{-5mm} M_i(p) = 1
    \end{equation}
    i.e., there is always exactly one token in any of the places in $Q \cup
    \set{\init, \fin}$.

    Let us now construct~$\rho'$ as follows:
    \begin{itemize}
      \item  $c_0 = (q_0, f_0)$ is constructed such that $q_0$ is picked to be
        the state $q_0 \in Q$ with $M_0(q_0) = 1$ (this is well defined due
        to \cref{eq:net_sim_invariant}).
      \item  For all $1 \leq i \leq n$, the transition~$t'_i$ is picked to be
        the transition such that $t_i \in \gamma(t'_i, \regionof\guard)$ for
        some region~$\regionof\guard$.
        The data value~$d_i$ of~$\wordof i$ is then chosen to be compatible with the
        guard of~$\regionof \guard$, i.e., $d_i \in \cap_{r \in \regionof
        \guard} f_{i-1}(r)$ and $d_i \notin \cup_{r \in \regs \setminus
        \regionof \guard}f_{i-1}(r)$.
    \end{itemize}

    It can then be shown by induction that, for all $0 \leq i < n$, the
    following holds:
    \begin{enumerate}
      \item  $c_i = (q_i, f_i)$ where~$q_i$ is the (exactly one) state such that
        $M_i(q_i) = 1$ and
      \item  the transition $t_{i+1}$ is enabled.
    \end{enumerate}
    We can then conclude that, since the last firing in~$\rho$ was $M_n
    \fire{t_\fin} M_\fin$, then, from the construction of~$\netof \aut$, it
    holds that $M_n(q_f) = 1$ for $q_f \in F$ and so~$\rho'$ is accepting.
    \claimqed
\end{enumerate}
\end{claimproof}

\cref{lem:rsa_emptiness_preserves_answer} and the observation that $\netof \aut$
is single-exponentially larger than~$\aut$ conclude the proof ($\fomega$ is
closed under primitive-recursive reductions).
\end{proof}

\begin{lemma}\label{lem:rsa-emptiness-fomega-hard}
The emptiness problem for $\rsa$ is $\fomega$-hard.
\end{lemma}

\begin{proof}
The proof is based on a~reduction of coverability in TPNs (which is
$\fomega$-complete) to non-emptiness of $\rsa$s.
Intuitively, given a~TPN~$\net$, we will construct the $\rsa$ $\autof \net$
simulating~$\net$, which will have the following structure:
\begin{itemize}
  \item  There will be the state~$\qmain$, which will be active before and after
    simulating the firing of TPN transitions.
  \item  Each place of~$\net$ will be simulated by a~register of~$\autof \net$;
    every token of~$\net$ will be simulated by a~unique data value.
  \item  For every TPN transition, $\autof \net$ will contain a~\emph{gadget}
    that transfers data values between the registers representing the places
    active in the TPN transition.
    The gadget will start in~$\qmain$ and end also in~$\qmain$.
  \item  Coverability of a~marking will be simulated by another gadget connected
    to~$\qmain$ that will try to remove the number of tokens given in the
    marking from the respective places and arrive at the single final
    state~$\qfin$.
\end{itemize}

Formally, let $\net = (P, T, M_0)$ with $P = \{p_1, \ldots, p_n\}$ be a~TPN.
W.l.o.g.\ we can assume that $M_0$ contains a~single token in the place~$p_1$,
i.e., $M_0 = \{p_1 \mapsto 1, p_2 \mapsto 0, \ldots, p_n \mapsto 0\}$.
We will show how to construct the $\rsa$ $\autof \net = (Q, \regs, \Delta,
\{\qinit\}, \{\qfin\})$ over the unary alphabet $\Sigma = \{a\}$ such that
a~marking~$M_f$ is coverable in~$\net$ iff the language of~$\autof \net$ is
non-empty.
The set of registers of $\autof \net$ will be the set $\regs = \{\regof
\inp, \regof \tmp\} \cup \{\regof p, \regof{p'} \mid p \in P\}$.

\begin{figure}[t]
\centering
\begin{subfigure}[b]{0.43\linewidth}
\begin{center}
\scalebox{0.95}{
  \input{figs/protogadget_lossyrem.tikz}
}
\end{center}
\caption{The $\lossyremgadgetof{p}$ protogadget.}
\label{fig:protogadget_lossyrem}
\end{subfigure}\\[4mm]
\hfill
\begin{subfigure}[b]{0.37\linewidth}
\begin{center}
\scalebox{0.95}{
  \input{figs/protogadget_move.tikz}
}
\end{center}
\caption{The $\movegadgetof{p_1, p_2}$ protogadget.}
\label{fig:protogadget_move}
\end{subfigure}
\hfill
\begin{subfigure}[b]{0.37\linewidth}
\begin{center}
\scalebox{0.95}{
  \input{figs/protogadget_newtoken.tikz}
}
\end{center}
\caption{The $\newtokengadgetof{p}$ protogadget.}
\label{fig:protogadget_newtoken}
\end{subfigure}
\caption{Protogadgets used in the construction of $\rsaof \net$.}
\end{figure}

Let us now define the following \emph{protogadgets}, which we will later use
for creating a~\emph{gadget} for each TPN transition and a~gadget for doing the
coverability test.
We define the following protogadgets:
\begin{enumerate}
  \item  The \emph{Lossy Removal} protogadget, which simulates a~(lossy)
    removal of one token from a~place~$p$ is the $\rsa$ defined as
    $\lossyremgadgetof{p} = (\{q_1, q_2, q_3\}, \regs, \Delta', \{q_1\}, \{q_3\})$
    where~$\Delta'$ contains the following three transitions (cf.\
    \cref{fig:protogadget_lossyrem}):
    \begin{equation}
      \Delta' = \left\{
        \begin{array}{c}
        \trans{q_1}{a}{\{\regof p\},\emptyset,\{\regof \inp \mapsto \{\inp\},
          \regof \tmp \mapsto \emptyset\}}{q_2},\\
        \trans{q_2}{a}{\{\regof p\}, \{\regof \inp\}, \{\regof \tmp \mapsto
          \{\regof \tmp, \inp\}\}}{q_2},\\
        \trans{q_2}{a}{\emptyset, \emptyset, \{\regof p \mapsto \{\regof \tmp\}\}}{q_3}
        \end{array}
      \right\}
    \end{equation}

    Intuitively, the protogadget stores the data value
    to be removed from $p$ in a special register $r_\inp$.
    Next, it simulates the calculation of the
    difference of $r_p$ and $r_\inp$.
    This is done by accumulating the values
    which are present in $r_p$ and are not present in
    $r_\inp$ into $r_\tmp$.

    Since some values may get ``lost'', and disappear because of not being added to
    the accumulated difference, this protogadget is
    considered \emph{lossy}.

  \item  The \emph{Move} protogadget, which simulates \emph{moving} all tokens
    from a~place~$p_1$ to a~place~$p_2$, is the following $\rsa$ (also depicted
    in \cref{fig:protogadget_move}):
    \begin{equation*}
      \movegadgetof{p_1, p_2} =
      (\{q_1, q_2\}, \regs, \{\trans{q_1}{a}{\emptyset, \emptyset, \{\regof{p_1}
      \mapsto \emptyset, \regof{p_2} \mapsto \{\regof{p_1},
      \regof{p_2}\}\}}{q_2}\}, \{q_1\}, \{q_2\}).
    \end{equation*}
    Intuitively, the protogadet empties out the register
    which represents the place $p_1$. Its previous
    value is assigned to register representing the place
    $p_2$ in union with its value.

  \item  The \emph{New Token} protogadget, which simulates adding a~token to
    a~place~$p$, is defined as follows (depiction is in
    \cref{fig:protogadget_newtoken}):
    \begin{equation*}
      \newtokengadgetof{p} =
      (\{q_1, q_2\}, \regs, \{\trans{q_1}{a}{\emptyset, \regs, \{\regof p
      \mapsto \{\regof p, \inp\}\}}{q_2}\}, \{q_1\}, \{q_2\}).
    \end{equation*}
    Intuitively, the protogadget adds the unique data
    value from the input tape into the register
    representing the place $p$. The uniqueness
    of the data value is ensured by the $\guardnotin$,
    which requires that the $\inp$ does not belong to
    any of the registers from $\regs$.

\end{enumerate}

For convenience, we will use the following notation.
Let $\autof 1 = (Q_1, \regs, \Delta_1, \{q^I_1\}, \{q^F_1\})$ and $\autof 2 =
(Q_2, \regs, \Delta_2, \{q^I_2\}, \{q^F_2\})$ be a~pair of $\rsa$s with a~single
initial state and a~single final state.
We will use $\autof 1 \concat \autof 2$ to denote the $\rsa$ $(Q_1 \uplus
Q_2, \regs, \Delta_1 \cup \{\transeps{q^F_1}{q^I_2}\} \cup \Delta_2, \{q^I_1\},
\{q^F_2\})$.
Moreover, for $n \in \natzero$, we use $\iterof{\autof 1}{n}$ to denote the
$\rsa$ defined inductively as
\begin{align*}
  \iterof{\autof 1}{0} = {} & (\{q\}, \regs, \emptyset, \{q\}, \{q\}), \\
  \iterof{\autof 1}{i+1} = {} &\iterof{\autof 1}{i} \concat \autof 1.
\end{align*}
Intuitively, $\iterof{\autof 1}{n}$ is a~concatenation of~$n$ copies of~$\autof
1$.

For each TPN transition $t = \tuple{\intrans, \outrans, \transfer}$, we then create
the \emph{gadget} $\rsa$ $\autof t$ in several steps.
\begin{enumerate}
  \item  First, we transform~$\intrans$ into the $\rsa$ $\autof \intrans =
    \autof{\intrans(p_1)} \concat \ldots \concat \autof{\intrans(p_n)}$ where
    every~$\autof{\intrans(p_i)}$ is defined as $\autof{\intrans(p_i)} =
    \iterof{\lossyremgadgetof{p_i}}{\intrans(p_i)}$, i.e., it is a~concatenation
    of~$\intrans(p_i)$ copies of $\lossyremgadgetof{p_i}$.

  \item  Second, from~$\outrans$ we create the $\rsa$ $\autof \outrans =
    \autof{\outrans(p_1)} \concat \ldots \concat \autof{\outrans(p_n)}$
    with~$\autof{\outrans(p_i)}$ defined as $\autof{\outrans(p_i)}
    = \iterof{\newtokengadgetof{p_i}}{\outrans(p_i)}$, i.e., it is
    a~concatenation of~$\outrans(p_i)$ copies of $\newtokengadgetof{p_i}$.

  \item  Third, from~$\transfer$ we obtain the $\rsa$ $\autof \transfer =
    \autof{\transfer(p_1)} \concat \ldots \concat \autof{\transfer(p_n)} \concat
    \autof{\unprimeof{p_1}} \concat \ldots \concat \autof{\unprimeof{p_n}}$ such
    that $\autof{\transfer(p_i)} = \movegadgetof{\regof{p_i},
    \regof{p_j'}}$ with $p_j = \transfer(p_i)$ and $\autof{\unprimeof{p_i}} =
    \movegadgetof{p'_i, p_i}$.
    Intuitively, $\autof \transfer$ first moves the contents of all registers
    according to $\transfer$ to primed instances of the target registers (in
    order to avoid mix-up) and then unprimes the register names.

  \item  Finally, we combine the $\rsa$s created above into the single gadget
    $\autof t = \autof \intrans \concat \autof \transfer \concat \autof
    \outrans$.

\end{enumerate}

The initial marking will be encoded by a~gadget that puts one new data value in
the register representing the place~$p_1$.
For this, we construct the $\rsa$ $\autof{M_0} = \newtokengadgetof{p_1}$ and
rename its initial state to~$\qinit$.

The last ingredient we need is to create a~gadget that will encode the
marking~$M_f$, whose coverability we are checking.
For this, we construct the gadget $\autof{M_f} = \autof{M_f(p_1)} \concat \ldots
\concat \autof{M_f(p_n)}$ where every~$\autof{M_f(p_i)}$ is defined as
$\autof{M_f(p_i)} = \iterof{\lossyremgadgetof{p_i}}{M_f(p_i)}$, i.e., it is
a~concatenation of~$M_f(p_i)$ copies of $\lossyremgadgetof{p_i}$.
We rename the final state of~$\autof{M_f}$ to~$\qfin$.
W.l.o.g.\ we assume that the set of states of all constructed gadgets are
pairwise disjoint.

We can now finalize the construction:
$\autof \net$ is obtained as the union of the following $\rsa$s:
$\autof{M_0} = (Q_{M_0},
\regs, \Delta_{M_0}, \{\qinit\}, \{q^F_{M_0}\})$, $\autof{M_f} = (Q_{M_f},
\regs, \Delta_{M_f}, \{q^I_{M_f}\}, \{\qfin\})$,
and~$\autof t = (Q_t, \regs, \Delta_t, \{q^I_t\}, \{q^F_t\})$ for every $t \in
T$, connected to the state~$\qmain$, i.e., $\autof \net = (Q, \regs, \Delta,
\{\qinit\}, \{\qfin\})$ where
\begin{itemize}
  \item  $Q =\{\qmain\} \cup Q_{M_0} \cup Q_{M_f} \cup \bigcup_{t \in T} Q_t$ and
  \item  $\Delta = \Delta_{M_0} \cup \Delta_{M_f} \cup
     \{ \transeps{q^F_{M_0}}{\qmain}, \transeps{\qmain}{q^I_{M_f}}\}
    \cup {}$\\
    \hspace*{8mm}$\bigcup_{t \in T} (\Delta_t \cup \{\transeps \qmain {q^I_t},
    \transeps{q^F_t} \qmain\})$.
\end{itemize}

\begin{claim}\label{lem:tpn_coverabilityu_preserves_answer}
The marking~$M_F$ is coverable in $\net$ iff $\langof{\autof \net} \neq
\emptyset$.
\end{claim}

\begin{claimproof}
\begin{enumerate}
    \item[$(\Rightarrow)$]
      Let there be the following run of~$\net$:
      \begin{equation*}
        \rho\colon
        M_0 \fire{t_1}
        M_1 \fire{t_2} \ldots \fire{t_n}
        M_n
      \end{equation*}
      such that~$M_n$ covers~$M_F$.
      We will show that there exists a~word $w \in (\Sigma \times \datadom)^*$
      and a~run
      \begin{align*}
        \rho'\colon&
        c_{(\init,0)} \stepofusing{\wordof 1} {t'_{(\init,1)}}
        c_{(\init,1)} \stepofusing{\wordof 2} {t'_{(\init,2)}}
        \ldots
        c_{(\init, k_\init)} \stepofusing{\wordof{i_0}} {t'_{(0, 0)}}
        \hspace*{25mm}\expl{\textrm{initialization}}
        \\
        &c_{(0,0)}
        \stepofusing{\wordof {i_0 + 1}} {t'_{(0,1)}}
        c_{(0,1)}
        \ldots
        c_{(0,k_1)}
        \stepofusing{\wordof{i_1}} {t'_{(1, 0)}}
        c_{(1,0)}
        \stepofusing{\wordof {i_1 + 1}} {t'_{(1,1)}}
        c_{(1,1)}
        \ldots
        c_{(1,k_1)}
        \ldots
        c_{(n-1,k_{n-1})}
        \stepofusing{\wordof {i_n}} {t'_{(n,0)}} \\
        &c_{(n, 0)}
        \stepofusing{\wordof {i_n + 1}} {t'_{(\fin,1)}}
        c_{(\fin, 1)}
        \ldots
        \stepofusing{\wordof {i_{\fin}}} {t'_{(\fin,k_\fin)}}
        c_{(\fin, k_\fin)}
        \hspace*{41mm}\expl{\textrm{finalization}}
      \end{align*}
      of~$\autof \net$ on~$\word$ such that $q \in F$ for $c_{(\fin, k_\fin)} =
      (q, \cdot)$.
      The run~$\rho'$ will be constructed to preserve the following invariant
      for each $0 \leq i \leq n$:
      \begin{equation}\label{eq:net_invariant}
        c_{(i, 0)} = (\qmain, f_i) \quad \text{such that} \quad \forall p
          \in P\colon|f_i(r_p)| = M_i(p).
      \end{equation}
      Therefore, configurations with state~$\qmain$ represent the TPN's state
      after (or before) firing a~transition.
      Firing a~transition~$t$ is simulated by going to the gadget for~$t$
      in~$\autof \net$ and picking input data values such that the run returns
      to~$\qmain$ in as many steps as possible (this is to take the run through
      the $\lossyremgadget$ protogadgets that preserves the precise value of
      the marking).
      By induction on $0 \leq i \leq n$, we can show that the
      invariant in \cref{eq:net_invariant} is preserved (the base case is proved
      by observing that the \emph{initialization} part is correct).

    \item[$(\Leftarrow)$]
      Let $w \in \langof{\autof \net}$ and
      \begin{align*}
        \rho\colon&
        c_{(\init,0)} \stepofusing{\wordof 1} {t'_{(\init,1)}}
        c_{(\init,1)} \stepofusing{\wordof 2} {t'_{(\init,2)}}
        \ldots
        c_{(\init, k_\init)} \stepofusing{\wordof{i_0}} {t'_{(0, 0)}}
        \hspace*{25mm}\expl{\textrm{initialization}}
        \\
        &c_{(0,0)}
        \stepofusing{\wordof {i_0 + 1}} {t'_{(0,1)}}
        c_{(0,1)}
        \ldots
        c_{(0,k_1)}
        \stepofusing{\wordof{i_1}} {t'_{(1, 0)}}
        c_{(1,0)}
        \stepofusing{\wordof {i_1 + 1}} {t'_{(1,1)}}
        c_{(1,1)}
        \ldots
        c_{(1,k_1)}
        \ldots
        c_{(n-1,k_{n-1})}
        \stepofusing{\wordof {i_n}} {t'_{(n,0)}} \\
        &c_{(n, 0)}
        \stepofusing{\wordof {i_n + 1}} {t'_{(\fin,1)}}
        c_{(\fin, 1)}
        \ldots
        \stepofusing{\wordof {i_{\fin}}} {t'_{(\fin,k_\fin)}}
        c_{(\fin, k_\fin)}
        \hspace*{41mm}\expl{\textrm{finalization}}
      \end{align*}
      be an accepting run of~$\autof \net$ on~$w$ such that for each $0 \leq i
      \leq n$, it holds that $c_{(i,0)} = (\qmain, \cdot)$---this follows from
      the structure of~$\autof \net$.
      We will construct a~run
      \begin{equation*}
        \rho'\colon
        M_0 \fire{t_1}
        M_1 \fire{t_2} \ldots \fire{t_n}
        M_n
      \end{equation*}
      where each~$t_i$ is the TPN's transition corresponding to the gadget
      which the corresponding part of~$\rho$ traversed.
      For all $0 \leq i \leq n$, the following invariant will hold:
      \begin{equation}\label{eq:aut_invariant}
        \forall p \in P\colon|f_i(r_p)| \leq M_i(p).
      \end{equation}
      We note that the run~$\rho$ might not have been the ``\emph{most
      precise}'' run of~$\autof \net$, so the markings in the TPN run
      overapproximate the contents of $\autof \net$'s registers in~$\rho$.
      The invariant can be proved by induction.
      \claimqed

%
%
%
%
%
%
\end{enumerate}

\end{claimproof}

\cref{lem:tpn_coverabilityu_preserves_answer} and the observation that $\autof
\net$
is single-exponentially larger than~$\net$ (assuming binary encoding of the
numbers in~$\net$) conclude the proof ($\fomega$ is closed under primitive-recursive reductions).
\end{proof}

From \cref{lem:rsa-emptiness-in-fomega} and
\cref{lem:rsa-emptiness-fomega-hard}, we immediately obtain
\cref{thm:rsa-emptiness}.

\thmRsaEmptiness*

\vspace{-0.0mm}
\subsection{Proofs of Closure Properties}\label{sec:proofs_closure_properties}
\vspace{-0.0mm}

\thmRsaClosure*

\begin{proof}
  The proofs for closure under union and intersection are standard: for two
  $\rsa$s $\autof 1 = (Q_1, \regs_1, \Delta_1, I_1, F_1)$ and $\autof 2 = (Q_2,
  \regs_2, \Delta_2, I_2, F_2)$  with disjoint sets of states and registers, the
  $\rsa$ $\autof \cup$ accepting the union of their languages is obtained as
  $\autof \cup = (Q_1 \cup Q_2, \regs_1 \cup \regs_2, \Delta_1 \cup
  \Delta_2, I_1 \cup I_2, F_1 \cup F_2)$.
  Similarly, $\autof \cap$ accepting their intersection is constructed as
  the product
  $\autof \cap = (Q_1 \times Q_2, \regs_1 \cup \regs_2, \Delta', I_1 \times I_2,
  F_1 \times F_2)$ where
  $$
    \trans {(s_1, s_2)} a {\guardin_1 \cup \guardin_2,
    \guardnotin_1 \cup \guardnotin_2, \update_1 \cup \update_2} {(s'_1, s'_2)}
    \in \Delta'
  $$
  \begin{center}
    iff
  \end{center}
  $$
    \trans {s_1} a {\guardin_1, \guardnotin_1, \update_1} {s'_1} \in \Delta_1
    \quad\text{and}\quad
    \trans {s_2} a {\guardin_2, \guardnotin_2, \update_2} {s'_2} \in \Delta_2.
  $$
  Correctness of the constructions is clear.

  For showing the non-closure under complement, consider the
  language~$\langnegallrep$ from \cref{ex:lang_negallrep}, which can be accepted by
  $\rsa$.
  Let us show that for the complement of the language,
  namely, the language
  \begin{equation*}
    \langallrep = \set{\word \mid \forall i\exists j\colon i\neq j \land
    \projdataof{\wordof i} = \projdataof{\wordof j}}
  \end{equation*}
  where all data values appear at least twice, there is no $\rsa$ that can
  accept it.

  Our proof is a~minor modification of the proof of Proposition~3.2
  in~\cite{Figueira12}.
  In particular, we show that if~$\langallrep$ were expressible using an~$\rsa$,
  then we could construct an~$\rsa$ encoding accepting runs of a~Minsky
  machine.
  Since emptiness of an RsA is decidable (cf.~\cref{thm:rsa-emptiness}) and
  emptiness of a~Minsky machines is not, we would then obtain a~contradiction.

  Let us consider a~Minsky machine~$\minsky$ with two counters and instructions of the
  form $(q, \ell, q')$ where~$q$ and~$q'$ are states of~$\minsky$ and
  $\ell \in \set{\minskyinc, \minskydec, \minskyifzero} \times \set{1,2}$ is the
  corresponding counter operation.
  A~run of~$\minsky$ is a~sequence of instructions (which can be viewed as
  $\Sigma$-symbols) together with a~data value $d \in \datadom$ assigned to
  every symbol.
  The data values are used to match increments with decrements of the same
  counter (intuitively, we are trying to say that ``\emph{each increment is
  matched with a decrement}'', in order to express that the value of the counter
  is zero).
  For instance, consider the following run:
  \begin{equation}
    \begin{array}{c}
    (q_1, \minskyincof 1, q_2)\\
      12
    \end{array}
    \hspace*{-2mm}
    \begin{array}{c}
    (q_2, \minskyincof 1, q_3)\\
      42
    \end{array}
    \hspace*{-2mm}
    \begin{array}{c}
    (q_3, \minskydecof 1, q_3)\\
      12
    \end{array}
    \hspace*{-2mm}
    \begin{array}{c}
    (q_3, \minskyincof 2, q_2)\\
      17
    \end{array}
    \hspace*{-2mm}
    \begin{array}{c}
    (q_2, \minskydecof 1, q_1)\\
      42
    \end{array}
    \hspace*{-2mm}
    \begin{array}{c}
    (q_1, \minskyifzeroof 1, q_4)\\
      7
    \end{array}
  \end{equation}
  Here, the first increment of counter~1 is matched with the first decrement of
  the counter (both having data value~12) and the second increment of counter~1
  is matched with the second decrement of the counter (both having data
  value~42).
  Since all increments of the counter are uniquely matched with a~decrement, the
  test at the end is satisfied so~$\minsky$ would accept (we assume~$q_4$ is
  a~final state).
  To can accept such words, we can construct an automaton that checks the
  following properties of the input word:
  \begin{enumerate}
    \item  The first instruction is of the form~$(q_1, \cdot, \cdot)$ for~$q_1$
      being the initial state of~$\minsky$.
      \label{point:minsky:starting}
    \item  Each instruction of the form~$(\cdot, \cdot, q_i)$ is followed by an
      instruction of the form~$(q_i, \cdot, \cdot)$.
      \label{point:minsky:sequence}
    \item  All increments have different data values and all decrements have
      different data values.
      \label{point:minsky:different}
    \item  Between every two $(\cdot, \minskyifzeroof i, \cdot)$ instructions
      (or between the start and the first such an~$\minskyifzeroof i$ instruction),
      \begin{enumerate}
        \item  every $(\cdot, \minskydecof i, \cdot)$ needs to be preceded by
          an~$(\cdot, \minskyincof i, \cdot)$ instruction with the same data
          value and
          \label{point:minsky:dec_match}
        \item  every $(\cdot, \minskyincof i, \cdot)$ needs to be followed by
          a~$(\cdot, \minskydecof i, \cdot)$ instruction with the same data
          value.
          \label{point:minsky:inc_match}
      \end{enumerate}
  \end{enumerate}
  Properties~\ref{point:minsky:starting} and~\ref{point:minsky:sequence} can
  be easily expressed using an NFA and, therefore, also using a~$\drsa$.
  Property~\ref{point:minsky:different} is easily expressible using an~$\rsa$
  (in fact, using a~$\drsa$) that collects data values of increments and
  decrements of each counter in registers (we need two registers for every
  counter).
  Property~\ref{point:minsky:dec_match} is also expressible using an~$\rsa$
  (again, using a~$\drsa$) that collects the data values of decrements and
  whenever it reads an increment, it checks whether it has seen the increment's
  data value before.

  Let us now focus on Property~\ref{point:minsky:inc_match}.
  The negation of this property would be ``\emph{there is an increment not
  followed by a~decrement with the same data value}''.
  This negated property is essentially captured by the
  language~$\langnegallrep$ and so it is expressible using~$\rsa$ (in fact, it
  can be expressed by an~$\nra$ with guessing; or by a~simple $\nra$ provided
  that we change the accepted language to be prepended by a sequence of data
  values that will be used in the run, separated from the run by a~delimiter).
  Therefore, if an~$\rsa$ could accept the complement of~$\langnegallrep$, i.e.,
  the language~$\langallrep$, then we would be able to solve the emptiness
  problem of a~Minsky machine, which is a~contradiction.
\end{proof}

\thmRsanClosure*

\begin{proof}
  The proof of closure under union is the same as in the proof of
  \cref{thm:rsa-closure} with the exception that the results uses only registers
  $\regs_1$ (we assume $|\regs_1 | = n$): all references to registers
  $r \in \regs_2$ are changed to references to $f(r)$ where $f\colon \regs_2 \to
  \regs_1$ is an injection.

  To show non-closure under intersection, consider the two languages
  $$
  \lang^A_n = \set{w \mid \forall i < n\colon \projdataof{w_i} =
  \projdataof{w_{|w| - i + 1}}}
  ~
  \text{and}
  ~
  \lang^B_n = \set{w \mid \forall n \leq i < 2n\colon \projdataof{w_i} =
  \projdataof{w_{|w| - i + 1}}}
  $$
  Intuitively, $\lang^A_n$~is the language of words where the first~$n$ data
  values in the word are repeated (in the reverse order) at the end of the word
  and $\lang^B_n$ is the language of words where the $(n+1)$-th to $2n$-th data
  values are repeated (also in the reverse order) at the $2n$-th to $(n+1)$-th
  position from the end.
  Both languages can be expressed via~$\nraof n$, and therefore also via~$\rsaof n$.
  Their intersection is the language
  $\lang^{AB}_n = \set{w \mid \forall i < 2n\colon \projdataof{w_i} =
  \projdataof{w_{|w| - i + 1}}}$, which is the same as~$\lang^A_{2n}$ and clearly needs $2n$~registers.

Non-closure under complement follows from \cref{thm:rsa-closure} (its proof uses $\rsaof 1$).
\end{proof}

\thmDrsaClosure*

\begin{proof}
The proof of closure of $\drsa$ under union is standard.
Let $\autof 1$ and $\autof 2$ be two complete $\drsa$s, such that
$\autof 1 = (Q_1, \regs_1, \Delta_1, I_1, F_1)$ and $\autof 2 = (Q_2,
  \regs_2, \Delta_2, I_2, F_2)$, and their sets of states and registers
  are disjoint. The $\drsa$ $\autof \cup$ accepting the union of their
  languages is obtained as $\autof \cup = (Q_1 \times Q_2, \regs_1 \cup \regs_2, \Delta', I_1 \times I_2,
  F')$ where
  $$
    \trans {(s_1, s_2)} a {\guardin_1 \cup \guardin_2,
    \guardnotin_1 \cup \guardnotin_2, \update_1 \cup \update_2} {(s'_1, s'_2)}
    \in \Delta'
  $$
  \begin{center}
    iff
  \end{center}
  $$
    \trans {s_1} a {\guardin_1, \guardnotin_1, \update_1} {s'_1} \in \Delta_1
    \quad\text{and}\quad
    \trans {s_2} a {\guardin_2, \guardnotin_2, \update_2} {s'_2} \in \Delta_2,
  $$ and $F' = \{(q_1, q_2) \mid q_1 \in F_1 \lor q_2 \in F_2\}$.

The construction of the $\drsa$ $\autof \cap =$ $(Q_1 \times Q_2, \regs_1 \cup \regs_2, \Delta', I_1 \times I_2,
  F'_\cap)$ accepting the intersection of
$\langof{\autof 1}$ and $\langof{\autof 2}$ is similar to the construction
of $\autof \cup$, with the exception of~$F'_\cap$, which is obtained as
$F'_\cap = F_1 \times F_2$.

The complement of $\drsa$ is obtained in the standard way
by completing it and swapping final and non-final states.
Since the automaton is already deterministic, the
correctness of the construction is obvious.
\end{proof}

\thmDrsanClosure*

\begin{proof}
  The closure under complement is trivial (complete the $\drsa$ and swap final
  and non-final states; no new register is introduced).

  To show that $\drsaof n$ is not closed under intersection, we use languages
  $\lang^A_n$ and $\lang^B_n$ from the proof of \cref{thm:rsan-closure}.
  In particular, both these languages are in $\draof n$ (the $\dra_n$ needs more
  states than the corresponding $\nra_n$ because it cannot \emph{guess} where
  the final part of the word starts and needs to consider all posibilites,
  making the $\dra_n$ exponentially larger).
  Similarly as in the proof of \cref{thm:rsan-closure}, a~$\drsa$ for the
  intersection of the languages, the language~$\lang^{AB}_n$, would need at
  least~$2n$ registers.

  Non-closure under union follows from De Morgan's laws.
%
%
%
%
\end{proof}

%% file: figs/rsa_emptiness_example_rsa.tikz
\begin{tikzpicture}
  \tikzset{
    ->, 
    >=stealth',
    initial text=$ $, 
    node distance=30mm,
  }

  \node[initial,state] (q) {$q$};
  \node[state] (s) [right of=q] {$s$};
  \node[state, accepting] (t) [right of=s] {$t$};

 \draw (q) edge[loop above] node {$\pictrans a {} {r_1 \gets r_1 \cup \set{\inp} \\ r_2 \gets r_1 \cup r_2}$} (q);
 \draw (q) edge[below] node {$\pictrans b {} {r_1 \gets \emptyset \\ r_2 \gets r_1}$} (s);
 \draw (s) edge[loop above] node {$\pictrans a {} {r_1 \gets r_1 \cup \set{\inp} \\ r_2 \gets r_2}$} (s);
 \draw (s) edge[below] node {$\pictrans b {\inp \in r_1 \\ \inp \in r_2}{r_1 \gets \emptyset \\ r_2 \gets \emptyset}$} (t);
\end{tikzpicture}

%% file: figs/tpn_example_other.tikz
\begin{tikzpicture}
  [
    >=stealth',
  ]

\node[place,
    tokens=0,
    label=$q$] (stateq) at (2.75,0) {};


\node[place,
    tokens=0,
    fill=red,
    fill opacity=0.2,
    label=below:$r_1 \cap r_2$] (place1) at (0,-4) {};

\node[place,
    tokens=0,
    fill=yellow,
    fill opacity=0.2,
    label=below:$\cmplof{r_1} \cap r_2$] (place2) at (0,-6) {};

\node[place,
    tokens=0,
    fill=cyan,
    fill opacity=0.2,
    label=below:$r_1 \cap \cmplof{r_2}$] (place3) at (6,-4) {};

\node[place,
    tokens=0,
    fill=lime,
    fill opacity=0.2,
    label=below:$\cmplof{r_1}\cap \cmplof{r_2}$] (place4) at (6,-6) {};

\node[transition,
    minimum height=1.5mm,
    minimum width=5mm,
    fill=red,
    label=left:\tiny $t_1$] (trans1) at (0.5,-2) {};

\node[transition,
    minimum height=1.5mm,
    minimum width=5mm,
    fill=cyan,
    label=left:\tiny $t_2$] (trans2) at (2,-2) {};

\node[transition,
    minimum height=1.5mm,
    minimum width=5mm,
    fill=blue!15,
    label=left:\tiny $t_3$] (trans3) at (3.5,-2) {};

\node[transition,
    minimum height=1.5mm,
    minimum width=5mm,
    fill=lime,
    label=left:\tiny $t_4$] (trans4) at (5,-2) {};


\path[->]
    (stateq) edge[bend left=10] node[above] {} (trans1);
\path[->]
    (trans1) edge[bend left=10] node[above] {} (stateq);
\path[->]
    (stateq) edge[bend left=10] node[above] {} (trans2);
\path[->]
    (trans2) edge[bend left=10] node[above] {} (stateq);
\path[->]
    (stateq) edge[bend left=10] node[above] {} (trans3);
\path[->]
    (trans3) edge[bend left=10] node[above] {} (stateq);
\path[->]
    (stateq) edge[bend left=10] node[above] {} (trans4);
\path[->]
    (trans4) edge[bend left=10] node[above] {} (stateq);

\path[->,thick,red]
    (place3) edge[bend left] node[below, pos=0.35,red] {\tiny $m$} (trans1);
\path[->,thick,red]
    (trans1) edge node[right,pos=0.2,red] {\tiny $m$} (place1);
\path[->]
    (place1) edge[bend left=50] (trans1);
\path[->]
    (trans1) edge[bend right=30] (place1);

\path[->]
    (place3) edge (trans2);
\path[->]
    (trans2) edge (place1);
\path[->,thick,cyan]
    (place3) edge[bend left] node[above, pos=0.4,cyan] {\tiny $n$} (trans2);
\path[->,thick,cyan]
    (trans2) edge[bend right=20] node[left,cyan] {\tiny $n$} (place1);

\path[->,thick,blue!15]
    (place3) edge node[above, pos=0.2] {\tiny $o$} (trans3);
\path[->,thick,blue!15]
    (trans3) edge[bend right=10] node[left, pos=0.5] {\tiny $o$} (place1);
\path[->]
    (place2) edge[bend right] (trans3);
\path[->]
    (trans3) edge (place1);

\path[->,thick,lime!80!black]
    (place3) edge[bend right] node[above] {\tiny $p$} (trans4);
\path[->,thick,lime!80!black]
    (trans4) edge node[above, pos=0.6] {\tiny $p$} (place1);
\path[->]
    (place4) edge[bend left] node[above] {} (trans4);
\path[->]
    (trans4) edge node[above] {} (place3);

\end{tikzpicture}

%% file: figs/protogadget_lossyrem.tikz
\begin{tikzpicture}
  \tikzset{
    ->, 
    >=stealth',
    initial text=$ $, 
    node distance=30mm,
  }

  \node[initial,state] (q1) {$q_1$};
  \node[state] (q2) [right of=q1] {$q_2$};
  \node[accepting,state] (q3) [right of=q2] {$q_3$};

 \draw (q1) edge[below] node {$\pictrans a {\inp \in \regof p} {\regof \inp
  \gets \{\inp\} \\ \regof \tmp \gets \emptyset}$} (q2);
 \draw (q2) edge[loop above] node {$\pictrans a {\inp \in \regof p \\ \inp
  \notin \regof \inp} {\regof \tmp \gets \regof \tmp \cup \{\inp\}}$} (q2);
 \draw (q2) edge[below] node {$\pictrans a {} {\regof p \gets \regof \tmp}$} (q3);
\end{tikzpicture}

%% file: figs/protogadget_move.tikz
\begin{tikzpicture}
  \tikzset{
    ->, 
    >=stealth',
    initial text=$ $, 
    node distance=30mm,
  }

  \node[initial,state] (q1) {$q_1$};
  \node[accepting,state] (q2) [right of=q1] {$q_2$};

 \draw (q1) edge[below] node {$\pictrans a {} {\regof{p_1} \gets \emptyset \\
  \regof{p_2} \gets \regof{p_1} \cup \regof{p_2}}$} (q2);
\end{tikzpicture}

%% file: figs/protogadget_newtoken.tikz
\begin{tikzpicture}
  \tikzset{
    ->, 
    >=stealth',
    initial text=$ $, 
    node distance=30mm,
  }

  \node[initial,state] (q1) {$q_1$};
  \node[accepting,state] (q2) [right of=q1] {$q_2$};

 \draw (q1) edge[below] node {$\pictrans a {\{\inp \notin r\}_{r \in \regs}}
  {\regof p \gets \regof p \cup \{\inp\}\\~}$} (q2);
\end{tikzpicture}

%% file: app-expressivity.tex
\newcommand{\tablelanguages}[0]{
\begin{table}[t]
\caption{Distinguishing languages for a~selection of register automata models.
Grey cells denote that the result is implied from the class being a
  sub/super-class of another class where the result is established.
  $\draof{1}^{(=)}$ denotes both $\draof 1$ and $\draeqof 1$, similarly for
  $\uraof{1}^{(=)}$.
  $\drsaof{(1)}$ denotes both $\drsa$ and $\drsaof 1$.
  None of the languages is accepted by~$\dra$.
}
\label{tab:languages}
\begin{center}
\scalebox{0.9}{
\input{table-expressivity.tex}
}
\end{center}
\end{table}
}

\newcommand{\figExpressivePower}[0]{
\begin{figure}[t]
\begin{minipage}{7.5cm}
  \centering
\input{figs/expressive_power.tikz}
\caption{Hasse diagram comparing the expressive power of a~selection of register
  automata models with decidable emptiness problem.
  All inclusions are strict.
  Languages distinguishing the different models can be found in \cref{tab:languages}.
}
\label{fig:expressive_power}
\end{minipage}
\hfill
\begin{minipage}{6cm}
  \begin{center}
\scalebox{1}{
  \input{figs/varphi_1.tikz}
}
  \end{center}
  \vspace{-5mm}
  \caption{A $\drsaof 1$ recognizing $\langexanob$}
  \label{fig:drsa_langexanob}

  \begin{center}
  \scalebox{1}{
    \input{figs/neg_varphi_1.tikz}
  }
  \end{center}
  \vspace{-5mm}
  \caption{A $\drsaof 1$ recognizing $\langnegexanob$}
  \label{fig:drsa_langnegexanob}
\end{minipage}

\end{figure}
}

\newcommand{\figrsaexamples}[0]{
\begin{center}
  \begin{minipage}{6cm}
\scalebox{1}{
  \input{figs/varphi_1.tikz}
}
  \centering

   (a) A $\drsaof 1$ recognizing $\langexanob$.
  \end{minipage}
\hfill
  \begin{minipage}{6cm}
\scalebox{1}{
  \input{figs/neg_varphi_1.tikz}
}
  \centering

   (b) A $\drsaof 1$ recognizing $\langnegexanob$.
  \end{minipage}
\end{center}
}

\vspace{-0.0mm}
\section{Expressivity of Register Set Automata}\label{sec:expressivity-rsa}
\vspace{-0.0mm}

In this section, we position $\rsa$s in the landscape of automata over data
words.
For this, we use the languages introduced previously and some other languages
defined below; these languages are then used to distinguish RSAs from various register
automata models.
We focus on distinguishing RSAs from their closest siblings: NRAs/URAs, alternating
RAs (ARAs)~\cite{DemriL09}, and alternating RAs with guess and
spread ($\araguessspreadshort{}$)~\cite{Figueira12} (due to space constraints, this listing is by no means
exhaustive).

\begin{example}\label{ex:lang_exanob}
  First, we define the language over $\Sigma = \set{a,b}$

\begin{equation*}
  \hlmathbox{
\langexanob = \set{w \mid \exists i\colon \projof a {w_i} \land \nexists j <
i\colon \projof b {w_j} \land \projdataof{w_j} = \projdataof{w_i}}
  }
\end{equation*}

\noindent
from~\cite[Proof of Proposition~3.2]{Figueira12}
and its complement $\langnegexanob$.
Intuitively, $\langexanob$~is the language of words~$w$ such that there exists
an input element~$\pair a d$ that is not preceded by an occurrence of a~$\pair b
d$ element.
Neither~$\langexanob$ nor~$\langnegexanob$ can be accepted by $\araof 1$ while
$\araguessspreadshortof 1$~accepts~$\langexanob$ but cannot accept
$\langnegexanob$.
On the other hand, as shown
in \cref{fig:drsa_langexanob,fig:drsa_langnegexanob}, $\drsaof 1$~can accept
both~$\langexanob$ and~$\langnegexanob$.

\begin{table}[t]
\caption{Distinguishing languages for a~selection of register automata models.
Grey cells denote that the result is implied from the class being a
  sub/super-class of another class where the result is established.
  $\draof{1}^{(=)}$ denotes both $\draof 1$ and $\draeqof 1$, similarly for
  $\uraof{1}^{(=)}$.
  $\drsaof{(1)}$ denotes both $\drsa$ and $\drsaof 1$.
  None of the languages is accepted by~$\dra$.
}
\label{tab:languages}
\begin{center}
\scalebox{0.9}{
\input{table-expressivity.tex}
}
\end{center}
\end{table}

  (It might seem suspicious that $\drsaof 1$ can express the
  language~$\langnegexanob$, since according to~\cite[Proof of
  Proposition~3.2]{Figueira12}, if an~$\araof 1$~$\aut$ could accept the language,
  $\aut$~could be used to decide language-emptiness of a~Minsky machine.
  So how come that we can express~$\langnegexanob$ using~$\drsaof 1$, which have
  a~decidable emptiness problem (cf.\ \cref{thm:rsa-emptiness})?
  The reason is that in the construction of the automaton representing the accepting
  runs of a~Minsky machine from~\cite{Figueira12}, apart from~$\langnegexanob$,
  we also need to be able to express the property ``\emph{every counter
  increment is matched with its decrement}'', which is not expressible by $\rsa$
  (cf.\ the proof of \cref{thm:rsa-closure}).)
  \qed
\end{example}

\begin{example}\label{ex:lang_allaexb}
  Moreover, let us define the following language
  from~\cite[Example~2.2]{DemriL09}:


  \begin{center}
  \hlbox{
    \begin{minipage}{7cm}
      \vspace*{-2mm}
    \begin{align*}
      \langallaexb = \big\{w \mid \forall i \colon \projof a {w_i} \implies
      \big( &(\forall j \neq i \colon \projof a {w_j} \implies
      \projdataof{w_i} \neq \projdataof{w_j}) \land {} \\
      & (\exists k > i \colon \projof b {w_k} \land \projdataof{w_i} = \projdataof{w_k})\big)
      \big\}
    \end{align*}
    \end{minipage}
}
  \end{center}

\noindent
Intuitively, $\langallaexb$ denotes the language of words where no two
$a$-positions contain the same data value and every $a$-position is followed by
a~matching $b$-position.
$\langallaexb$~is recognizable by~$\araof 1$~\cite[Example~2.6]{DemriL09}, but
  is not recognizable by any $\ura$, $\nra$, or $\rsa$.
\qed
\end{example}

The Hasse diagram comparing the expressive power of a~selection of register
automata models with decidable emptiness is in
\cref{fig:expressive_power} and languages distinguishing the various classes
are in \cref{tab:languages}.

\begin{remark}
  $\rsa$s are also incomparable to the class of \emph{pebble
  automata}~\cite{NevenSV04}, since $\drsa$s generalize $\dra$s, they can accept
  a~language not expressible by $\pa$s (as shown by Tan in~\cite{Tan13}).
  On the other hand, $\rsa$s~cannot express~$\langallrep$, which is expressible
  by~$\pa$s.
  \qed
\end{remark}

\begin{figure}[t]
\begin{minipage}{7.5cm}
  \centering
\input{figs/expressive_power.tikz}
\caption{Hasse diagram comparing the expressive power of a~selection of register
  automata models with decidable emptiness problem.
  All inclusions are strict.
  Languages distinguishing the different models can be found in \cref{tab:languages}.
}
\label{fig:expressive_power}
\end{minipage}
\hfill
\begin{minipage}{6cm}
  \begin{center}
\scalebox{1}{
  \input{figs/varphi_1.tikz}
}
  \end{center}
  \vspace{-5mm}
  \caption{A $\drsaof 1$ recognizing $\langexanob$}
  \label{fig:drsa_langexanob}

  \begin{center}
  \scalebox{1}{
    \input{figs/neg_varphi_1.tikz}
  }
  \end{center}
  \vspace{-5mm}
  \caption{A $\drsaof 1$ recognizing $\langnegexanob$}
  \label{fig:drsa_langnegexanob}
\end{minipage}

\end{figure}


We also considered several extensions of the $\rsa$ model, such as RSAs
\begin{inparaenum}[(i)]
  \item  with register emptiness test,
  \item  with register equality test,
  \item  with removal, and 
  \item  with removal and register emptiness test.
\end{inparaenum}
Since they are not directly relevant to regex matching, we give details in~\cref{sec:app_extensions}.

%% file: table-expressivity.tex
\begin{tabular}{l|lllllll}
  \toprule
  Language          & $\nraof{1}^{(=)}$                   & $\uraof{1}^{(=)}$                   & $\booleanof{\nraeqof 1}$             & $\araof 1$               & $\drsaof{(1)}$                         & $\rsaof 1$                          & $\araguessspreadshortof 1$ \\
  \midrule
  $\langexrep$      & \tabyes~Ex.~\ref{ex:ra_example}     & \tabno~Ex.~\ref{ex:ra_example}      & \greycell \tabyes                    & \greycell \tabyes        & \tabyes~Ex.~\ref{ex:drsa_langexrep}    & \greycell \tabyes                   & \greycell \tabyes     \\
  $\langnegexrep$   & \tabno~Ex.~\ref{ex:ra_example}      & \tabyes~Ex.~\ref{ex:ra_example}     & \greycell \tabyes                    & \greycell \tabyes        & \tabyes~Ex.~\ref{ex:drsa_langnegexrep} & \greycell \tabyes                   & \greycell \tabyes     \\
  $\langexnegexrep$ & \tabno~Ex.~\ref{ex:lang_exnegexrep} & \tabno~Ex.~\ref{ex:lang_exnegexrep} & \tabyes~Ex.~\ref{ex:lang_exnegexrep} & \greycell \tabyes        & \greycell \tabyes                      & \greycell \tabyes                   & \greycell \tabyes     \\
  $\langallrep$     & \greycell \tabno                    & \greycell \tabno                    & \greycell \tabno                     & \tabno~\cite{Figueira12} & \tabno~Thm.~\ref{thm:drsa_lt_rsa}      & \tabno~Thm.~\ref{thm:rsa-closure}   & \tabno~\cite{Figueira12}     \\
  $\langnegallrep$  & \greycell \tabno                    & \greycell \tabno                    & \greycell \tabno                     & \tabno~\cite{Figueira12} & \tabno~Thm.~\ref{thm:drsa_lt_rsa}      & \tabyes~Ex.~\ref{ex:lang_negallrep} & \tabyes~\cite{Figueira12}    \\
  $\langexanob$     & \greycell \tabno                    & \greycell \tabno                    & \greycell \tabno                     & \greycell \tabno         & \tabyes~Ex.~\ref{ex:lang_exanob}       & \greycell \tabyes                   & \tabyes~\cite{Figueira12} \\
  $\langnegexanob$  & \greycell \tabno                    & \greycell \tabno                    & \greycell \tabno                     & \greycell \tabno         & \tabyes~Ex.~\ref{ex:lang_exanob}       & \greycell \tabyes                   & \tabno~\cite{Figueira12}  \\
  $\langallaexb$    & \tabno~Ex.~\ref{ex:lang_allaexb}    & \tabno~Ex.~\ref{ex:lang_allaexb}    & \greycell \tabno                     & \tabyes~\cite{DemriL09}  & \greycell \tabno                       & \tabno~Ex.~\ref{ex:lang_allaexb}    & \greycell \tabyes \\
  \bottomrule
\end{tabular}

%% file: figs/expressive_power.tikz
\begin{tikzpicture}
[
  node distance=20mm,
  on grid,
  auto
]
  \newcommand{\scalefactor}[0]{0.8}

  \tikzstyle{autclass}=[
    draw,
    rectangle,
    minimum height=6mm,
    minimum width=12mm
  ]

  \newcommand{\yshift}[0]{8mm}

  \coordinate (start);
  \node[autclass] (rsa)    [right of=start] {$\rsa$};
  \node[autclass] (drsa)   [below of=start,yshift=\yshift] {$\drsa$};
  \node[autclass] (drsa1)  [below of=drsa,xshift=0mm,yshift=\yshift] {$\drsaof 1$};
  \node[autclass] (dra1)   [below of=drsa1,xshift=0mm,yshift=\yshift] {$\draof 1$};
  \node[autclass] (dra1eq) [below of=dra1,yshift=\yshift] {$\draeqof 1$};

  \coordinate[right of=rsa] (tmp1);
  \node[autclass] (ara1guessspread) [right of=rsa] {$\araguessspreadshortof 1$};
  \node[autclass] (ara1)   [below of=ara1guessspread, yshift=\yshift] {$\araof 1$};

  \node[autclass] (nra1eq) [right of=dra1,yshift=0mm] {$\nraeqof 1$};
  \node[autclass] (ura1eq) [right of=nra1eq,yshift=-0mm] {$\uraeqof 1$};

  \node[autclass] (nra1)   [right of=drsa1,xshift=-0mm] {$\nraof 1$};
  \node[autclass] (boolean)[right of=nra1,yshift=-0mm] {$\booleanof{\nraeqof 1}$};


  \draw (drsa)    edge (rsa);
  \draw (ura1eq)  edge (boolean);
  \draw (drsa1)   edge (drsa);
  \draw (nra1eq)  edge (boolean);
  \draw (boolean) edge (drsa);
  \draw (boolean) edge (ara1);
  \draw (dra1eq)  edge (nra1eq);
  \draw (dra1eq)  edge (ura1eq);
  \draw (ara1)    edge (ara1guessspread);
  \draw (dra1)    edge (drsa1);
  \draw (dra1eq)  edge (dra1);
  \draw (nra1eq)  edge (nra1);
  \draw (nra1)    edge (ara1);
\end{tikzpicture}

%% file: figs/varphi_1.tikz
\begin{tikzpicture}
  \tikzset{
    ->, 
    >=stealth',
    initial text=$ $, 
    node distance=30mm,
  }

  \node[initial,state] (q) {$q$};
  \node[state, accepting] (s) [right of=q] {$s$};

 \draw (q) edge[loop above] node[yshift=-4mm] {$\pictrans a {in \in r_b} {}$} (q);
 \draw (q) edge[loop below] node {$\pictrans b {} {r_b \gets r_b \cup \{\inp\}}$} (q);
 \draw (q) edge node[above,yshift=-4mm] {$\pictrans a {\inp \notin r_b} {}$} (s);
 \draw (s) edge[loop above] node[yshift=-4mm] {$\pictrans \Sigma {} {}$} (s);

\end{tikzpicture}

%% file: figs/neg_varphi_1.tikz
\begin{tikzpicture}
  \tikzset{
    ->, 
    >=stealth',
    initial text=$ $, 
    node distance=30mm,
  }

  \node[initial,state,accepting] (q) {$q$};

 \draw (q) edge[loop above] node[yshift=-4mm] {$\pictrans a {in \in r_b} {}$} (q);
 \draw (q) edge[loop below] node {$\pictrans b {} {r_b \gets r_b \cup \{\inp\}}$} (q);

\end{tikzpicture}

%% file: app_proof_det.tex
\section{Proof of \cref{thm:alg_soundness}}\label{sec:det-proof}

In this section, we prove the correctness of \cref{alg:det}. First, we introduce the notion 
of normalized copyless RAs simplifying the RA transition structure. 

\begin{lemma}\label{lem:norm-single-val}
  Let $\aut = (\cdot, \{ r_1, \dots, r_k \}, \Delta, \cdot, \cdot)$ be a copyless RA. $\aut$ is equivalent to a RA obtained from $\aut$
  removing transitions from $\Delta$ s.t. 
  \begin{enumerate}[(i)]
    \item $\trans \cdot a {\guardeq, \guardneq, \update} \cdot$ where $|\guardeq| > 1$,
    \item $\trans \cdot a {\guardeq, \guardneq, \update} \cdot$ where $\update(r_i) = \inp$ and $\update(r_j) \in \guardeq$ for $i \neq j$, or
    \item $\trans \cdot a {\guardeq, \guardneq, \update} \cdot$ where $\update(r_i) = \update(r_j)$ for $i \neq j$.
  \end{enumerate}
\end{lemma}
\begin{proof}
  From the definition of copyless RAs we have that there is no reachable configuration $(q,f)$ s.t. $f(r_1) = f(r_2)$ for 
  registers $r_1 \neq r_2$. Therefore, transitions with $|\guardeq| > 1$ are never used. Similarly, the cases (ii) and (iii) 
  lead to a situation where the output configuration violates the condition of copyless RA.
\end{proof}

\begin{figure}[t]
  \begin{tikzpicture}
    \node (r) at (0, 1) {$r_1$};
    \node (s) at (1, 1) {$r_2$};
    \node (t) at (2, 1) {$r_3$};
    \node (u) at (3, 1) {$r_4$};

    \node (v) at (0.5, 0) {$r_1'$};
    \node (w) at (1.5, 0) {$r_2'$};
    \node (in) at (2.5, 0) {$\inp$};

    \node (up) at (-0.5, 0.5) {$\update:$};

    \draw[->, out=270, in=110] (r) to (v);
    \draw[->, out=270, in=110] (s) to (w);
    \draw[->, out=270, in=110] (t) to (in);
    \draw[->, red, dashed, out=250, in=90] (u) to (w);
  \end{tikzpicture}
  \label{fig:single-val-up}
  \caption{Example of a transition update $\update$ of a copyless RA. The red dashed arrow is forbidden in normalized copyless RAs.}
\end{figure}

A copyless RA is called \emph{normalized} if (i) it does not contain 
transitions mentioned in \cref{lem:norm-single-val}, and (ii) for each transition $t$ and the corresponding guard $\guardeq$ 
and update function $\update$, $\img(\update) \cap \guardeq = \emptyset$. Note that arbitrary copyless RA can be normalized by 
removing transitions of \cref{lem:norm-single-val} and replace each transition update to $\update[\guardeq/\inp]$.
Note that for the sake of simplicity we assume that the transition updates do not contain $\bot$ values (this means that the 
transition updates are partial functions instead of functions assigninng $\bot$ to unused registers). Example of an 
update function violating the normalization conditions is shown in \cref{fig:single-val-up} (dashed red arrow). 
Let $\aut = (Q, \{ r_1, \dots, r_k \}, \cdot, \cdot, \cdot)$ be a copyless RA. 
For a state $q$, we use $\sngconf[q]$ to denote the set of copyless configurations corresponding to the state $q$.
In particular, $\sngconf[q] = \{ h \mid \dom(h) = \regs[q], h(r) \in \datadom, h(r_1) \neq h(r_2) \text{ for all } r_1 \neq r_2 \}$.
We say that a set of configurations $\confs$ corresponding to a state $q$ is $\beta$-\emph{cartesian-representable}
if $\dom(\beta) = \regsof{q}$ and $\confs = \{ (q,f) \mid f(r) \in \beta(r) \text{ for } r \in \regs, f(r_i) \neq 
f(r_j) \text{ for } r_i\neq r_j \}$. Further, for a set of registers $R$ we say that $\confs$ is
\emph{$R$-singleton} cartesian-representable if $|\beta(r)| \leq 1$ for each $r \in R$.

\begin{lemma}\label{lem:up-inv}
  Let $\aut$ be a normalized copyless RA and let $t = \trans \cdot a {\guardeq, \guardneq, \update} \cdot$ be a 
  transition. Then $\update^{-1}$ is a function.
\end{lemma}
\begin{proof}
 Follows directly from \cref{lem:norm-single-val}.
\end{proof}

Let $\aut (Q, \{ r_1, \dots, r_k \}, \cdot, \cdot, \cdot)$ be a copyless RA. For a transition $t$ of $\aut$, symbol $\pair a d$, 
and a set of configurations $C$, we use $\post$ to denote the 
set of successor configurations of $C$ over $t$. In particular, 
$
  \post(C, t, \pair a d) = \{ (q', h') \mid (q, h) \in C, (q,h) \stepofusing{\pair {a} {d}}{t} (q', h') \}
$
We lift the defnition to a set of transitions in an usual way.
Moreover, for a partial function $f: \regs\cup \{ \inp \} \to 2^\datadom$ and a value $d \in \datadom$ we define $f^d$ as 
$$
f^d(r) =
\begin{cases}
  \{ d \} & \text{if } r = \inp,\\
  f(r) & \text{otherwise } .
\end{cases}
$$

\begin{lemma}\label{lem:cart-repr}
  Let $\aut = (Q, \regs, \Delta, \cdot, \cdot)$ be a normalized copyless RA, $t = \trans q a {\guardeq, \guardneq, \update} q' \in \Delta$ be a 
  transition, and let $\confs$ be a $\guardneq$-singleton $\beta$-cartesian-representable set of configurations s.t. $d\in \beta(r)$ for 
  all $r \in \guardeq$ and $d \notin \beta(r)$ for all $r \in \guardneq$. Then the set $\post(C, t, \pair a d)$ is
  $\theta$-cartesian-representable for each $d\in\datadom$ where $\theta(r) = \beta^d(\update(r))$.
\end{lemma}
\begin{proof}
  Since $\aut$ is normalized, from \cref{lem:norm-single-val} we have that 
  $\guardeq = \{ r_i \}$ or $\guardeq = \emptyset$. First, we investigate the case $\guardeq = \{ r_i \}$.
  As $\aut$ is normalized, we have that $r \notin \img(\update)$, which means that $t$ is enabled (given also by the assumptions of $\guardneq$ values)
  but the concrete set of $\beta(r)$ is not important (the values of $\beta(r)$ are not used in the update).
  From \cref{lem:up-inv} we have that $\update^{-1}$ is a function meaning that each $r$ is copied to 
  at most one (possible other) register $r'$ or $\inp$. Therefore the register values of configurations 
  from $\post(C, t, \pair a d)$ contain all copied values given by $\beta$. Therefore, the set 
  is $\theta$-cartesian-representable where $\theta(r) = \beta^d(\update(r))$.
\end{proof}


\begin{lemma}\label{lem:det-fst-incl}
  Assume that \cref{alg:det} does not return $\bot$. Then $\langof{\aut'} \subseteq \langof{\aut}$.
\end{lemma}
\begin{proof}
  Let $w$ be a dataword s.t. $w \in \langof{\aut'}$. We show that $w \in \langof{\aut}$. 
  Since $w \in \langof{\aut'}$, there is a run over $w$ of the form 
  $$
    ((S_1, c_1), f_1) \stepofusing{\pair {a_1} {d_1}}{t_1} ((S_2, c_2), f_2)\stepofusing{\pair {a_2} {d_2}}{t_2} \dots  \stepofusing{\pair {a_{k-1}} {d_{k-1}}}{t_{k-1}} ((S_k, c_k), f_k)
  $$
  where $S_i \subseteq Q$, $c$ is the register-cardinality mapping from \cref{alg:det}, and $t_i$ is 
  a transition of the form $\trans {(S_i, c_i)} {a_i} {\guardeq, \regs \setminus \guardeq, \update} {(S_{i+1}, c_{i+1})}$. 
  Since the macrostates encode a set of configuration of the original RA $\aut$, we use 
  $\confsenc$ to denote the set of configuration covered in the given macrostate. Formally,
  $$
    \confsenc((S, c), f) = \{ (q, h) \mid q \in S, h \in \sngconf[q], h(r) \in f(r) \}
  $$
  Intuitively, $\confsenc$ contains all copyless configurations obtained from picking arbitrary register values from 
  the register-sets of the macrostate. Note that, the value of $c$ does not affect the set of configurations as 
  it only expresses some properties of $f$. 

  \begin{claim}\label{claim:prop-c}
    For each $((S_i, c_i), f_i)$ we have that $c_i(r) = 0$ then $|f_i(r)| = 0$ and 
    $c_i(r) = 1$ then $|f_i(r)| = 1$.
  \end{claim}
  \begin{claimproof}
    Follows from \lnref{ln:update_collapse} and \lnref{ln:counter_update} of \cref{alg:det}.
  \end{claimproof}

  \begin{claim}\label{claim:confs-repr}
    $\confsenc((S, c), f)$ is $f_{|\regs[q]}$-cartesian-representable for each $q \in S$.
  \end{claim}
  \begin{claimproof}
    Follows directly from the definition of $\confsenc((S, c), f)$.
  \end{claimproof}

  \begin{claim}\label{claim:conf-post}
    Fix a configuration $((S_i, c_i), f_i)$ and let $T^\bullet$ be a set of transitions of $\aut$ used to construct 
    $t_i$ in \cref{alg:det}. Then, the following equality holds:
    $$
      \confsenc((S_{i+1}, c_{i+1}), f_{i+1}) = \post(\confsenc((S_i, c_i), f_i), T^\bullet, \pair {a_i}{d_i})
    $$
  \end{claim}
  \begin{claimproof}
    We start with the following reasoning
    \begin{align}
      \post(\confsenc((S_i, c_i), f_i), T^\bullet, \pair {a_i}{d_i}) = \bigcup_{t' \in T^\bullet} \post(\confsenc((S_i, c_i), f_i), t', \pair {a_i}{d_i})
    \end{align}
    From the definition of $\confsenc$ we have that $\confsenc((S_i, c_i), f_i)$ is $f_i$-cartesian-representable for $S_i$.
    Assume that $t_i = \trans {(S_i, c_i)} {a_i} {g, \regs \setminus g, \update} {(S_{i+1}, c_{i+1})}$. Since $t_i$ is enabled 
    (assumption of existence of the run), $c_i(r) \leq 1$ for each $r \in \regs \setminus g$ meaning that 
    $|f_i(r)| \leq 1$ for each $r \in \regs \setminus g$ (\cref{claim:prop-c}). Therefore, $\confsenc((S_i, c_i), f_i)$ is $\regs\setminus g$-singleton 
    cartesian-representable set of configurations. Now assume $t' = \trans q {a_i} {\guardeq, \guardneq, \update'} {q'}$. 
    From \cref{alg:det} we have that $\guardeq \subseteq g$ and $g \cap \guardneq \subseteq \regs \setminus g$. 
    From \cref{claim:confs-repr} we have that $\confsenc((S_i, c_i), f_i)$ is $(f_i)_{|\regs[q]}$-cartesian-representable and 
    moreover from $g \cap \guardneq \subseteq \regs \setminus g$ we get that this set is also $\guardneq$-singleton 
    cartesian-representable.
    From \cref{lem:cart-repr} we hence have that 
    $\post(\confsenc((S_i, c_i), f_i), t', \pair {a_i}{d_i})$ is $\theta$-cartesian-representable where $\theta(r) = f_i^{d_i}(\update'(r))$.
    Now set
    $$
      V(r) = \bigcup_{t' \in T^\bullet, t' = \trans \cdot {a_i} {\cdot, \cdot, \update'} {\cdot}} f_i^{d_i}(\update'(r))
    $$
    We show that $\confsenc((S_{i+1}, c_{i+1}), V) = \post(\confsenc((S_i, c_i), f_i), T^\bullet, \pair {a_i}{d_i})$. 
    The direction right to left is trivial. Now we prove the opposite direction. Let $(q', h') \in \confsenc((S_{i+1}, c_{i+1}), V)$. 
    From the definition of $\sngconf[q]$ we have that $\dom(h') = \regs[q']$. We also have that for each $r \in \regs[q']$ there is 
    a transition $t_r \in T^{\bullet}$ with the update function $\update_{r}$ s.t. $h'(r) \in f_i^{d_i}(\update_r(r))$. Now assume that $\regs[q'] = \{ r_1, \dots, r_n \}$. 
    We set $(x_1, \dots, x_n) = (\update_{r_1}(r_1), \dots, \update_{r_n}(r_n))$. From the definition of $\mathit{op}_{r_i}$ we have 
    that $(x_1, \dots, x_n) \in P$ ($P$ is given by \lnref{ln:cartesian}). Then since the algorithm does not return $\bot$ (assumption on the existence of $w$),
    there is a single transition $t \in T^\bullet$ with the target $q'$ and having the update function $\update_t$ s.t. $h'(r) \in f_i^{d_i}(\update_t(r))$ for each $r \in \regs[q']$.
    Hence $(q', h') \in \post(\confsenc((S_i, c_i), f_i), T^\bullet, \pair {a_i}{d_i})$.

    The last tile to the puzzle is the observation that $V = f_{i+1}$. Indeed $\update'$ constructed on \lnref{ln:update} gathers updates of all transitions 
    from $T^\bullet$ and then $f_{i+1}$ unions the register values from $f_i$ according to $\update'$, which mathes the definition of $V$.
  \end{claimproof}

  Now we can finish the original claim of the lemma. We construct the run of $\aut$ as follows. We pick a configuration from $\confsenc((S_k, c_k), f_k)$. 
  Then, from \cref{claim:conf-post} we have that there is a configuration of $\aut$ from $\confsenc((S_{k-1}, c_{k-1}), f_{k-1})$. s.t. 
  these two configurations are connected with a transition over $\pair {a_{k-1}} {d_{k-1}}$. Using this approach we iteratively construct an 
  accepting run of $\aut$ (recall that $\confsenc((S_1, c_1), f_1)$ contains only initial configurations).
\end{proof}

\begin{lemma}\label{lem:det-snd-incl}
  Assume that \cref{alg:det} does not return $\bot$. Then $\langof{\aut} \subseteq \langof{\aut'}$.
\end{lemma}
\begin{proof}
  Assume that $w \in \langof{\aut}$. Then, there is a run over $w$
  $$
    (q_1, h_1) \stepofusing{\pair {a_1} {d_1}}{t_1} (q_1, h_1)\stepofusing{\pair {a_2} {d_2}}{t_2} \dots \stepofusing{\pair {a_{k-1}} {d_{k-1}}}{t_{k-1}} (q_k, h_k)
  $$
  We construct an accepting run of $\aut'$. Since \cref{alg:det} does not return $\bot$, we can choose 
  almost-arbirary set $g$ on \lnref{ln:guards}. Note that picking of the set $g_i$ for each step then uniquely 
  identifies the run of $\aut'$. We set $g_i = \guardeq_i \cup \{ r \mid h_i(r) = d_i \}$ where $\guardeq$
  is the guard of the transition $t_i$. 
  Using $g_i$s we can construct the run of $\aut'$:
  $$
    ((S_1, c_1), f_1) \stepofusing{\pair {a_1} {d_1}}{T_1} ((S_2, c_2), f_2)\stepofusing{\pair {a_2} {d_2}}{T_2} \dots \stepofusing{\pair {a_{k-1}} {d_{k-1}}}{T_{k-1}} ((S_k, c_k), f_k)
  $$
  (all items except the choice of $g_i$ are given implicitly by the construction in \cref{alg:det}). 
  From the construction of transitions $T_i$ of $\aut'$ we have that 
  the set $T^\bullet$ will contain $t_i$ during the $i$-the iteration and $q_k \in S_k$ meaning that this run is accepting 
  and $w \in \langof{\aut'}$.
\end{proof}

\begin{corollary}
  The \cref{alg:det} is correct.
\end{corollary}
\begin{proof}
  Follows directly from \cref{lem:det-fst-incl} and \cref{lem:det-snd-incl}.
\end{proof}



%% file: figs/ra_nodet.tikz
\begin{tikzpicture}
  \tikzset{
    ->, 
    >=stealth',
    initial text=$ $, 
    node distance=25mm,
  }

  \node[initial,state] (q1) {$1$};
  \node[state] (q2) [right of=q1, yshift=5mm] {$2$};
  \node[state] (q4) [right of=q1, yshift=-5mm] {$4$};
  \node[state,accepting] (q3) [right of=q2] {$3$};
  \node[state,accepting] (q5) [right of=q4] {$5$};


 \draw (q1) edge[above] node[yshift=-1mm] {$\pictrans a {\top} {r \gets \inp}$} (q2);
 \draw (q1) edge[below] node[yshift=-1mm] {$\pictrans a {\top} {\emptyset}$} (q4);
 \draw (q2) edge[above] node[yshift=-1mm] {$\pictrans a {\top} {r \gets r}$} (q3); 
 \draw (q4) edge[above] node[yshift=-1mm] {$\pictrans a {\top} {r \gets \inp}$} (q5); 
\end{tikzpicture}

%% file: figs/disjoint.tikz
\begin{tikzpicture}
  \tikzset{
    ->, 
    initial text=$ $, 
    >=stealth',
    node distance=30mm,
  }

  \node[initial,state] (q) {$q$};
  \node[state, accepting] (s) [right of=q] {$s$};

 \draw (q) edge[loop above] node {$\pictrans a {} {r \gets r \cup \{\inp\}}$} (q);
 \draw (q) edge node[above,yshift=-4mm] {$\pictrans a {\inp \notin r} {}$} (s);
 \draw (s) edge[loop above] node {$\pictrans a {} {\inp \notin r}$} (s);

\end{tikzpicture}

%% file: figs/postproc_example.tikz
\begin{tikzpicture}
  \tikzset{
    ->, 
    >=stealth',
    initial text=$ $, 
    node distance=30mm,
  }

  \node[state,initial] (q) {$q$};
  \node[state] (s) [right of=q] {$s$};
  \node[state](t) [right of=s] {$t$};
  \node[state,accepting](u) [right of=t] {$u$};
  \path[->]

  (q) edge[above] node {$\pictrans{a}{}{r^s_1 \gets \inp}$} (s)
  (s) edge[above] node {$\pictrans{b}{}{r^t_1 \gets r^s_1}$} (t)
  (t) edge[loop above] node {$\pictrans{a}{}{r^t_1 \gets r^t_1}$} (t)
  (t) edge[above] node {$\pictrans{a}{}{r^u_1 \gets r^t_1 \\ r^u_2 \gets \inp}$} (u)
  (u) edge[loop above] node {$\pictrans{a}{}{r^u_1 \gets r^u_1 \\ r^u_2 \gets r^u_2}$} (u);
\end{tikzpicture}

%% file: app_extensions.tex
\vspace{-0.0mm}
\subsection{RsAs with Register Emptiness Test}\label{sec:rsa-empty}
\vspace{-0.0mm}
An \emph{$\rsa$ with register emptiness test} ($\rsaempty$) is a~tuple $\aut_E
= (Q, \regs, \Deltaempty, I, F)$ where
$Q, \regs, I, F$ are the same as for $\rsa$s and the transition relation $\Deltaempty$ for $\rsaempty$ is defined as
$\Deltaempty \subseteq Q \times \Sigma \times 2^\regs \times 2^\regs \times 2^\regs \times (\regs \to 2^{\regs \cup \set{\inp}}) \times Q$.
The semantics of a~transition $\trans q a {\guardin, \guardnotin, \guardempty, \update}{s}$
is such that~$\aut_E$ can move from state~$q$ to state~$s$ if the $\Sigma$-symbol at
the current position of the input word is~$a$ and the $\datadom$-symbol at the
current position is in all registers from~$\guardin$ and in no register
from~$\guardnotin$ and all registers from~$\guardempty$ are empty; the content of the
registers is updated so that $r_i \gets \bigcup \set{x \mid x \in \update(r_i)}$.

\begin{lemma}\label{lem:todo}
For every $\rsaempty$ $\aut$, there exists an $\rsa$ $\aut'$ with
the same language.
\end{lemma}

\begin{proof}
Proof is done by showing the construction of modified $\rsa$ $\aut'$
corresponding to any given $\rsaempty$ $\aut^{=\emptyset}$.

The modification of $\aut'$ appears in the the structure of states,
where we code the information about the empty registers into the states themselves,
and modify the transition relation accordingly.
The information about the emptiness of registers is kept in a~form of
binary vector, denoted $\vectof {\emptyset}$, such that $\vectof {\emptyset}[r_i] = 0$ iff
$r_i = \emptyset$, and $\vectof {\emptyset}[r_i] = 1$ otherwise.
s
Let $\rsaempty$ be a~tuple $\aut^{=\emptyset} = (Q, \regs, \Delta, I, F)$
then the equivalent $\rsa$ is a~tuple
$\aut' = (Q', \regs, \Delta', I', F')$ where
$Q' = Q \times \vectof {\emptyset}$,
$I' = \{(q_i, \vectof {\emptyset}^0) \mid q_i \in I \land \forall r_i \in \regs: \vectof {\emptyset}[r_i]=0)\}$,
$F' = \{(q_f, \vectof {\emptyset})\mid q_f \in F\}$,
and $\Delta' =
  \{\trans{(q_1, \vectof {\emptyset}^1)}{a}{\guardin, \guardnotin, \update}{(q_2, \vectof {\emptyset}^2)} \mid
  \trans{q_1}{a}{\guardin, \guardnotin, S, \update}{q_2} \in \Delta \land \forall r \in S: \vectof {\emptyset}^1[r] = 0\}$, such that
  \begin{equation}
    \vectof {\emptyset}^2[r_i] =
    \begin{cases}
    1 & \text{if } (r_i \mapsto \{r_j\} \in \update \land \vectof {\emptyset}^1[r_j]
        \neq 0)~\lor \\
      &  (r_i \mapsto \{\inp\} \in \update) \text{ and} \\
    0 & \text{otherwise}.
\end{cases}
\end{equation}
\end{proof}

\vspace{-0.0mm}
\subsection{RsAs with Register Equality Test}\label{sec:label}
\vspace{-0.0mm}

An \emph{$\rsa$ with register equality test} ($\rsaeqreg$) is a~tuple $\aut_{Eq}
= (Q, \regs, \Deltaeqreg, I, F)$ where
$Q, \regs, I, F$ are the same as for $\rsa$s and the transition relation $\Deltaeqreg$ for $\rsaeqreg$ is defined as
$\Deltaeqreg \subseteq Q \times \Sigma \times 2^\regs \times 2^\regs \times (\regs \to 2^{\regs}) \times (\regs \to 2^{\regs \cup \set{\inp}}) \times Q$.
The semantics of a~transition $\trans q a {\guardin, \guardnotin, \guardeq, \update}{s}$
is such that~$\aut_E$ can move from state~$q$ to state~$s$ if the $\Sigma$-symbol at
the current position of the input word is~$a$ and the $\datadom$-symbol at the
current position is in all registers from~$\guardin$ and in no register
from~$\guardnotin$ and for all $ r_i \in \guardeq(r)$ it holds that $r_i = r$. The content of the
registers is updated so that $r_i \gets \bigcup \set{x \mid x \in \update(r_i)}$.

\begin{theorem}\label{thm:rsa-eq-emptiness}
 The emptiness problem for $\rsaeqreg$ is undecidable.
\end{theorem}

\begin{proof}
The proof is done by reduction from coverability in \emph{Petri nets} with
\emph{inhibitor arcs}, which is an undecidable problem~\cite{reinhardt2008reachability}.
Given a~$\text{PN}_I$ $\net_I$ with inhibitor arcs, we construct a~corresponding
$\rsaeqreg$ $\aut^{=r}_I$.
The process of construction of $\aut^{=r}_I$ follows the
reduction of TPN to $\rsa$ in the proof of
~\cref{lem:rsa-emptiness-fomega-hard} of \cref{sec:proof-rsa-emptiness}.
The structure of the resulting automaton differs in
\emph{protogadgets} whose concatenation is used for construction of \emph{gadgets}
which make up the reduced $\rsaeqreg$.

The following \emph{protogadgets} are used in the reduction:
\begin{enumerate}
    \item The \emph{EmptyEq} protogadget (depicted in~\cref{fig:protogadget_emptyeq}),
    which simulates the inhibitor arc leading from place $p$ is an $\rsaeqreg$ defined
    as:
    \begin{multline*}
      \emptyeqgadgetof {p} =
      (\{q_1, q_2, q_3\}, \regs, \\
      \{\trans{q_1}{a}{\emptyset, \emptyset, \emptyset, \{r_e \mapsto \emptyset\}}{q_2},
      \trans{q_2}{a}{\emptyset, \emptyset, \{r_p \mapsto \{r_e\}\}, \emptyset}{q_3}\},
      \{q_1\}, \{q_3\}).
    \end{multline*}
    Intuitively, the $\emptyeqgadget$ simulates a register
    emptiness test. At first, the protogadget explicitly
    assigns the value of empty set to the register $r_e$
    and then it compares its equality with the content of the register
    representing the place in which the inhibitor arc
    originates.

    \item The \emph{New Token} (depicted in \cref{fig:protogadget_new_token_eq})
    which simulates adding a~token to a~place \emph{p}, is an $\rsaremempty$ defined in the following way:
    $$
      \newtokengadgetof {p} =
      (\{q_1, q_2\}, \regs,
      \{\trans{q_1}{a}{\emptyset, \regs, \emptyset, \{\regof{p} \mapsto \{r_p, \inp\}\}\}}{q_2}\}, \{q_1\}, \{q_2\}).
    $$
    Intuitively, for each arc originating in the transition and ending in a~particular place,
    a token is added to the register representing the destination.
    The guard ensures that the added value is not already present within the register,
    so that the number of values actually increases.

    \item The \emph{Non-lossy Remove Token} (depicted in \cref{fig:protogadget_nonlossy_tokenrem})
    which simulates removal of a~token from a~place $p$ is an $\rsaeqreg$ defined in the following way:
    \begin{multline*}
      \nonlossyrmgadgetof {p} =
      (\{q_1, q_2, q_3\}, \regs,
      \{\trans{q_1}{a}{\{r_p\}, \emptyset, \emptyset, \{r_\inp \mapsto \{\inp\}, r_s \mapsto \{\inp\}\}}{q_2}, \\
      \trans{q_2}{a}{\{r_p\}, \{r_\inp\}, \emptyset, \{r_{\mathit{tmp}} \mapsto \{r_{\mathit{tmp}}, inp\},
      r_s \mapsto \{r_{\mathit{tmp}}, r_\inp, \inp\}\}}{q_2}, \\
      \trans{q_2}{a}{\emptyset, \emptyset, \{r_p \mapsto \{r_s\}\}, \{r_p \mapsto \{r_{\mathit{tmp}}\}\}}{q_3}
      \}, \{q_1\}, \{q_2\}).
    \end{multline*}
    Intuitively, for each arc originating in a~place $p$ and terminating in transition, respective number
    of values has to be removed from the register representing place $p$.
    Therefore on each \emph{protogadget} of this kind, one value is removed from a~register
    representing the source place in a~lossless manner. The quality of being lossless is necessary in order
    to sustain the semantics of source PN. Otherwise, some transitions may be enabled even though
    they were not in the source PN.
\end{enumerate}

\begin{figure}[t]
\centering
\begin{subfigure}[b]{0.4\linewidth}
\begin{center}
\scalebox{0.8}{
  \input{figs/protogadget_empty_eq.tikz}
}
\end{center}
\caption{The $\emptyeqgadgetof {p}$ protogadget.}
\label{fig:protogadget_emptyeq}
\end{subfigure}
\hfill
\begin{subfigure}[b]{0.4\linewidth}
\begin{center}
\scalebox{0.8}{
  \input{figs/protogadget_newtoken.tikz}
}
\end{center}
\caption{The $\newtokengadgetof {p}$ protogadget.}
\label{fig:protogadget_new_token_eq}
\end{subfigure}\\[4mm]
\begin{subfigure}[b]{0.5\linewidth}
\begin{center}
\scalebox{0.8}{
  \input{figs/protogadget_nonlossyrem.tikz}
}
\end{center}
\caption{The $\nonlossyrmgadgetof {p}$ protogadget.}
\label{fig:protogadget_nonlossy_tokenrem}
\end{subfigure}
\caption{Protogadgets used in the construction of the $\rsaeqreg$ for $\net_{I}$.}
\end{figure}

\end{proof}

\vspace{-0.0mm}
\subsection{RsAs with Removal}\label{sec:rsa_rem}
\vspace{-0.0mm}

An $\rsa$ with \emph{removal} ($\rsarem$) is a~tuple
$\aut_{R}
= (Q, \regs, \Deltarem, I, F)$,
where $Q, \regs, I, F$ are the same as for $\rsa$s and the transition relation
$\Deltarem$ is defined as
$\Deltarem \subseteq Q \times \Sigma \times 2^\regs \times 2^\regs \times 2^\regs \times (\regs \to 2^{\regs \cup \set{\inp}}) \times Q$.
The semantics of a~transition $\trans q a {\guardin, \guardnotin, \remove, \update}{s}$
is such that~$\aut_{R}$ can move from state $q$ to state $s$ if the
$\Sigma$-symbol at the current position of the input word is~$a$ and the
$\datadom$-symbol at the current position is in all registers from~$\guardin$
and in no register from~$\guardnotin$; with respect to the~$\remove$ and~$\update$, the content of the
registers is updated so that for all~$r_i \in \regs$, if $r_i \in \remove$, then
$r_i \gets \bigcup \set{x \mid x \in \update(r_i)} \setminus \{\inp\}$, else $r_i \gets \bigcup \set{x \mid x \in \update(r_i)}$.

\begin{theorem}\label{thm:}
The emptiness problem for $\rsarem$ is $\fomega$-complete.
\end{theorem}

\begin{proof}
$\fomega$-hardness follows trivially from \cref{thm:rsa-emptiness}.
To prove membership in~$\fomega$, we can reuse the proof of
  \cref{lem:rsa-emptiness-in-fomega}, with the only difference being in the
  definition of~$\outrans$ for a~transition $\trans q a {\guardin, \guardnotin,
  \remove, \update}{s}$ as follows:
  \begin{itemize}
  \item  $\outrans = \unitof{\dst} + \unitof s$ where
    $\dst = \set{r_i \in \regs \mid (\set{\inp} \cup \regionof g) \cap
      \update(r_i) \neq \emptyset \land r_i \notin \remove}$.
  \end{itemize}
  Everything else can stay the same.
\end{proof}

\vspace{-0.0mm}
\subsection{RsAs with Removal and Register Emptiness Test}\label{sec:rsa_rem_empty}
\vspace{-0.0mm}

An $\rsa$ with \emph{removal and register emptiness test} ($\rsaremempty$) is a~tuple
$\aut_{RE}
= (Q, \regs, \Deltaremempty, I, F)$,
where $Q, \regs, I, F$ are the same as for $\rsa$s and the transition relation
$\Deltaremempty$ is defined as
$\Deltaremempty \subseteq Q \times \Sigma \times 2^\regs \times 2^\regs \times 2^\regs \times 2^\regs \times (\regs \to 2^{\regs \cup \set{\inp}}) \times Q$.
The semantics of a~transition $\trans q a {\guardin, \guardnotin, \guardempty, \remove, \update}{s}$
is such that~$\aut_{RE}$ can move
from state $q$ to state $s$ if the
$\Sigma$-symbol at the current position of the
input word is~$a$
and the $\datadom$-symbol at the
current position is in all registers from~$\guardin$ and in no register
from~$\guardnotin$ and all registers from~$\guardempty$ are empty; with respect
to the $\remove$ and $\update$,
the content of the
registers is updated so that for all~$r_i \in \regs$, if $r_i \in \remove$, then
$r_i \gets \bigcup \set{x \mid x \in \update(r_i)} \setminus \{\inp\}$, else $r_i \gets \bigcup \set{x \mid x \in \update(r_i)}$.

\begin{theorem}
 The emptiness problem for $\rsaremempty$ is undecidable.
\end{theorem}

\begin{proof}
The proof is done by showing the reducibility from \emph{reachability}
in \emph{Petri nets} with \emph{inhibitor arcs}, which is an undecidable problem.

Given a~$\text{PN}_I$ $\net_I$ with inhibitor arc, we construct the $\rsaremempty$ $\autof {\net_I}$.
The structure of the $\rsaremempty$ $\autof {\net_I}$ is similar to the structure of $\rsa$ $\autof {\net}$
in the proof of the $\fomega$-hardness of emptiness in $\rsa$,
which can be seen in \cref{lem:rsa-emptiness-fomega-hard} of \cref{sec:proof-rsa-emptiness}.

The only difference is in \emph{gadgets} used for simulation of respective transition in $\text{PN}_\mathit{I}$
and gadget for doing the reachability test.
These are created by concatenation of respective \emph{protogadgets}, which are defined in the following way:

\begin{figure}[t]
\centering
\begin{subfigure}[b]{0.4\linewidth}
\begin{center}
\scalebox{1}{
  \input{figs/protogadget_inh_empty.tikz}
}
\end{center}
\caption{The $\emptygadgetof {p}$ protogadget.}
\label{fig:protogadget_empty}
\end{subfigure}
\hfill
\begin{subfigure}[b]{0.4\linewidth}
\begin{center}
\scalebox{1}{
  \input{figs/protogadget_newtoken.tikz}
}
\end{center}
\caption{The $\newtokengadgetof {p}$ protogadget.}
\label{fig:protogadget_new_token}
\end{subfigure}\\[6mm]
\begin{subfigure}[b]{0.4\linewidth}
\begin{center}
\scalebox{1}{
  \input{figs/protogadget_inh_token_rem.tikz}
}
\end{center}
\caption{The $\tokenremgadgetof {p}$ protogadget.}
\label{fig:protogadget_tokenrem}
\end{subfigure}
\caption{Protogadgets used in the construction of the $\rsaremempty$ for $\net_I$.}
\end{figure}

\begin{enumerate}
    \item The \emph{Empty} protogadget (depicted in \cref{fig:protogadget_empty}), which simulates the inhibitor
    arc leading from the place \emph{p} is an $\rsaremempty$ defined as:
    $$
      \emptygadgetof {p} =
      (\{q_1, q_2\}, \regs,
      \{\trans{q_1}{a}{\emptyset, \emptyset, \{r_p\},
      \emptyset, \{\regof{p} \mapsto \emptyset\}\}}{q_2}\}, \{q_1\}, \{q_2\}).
    $$
    Intuitively, since the inhibitor arc enables its transition only if the source place is empty,
    the respective \emph{protogadget} checks on its guard whether
    the register representing the source place is
    empty as well.

    \item The \emph{New Token} protogadget (depicted in \cref{fig:protogadget_new_token})
    which simulates adding a~token to a~place \emph{p}, is an $\rsaremempty$ defined in the following way:
    $$
      \newtokengadgetof {p} =
      (\{q_1, q_2\}, \regs,
      \{\trans{q_1}{a}{\emptyset, \regs, \emptyset, \emptyset, \{\regof{p} \mapsto \{r_p, \inp\}\}\}}{q_2}\}, \{q_1\}, \{q_2\}).
    $$
    Intuitively, for each arc originating in the transition and ending in a~particular place,
    a token is added to the register representing the destination.
    The guard ensures that the added value is not already present within the register,
    so that the number of values actually increases.

    \item The \emph{Remove Token} protogadget (depicted in \cref{fig:protogadget_tokenrem})
    which simulates the removal of a~token from a~place \emph{p}, is an $\rsaremempty$ defined
    in the following way:
    $$
      \tokenremgadgetof {p} =
      (\{q_1, q_2\}, \regs,
      \{\trans{q_1}{a}{\emptyset, \{r_p\}, \emptyset,
      \{r_p\}, \{\regof{p} \mapsto \{r_p\}\}\}}{q_2}\}, \{q_1\}, \{q_2\}).
    $$
    Intuitively, for each token removed from the source place, one value is removed
    from the register which represents that place.
    \qedhere
\end{enumerate}

\end{proof}

%% file: figs/protogadget_empty_eq.tikz
\begin{tikzpicture}
  \tikzset{
    ->, 
    >=stealth',
    initial text=$ $, 
    node distance=30mm,
  }

  \node[initial,state] (q1) {$q_1$};
  \node[state] (q2) [right of=q1] {$q_2$};
  \node[accepting,state] (q3) [right of=q2] {$q_3$};

 \draw (q1) edge[below] node {$\pictrans a {} {\regof e \gets \emptyset}$} (q2);
 \draw (q2) edge[below] node {$\pictrans a {\regof p = \regof e} {}$} (q3);
\end{tikzpicture}

%% file: figs/protogadget_nonlossyrem.tikz
\begin{tikzpicture}
  \tikzset{
    ->, 
    >=stealth',
    initial text=$ $, 
    node distance=30mm,
  }

  \node[initial,state] (q1) {$q_1$};
  \node[state] (q2) [right of=q1] {$q_2$};
  \node[accepting,state] (q3) [right of=q2] {$q_3$};

 \draw (q1) edge[below] node {$\pictrans a {\inp \in \regof p} {\regof \inp
  \gets \{\inp\} \\ \regof s \gets \{\inp\}}$} (q2);
 \draw (q2) edge[loop above] node {$\pictrans a {\inp \in \regof p \\ \inp
  \notin \regof \inp} {\regof \tmp \gets \regof \tmp \cup \{\inp\} \\
  \regof s \gets \regof \tmp \cup \regof \inp \cup \{\inp\}}$} (q2);
 \draw (q2) edge[below] node {$\pictrans a {} {\regof p \gets \regof \tmp \\
  \regof p = \regof s}$} (q3);
\end{tikzpicture}

%% file: figs/protogadget_inh_empty.tikz
\begin{tikzpicture}
  \tikzset{
    ->, 
    >=stealth',
    initial text=$ $, 
    node distance=30mm,
  }

  \node[initial,state] (q1) {$q_1$};
  \node[accepting,state] (q2) [right of=q1] {$q_2$};

 \draw (q1) edge[below] node {$\pictrans a {\regof{p} = \emptyset \\~} {}$} (q2);
\end{tikzpicture}

%% file: figs/protogadget_inh_token_rem.tikz
\begin{tikzpicture}
  \tikzset{
    ->, 
    >=stealth',
    initial text=$ $, 
    node distance=30mm,
  }

  \node[initial,state] (q1) {$q_1$};
  \node[accepting,state] (q2) [right of=q1] {$q_2$};

 \draw (q1) edge[below] node {$\pictrans a {\inp \in \regof p}
  {\regof p \gets \regof p \setminus \{\inp\}\\~}$} (q2);
\end{tikzpicture}

%% file: figs/DRsA_no_DHRA.tikz
\begin{tikzpicture}
  \tikzset{
    ->, 
    >=stealth',
    initial text=$ $, 
    node distance=30mm,
  }

  \node[initial,state] (q) {$q$};
  \node[state, accepting] (f) [right of=q] {$f$};


 \draw (q) edge[loop above] node {$\pictrans {a} {} {r_2 \gets r_2 \cup \{\inp\}}$} (q);
 \draw (q) to [out=330,in=300,looseness=8]  node[yshift=-7mm, xshift=3mm] {$\pictrans {co} {} {r_1 \gets r_1 \cup r_2\\ r_2 \gets \emptyset}$} (q);
 \draw (q) to [out=240,in=210,looseness=8] node[yshift=-5mm, xshift=-3mm] {$\pictrans {rb} {} {r_2 \gets \emptyset}$} (q);

 \draw (q) edge node[above,yshift=-3mm] {$\pictrans \# {\inp \in r_1} {}$} (f);



\end{tikzpicture}

%% file: literature.bib
@inproceedings{SchmitzS13,
  author    = {Sylvain Schmitz and
               Philippe Schnoebelen},
  editor    = {Pedro R. D'Argenio and
               Hern{\'{a}}n C. Melgratti},
  title     = {The Power of Well-Structured Systems},
  booktitle = {{CONCUR} 2013, Buenos Aires, Argentina, August 27-30, 2013. Proceedings},
  series    = {Lecture Notes in Computer Science},
  volume    = {8052},
  pages     = {5--24},
  publisher = {Springer},
  year      = {2013},
  url       = {https://doi.org/10.1007/978-3-642-40184-8\_2},
  doi       = {10.1007/978-3-642-40184-8\_2},
  bibsource = {dblp computer science bibliography, https://dblp.org}
}

@phdthesis{Schmitz17,
  TITLE = {{Algorithmic Complexity of Well-Quasi-Orders}},
  AUTHOR = {Schmitz, Sylvain},
  URL = {https://tel.archives-ouvertes.fr/tel-01663266},
  SCHOOL = {{{\'E}cole normale sup{\'e}rieure Paris-Saclay}},
  YEAR = {2017},
  MONTH = Nov,
  TYPE = {Habilitation {\`a} diriger des recherches},
  HAL_ID = {tel-01663266},
  HAL_VERSION = {v1},
}

@unpublished{SchmitzS12,
  TITLE = {{Algorithmic Aspects of WQO Theory}},
  AUTHOR = {Schmitz, Sylvain and Schnoebelen, Philippe},
  URL = {https://cel.archives-ouvertes.fr/cel-00727025},
  TYPE = {DEA},
  YEAR = {2012},
  MONTH = Aug,
  HAL_ID = {cel-00727025},
  HAL_VERSION = {v1},
}

@inproceedings{EsparzaFM99,
  author    = {Javier Esparza and
               Alain Finkel and
               Richard Mayr},
  title     = {On the Verification of Broadcast Protocols},
  booktitle = {14th Annual {IEEE} Symposium on Logic in Computer Science, Trento,
               Italy, July 2-5, 1999},
  pages     = {352--359},
  publisher = {{IEEE} Computer Society},
  year      = {1999},
  url       = {https://doi.org/10.1109/LICS.1999.782630},
  doi       = {10.1109/LICS.1999.782630},
  bibsource = {dblp computer science bibliography, https://dblp.org}
}

@article{DemriL09,
	title = {{LTL} with the freeze quantifier and register automata},
	volume = {10},
	issn = {1529-3785, 1557-945X},
	url = {https://dl.acm.org/doi/10.1145/1507244.1507246},
	doi = {10.1145/1507244.1507246},
	language = {en},
	number = {3},
	urldate = {2021-09-05},
	journal = {ACM Transactions on Computational Logic},
	author = {Demri, Stéphane and Lazić, Ranko},
	month = apr,
	year = {2009},
	pages = {1--30},
}

@article{KaminskiF94,
	title = {Finite-memory automata},
	volume = {134},
	issn = {0304-3975},
	url = {https://www.sciencedirect.com/science/article/pii/0304397594902429},
	doi = {10.1016/0304-3975(94)90242-9},
	language = {en},
	number = {2},
	urldate = {2021-09-05},
	journal = {Theoretical Computer Science},
	author = {Kaminski, Michael and Francez, Nissim},
	month = nov,
	year = {1994},
	pages = {329--363},
}

@article{SakamotoI00,
  author    = {Hiroshi Sakamoto and
               Daisuke Ikeda},
  title     = {Intractability of decision problems for finite-memory automata},
  journal   = {Theor. Comput. Sci.},
  volume    = {231},
  number    = {2},
  pages     = {297--308},
  year      = {2000},
  url       = {https://doi.org/10.1016/S0304-3975(99)00105-X},
  doi       = {10.1016/S0304-3975(99)00105-X},
  bibsource = {dblp computer science bibliography, https://dblp.org}
}

@article{Figueira12,
	title = {Alternating register automata on finite words and trees},
	volume = {8},
	issn = {18605974},
	url = {http://arxiv.org/abs/1202.3957},
	doi = {10.2168/LMCS-8(1:22)2012},
	number = {1},
	urldate = {2021-09-05},
	journal = {Logical Methods in Computer Science},
	author = {Figueira, Diego},
	month = mar,
	year = {2012},
	pages = {22},
}

@incollection{Segoufin06,
	address = {Berlin, Heidelberg},
	title = {Automata and {Logics} for {Words} and {Trees} over an {Infinite} {Alphabet}},
	volume = {4207},
	isbn = {978-3-540-45458-8 978-3-540-45459-5},
	url = {http://link.springer.com/10.1007/11874683_3},
	language = {en},
	urldate = {2021-09-06},
	booktitle = {Computer {Science} {Logic}},
	publisher = {Springer Berlin Heidelberg},
	author = {Segoufin, Luc},
	year = {2006},
	doi = {10.1007/11874683_3},
	pages = {41--57},
}

@article{Schmitz16a,
	title = {Complexity {Hierarchies} {Beyond} {Elementary}},
	volume = {8},
	issn = {1942-3454, 1942-3462},
	url = {http://arxiv.org/abs/1312.5686},
	doi = {10.1145/2858784},
	number = {1},
	urldate = {2021-09-10},
	journal = {ACM Transactions on Computation Theory},
	author = {Schmitz, Sylvain},
	month = feb,
	year = {2016},
	pages = {1--36},
}

@article{NevenSV04,
  author    = {Frank Neven and
               Thomas Schwentick and
               Victor Vianu},
  title     = {Finite state machines for strings over infinite alphabets},
  journal   = {{ACM} Trans. Comput. Log.},
  volume    = {5},
  number    = {3},
  pages     = {403--435},
  year      = {2004},
  url       = {https://doi.org/10.1145/1013560.1013562},
  doi       = {10.1145/1013560.1013562},
  bibsource = {dblp computer science bibliography, https://dblp.org}
}

@article{TuronovaHLSVV20,
	title = {Regex matching with counting-set automata},
	volume = {4},
	copyright = {All rights reserved},
	url = {https://doi.org/10.1145/3428286},
	doi = {10.1145/3428286},
	number = {OOPSLA},
	urldate = {2021-09-06},
	journal = {Proceedings of the ACM on Programming Languages},
	author = {Turoňová, Lenka and Holík, Lukáš and Lengál, Ondřej and Saarikivi, Olli and Veanes, Margus and Vojnar, Tomáš},
	month = nov,
	year = {2020},
	pages = {218:1--218:30},
}

@inproceedings{HolikLSTVV19,
	address = {Cham},
	series = {Lecture {Notes} in {Computer} {Science}},
	title = {Succinct {Determinisation} of {Counting} {Automata} via {Sphere} {Construction}},
	copyright = {All rights reserved},
	isbn = {978-3-030-34175-6},
	doi = {10.1007/978-3-030-34175-6_24},
	language = {en},
	booktitle = {Programming {Languages} and {Systems}},
	publisher = {Springer International Publishing},
	author = {Holík, Lukáš and Lengál, Ondřej and Saarikivi, Olli and Turoňová, Lenka and Veanes, Margus and Vojnar, Tomáš},
	editor = {Lin, Anthony Widjaja},
	year = {2019},
	pages = {468--489},
}

@article{Minsky61,
	title = {Recursive {Unsolvability} of {Post}'s {Problem} of "{Tag}" and other {Topics} in {Theory} of {Turing} {Machines}},
	volume = {74},
	issn = {0003-486X},
	url = {https://www.jstor.org/stable/1970290},
	doi = {10.2307/1970290},
	number = {3},
	urldate = {2021-12-06},
	journal = {Annals of Mathematics},
	author = {Minsky, Marvin L.},
	year = {1961},
	pages = {437--455},
}

@inproceedings{ChenHLT20,
	address = {Cham},
	title = {A {Symbolic} {Algorithm} for the {Case}-{Split} {Rule} in {String} {Constraint} {Solving}},
	volume = {12470},
	copyright = {All rights reserved},
	isbn = {978-3-030-64436-9 978-3-030-64437-6},
	url = {https://link.springer.com/10.1007/978-3-030-64437-6_18},
	language = {en},
	urldate = {2021-09-06},
	booktitle = {Programming {Languages} and {Systems}},
	publisher = {Springer International Publishing},
	author = {Chen, Yu-Fang and Havlena, Vojtěch and Lengál, Ondřej and Turrini, Andrea},
	year = {2020},
	doi = {10.1007/978-3-030-64437-6_18},
	note = {Series Title: Lecture Notes in Computer Science},
	pages = {343--363},
}

@inproceedings{TovaDV00,
	address = {Dallas, USA},
	title = {Typechecking for {XML} transformers},
	url = {https://dl.acm.org/doi/abs/10.1145/335168.335171},
	doi = {10.1145/335168.335171},
	language = {EN},
	urldate = {2021-12-19},
	booktitle = {{POD}'00},
	publisher = {ACM},
	author = {Tova, Milo and Dan, Suciu and Victor, Vianu},
	year = {2000},
	pages = {11--22},
}

@inproceedings{BolligHLM13,
	address = {Berlin, Heidelberg},
	series = {Lecture {Notes} in {Computer} {Science}},
	title = {A {Fresh} {Approach} to {Learning} {Register} {Automata}},
	isbn = {978-3-642-38771-5},
	doi = {10.1007/978-3-642-38771-5_12},
	language = {en},
	booktitle = {Developments in {Language} {Theory}},
	publisher = {Springer},
	author = {Bollig, Benedikt and Habermehl, Peter and Leucker, Martin and Monmege, Benjamin},
	year = {2013},
	pages = {118--130},
}

@article{Schmid16,
	title = {Characterising {REGEX} {Languages} by {Regular} {Languages} {Equipped} with {Factor}-{Referencing}},
	volume = {249},
	url = {https://www.sciencedirect.com/science/article/pii/S0890540116000109},
	doi = {10.1016/j.ic.2016.02.003},
	language = {en},
	author = {Schmid, Markus L},
	month = aug,
	year = {2016},
	pages = {1--17},
}

@inproceedings{IosifX19,
	address = {Cham},
	series = {Lecture {Notes} in {Computer} {Science}},
	title = {Alternating {Automata} {Modulo} {First} {Order} {Theories}},
	volume = {11562},
	isbn = {978-3-030-25542-8 978-3-030-25543-5},
	url = {http://link.springer.com/10.1007/978-3-030-25543-5_3},
	language = {en},
	urldate = {2021-09-06},
	booktitle = {Computer {Aided} {Verification}},
	publisher = {Springer International Publishing},
	author = {Iosif, Radu and Xu, Xiao},
	year = {2019},
	doi = {10.1007/978-3-030-25543-5_3},
	pages = {43--63},
}

@misc{OwaspReDoS,
  author = "{Adar Weidman}",
  title = "Regular expression Denial of Service --- {ReDoS}",
  year = "2025",
  howpublished = "\url{https://owasp.org/www-community/attacks/Regular_expression_Denial_of_Service_-_ReDoS}",
  note = "[Online; accessed 25-March-2025]"
}

@misc{stackoutage,
author = {Stack Exchange}, 
year = {2016}, 
title = {Outage Postmortem}, 
 howpublished ={\url{http://stackstatus.net/post/147710624694/outage-postmortem-july-20-2016}},
urldate = {2019-09-21},
}

@misc{expressjsoutage,
author = {Adam Baldwin}, 
year = {2016}, 
title = {Regular Expression Denial of Service affecting {Express.js}}, 
howpublished = {\url{https://medium.com/node-security/regular- expression-denial-of-service-affecting- express-js-9c397c164c43}},
  note = "[Online; accessed 25-March-2025]"
}

@misc{regex101,
author = {Regex101.com}, 
year = {2025},
title = {}, 
howpublished = {\url{https://regex101.com}},
note = "[Online; accessed 25-March-2025]"
}

@article{IsbernerHS2014,
	title = {Learning register automata: from languages to program structures},
	volume = {96},
	issn = {1573-0565},
	shorttitle = {Learning register automata},
	url = {https://doi.org/10.1007/s10994-013-5419-7},
	doi = {10.1007/s10994-013-5419-7},
	language = {en},
	number = {1},
	urldate = {2021-12-22},
	journal = {Machine Learning},
	author = {Isberner, Malte and Howar, Falk and Steffen, Bernhard},
	month = jul,
	year = {2014},
	pages = {65--98},
}

@inproceedings{GarhewalVHSLS2020,
	address = {Cham},
	series = {Lecture {Notes} in {Computer} {Science}},
	title = {Grey-{Box} {Learning} of {Register} {Automata}},
	isbn = {978-3-030-63461-2},
	doi = {10.1007/978-3-030-63461-2_2},
	language = {en},
	booktitle = {Integrated {Formal} {Methods}},
	publisher = {Springer International Publishing},
	author = {Garhewal, Bharat and Vaandrager, Frits and Howar, Falk and Schrijvers, Timo and Lenaerts, Toon and Smits, Rob},
	editor = {Dongol, Brijesh and Troubitsyna, Elena},
	year = {2020},
	pages = {22--40},
}

@article{MurawskiRT2018,
	title = {Polynomial-{Time} {Equivalence} {Testing} for {Deterministic} {Fresh}-{Register} {Automata}},
	copyright = {Creative Commons Attribution 3.0 Unported license (CC-BY 3.0)},
	url = {http://drops.dagstuhl.de/opus/volltexte/2018/9654/},
	doi = {10.4230/LIPICS.MFCS.2018.72},
	language = {en},
	urldate = {2021-09-06},
	author = {Murawski, Andrzej S. and Ramsay, Steven J. and Tzevelekos, Nikos},
	collaborator = {Wagner, Michael},
	year = {2018},
	pages = {14 pages},
}

@online{Snort,
 author = {M. Roesch and others},
 year = {2022},
 title = {Snort: A Network Intrusion Detection and Prevention System,},
 address = {Cisco},
 url = {http://www.snort.org},
 urldate = {2022-01-06},
}

@inproceedings{BecchiC2008,
	address = {Madrid, Spain},
	title = {Extending finite automata to efficiently match {Perl}-compatible regular expressions},
	isbn = {978-1-60558-210-8},
	url = {http://portal.acm.org/citation.cfm?doid=1544012.1544037},
	doi = {10.1145/1544012.1544037},
	language = {en},
	urldate = {2021-09-05},
	booktitle = {Proceedings of the 2008 {ACM} {CoNEXT} {Conference} on - {CONEXT} '08},
	publisher = {ACM Press},
	author = {Becchi, Michela and Crowley, Patrick},
	year = {2008},
	pages = {1--12},
}

@article{Thompson1968,
	title = {Programming {Techniques}: {Regular} expression search algorithm},
	volume = {11},
	issn = {0001-0782},
	shorttitle = {Programming {Techniques}},
	url = {https://doi.org/10.1145/363347.363387},
	doi = {10.1145/363347.363387},
	number = {6},
	urldate = {2022-01-14},
	journal = {Communications of the ACM},
	author = {Thompson, Ken},
	year = {1968},
	pages = {419--422},
}

@online{re2,
 author = {Google},
 year = {2025},
 title = {{RE2}},
 url = {https://github.com/google/re2},
  urldate = {2025-03-25},
}

@inproceedings{WangHCPLHZ19,
  author    = {Xiang Wang and
               Yang Hong and
               Harry Chang and
               KyoungSoo Park and
               Geoff Langdale and
               Jiayu Hu and
               Heqing Zhu},
  editor    = {Jay R. Lorch and
               Minlan Yu},
  title     = {Hyperscan: {A} Fast Multi-pattern Regex Matcher for Modern {CPUs}},
  booktitle = {NSDI'19},
  pages     = {631--648},
  publisher = {{USENIX} Association},
  year      = {2019},
  url       = {https://www.usenix.org/conference/nsdi19/presentation/wang-xiang},
  bibsource = {dblp computer science bibliography, https://dblp.org}
}

@online{hyperscan-unsupported-backref,
 author = {Intel},
 year = {2022},
 title = {HyperScan manual, unsupported features},
 url = {https://intel.github.io/hyperscan/dev-reference/compilation.html#unsupported-constructs},
  urldate = {2022-01-06},
}

@online{re2-add-backref,
 author = {Google},
 year = {2022},
 title = {{RE2} Issue Tracker: Feature request \#101: Add backreference support},
 url = {https://github.com/google/re2/issues/101},
  urldate = {2022-01-06},
}

@misc{Nvidia21,
  author = {Nvidia Mellanox team},
  title = {Personal communication},
  year = {2021}
}

@inproceedings{NamjoshiN10,
  author    = {Kedar S. Namjoshi and
               Girija J. Narlikar},
  title     = {Robust and Fast Pattern Matching for Intrusion Detection},
  booktitle = {INFOCOM'10},
  pages     = {740--748},
  publisher = {{IEEE}},
  year      = {2010},
  url       = {https://doi.org/10.1109/INFCOM.2010.5462149},
  doi       = {10.1109/INFCOM.2010.5462149},
  bibsource = {dblp computer science bibliography, https://dblp.org}
}

@article{reinhardt2008reachability,
  title={Reachability in Petri nets with inhibitor arcs},
  author={Reinhardt, Klaus},
  journal={Electronic Notes in Theoretical Computer Science},
  volume={223},
  pages={239--264},
  year={2008},
  publisher={Elsevier}
}

@article{RabinS59,
  author    = {Michael O. Rabin and
               Dana S. Scott},
  title     = {Finite Automata and Their Decision Problems},
  journal   = {{IBM} J. Res. Dev.},
  volume    = {3},
  number    = {2},
  pages     = {114--125},
  year      = {1959},
  url       = {https://doi.org/10.1147/rd.32.0114},
  doi       = {10.1147/rd.32.0114},
  bibsource = {dblp computer science bibliography, https://dblp.org}
}

@inproceedings{DAntoniV14,
  author    = {Loris D'Antoni and
               Margus Veanes},
  editor    = {Suresh Jagannathan and
               Peter Sewell},
  title     = {Minimization of symbolic automata},
  booktitle = {The 41st Annual {ACM} {SIGPLAN-SIGACT} Symposium on Principles of
               Programming Languages, {POPL} '14, San Diego, CA, USA, January 20-21,
               2014},
  pages     = {541--554},
  publisher = {{ACM}},
  year      = {2014},
  url       = {https://doi.org/10.1145/2535838.2535849},
  doi       = {10.1145/2535838.2535849},
  bibsource = {dblp computer science bibliography, https://dblp.org}
}

@article{Tan13,
  author    = {Tony Tan},
  title     = {Graph Reachability and Pebble Automata over Infinite Alphabets},
  journal   = {{ACM} Trans. Comput. Log.},
  volume    = {14},
  number    = {3},
  pages     = {19:1--19:31},
  year      = {2013},
  url       = {https://doi.org/10.1145/2499937.2499940},
  doi       = {10.1145/2499937.2499940},
  bibsource = {dblp computer science bibliography, https://dblp.org}
}

@InProceedings{tzevelekos-hra,
	author="Tzevelekos, Nikos
	and Grigore, Radu",
	title="History-Register Automata",
	booktitle="Foundations of Software Science and Computation Structures",
	year="2013",
	publisher="Springer Berlin Heidelberg",
	address="Berlin, Heidelberg",
	pages="17--33",
	isbn="978-3-642-37075-5"
}

@InProceedings{benerjee-safa,
	author="Banerjee, Ansuman
	and Chatterjee, Kingshuk
	and Guha, Shibashis",
	title="Set Augmented Finite Automata over Infinite Alphabets",
	booktitle="Developments in Language Theory",
	year="2023",
	publisher="Springer Nature Switzerland",
	address="Cham",
	pages="36--50",
	isbn="978-3-031-33264-7"
}

@misc{grep,
  author = {M. Haertel et al.},
  title = {GNU grep},
  howpublished = {Online},
  version = {3.6},
  month = {September},
  year = {2022},
  note = {Available: \url{https://www.gnu.org/software/grep/}, [Accessed: 2024-11-17]}
}

@misc{python_re,
  author = {F. Lundh and A. M. Kuchling},
  title = {Python Standard Library: re module},
  howpublished = {Online},
  version = {3.10.12},
  month = {June},
  year = {2023},
  note = {Available: \url{https://docs.python.org/3/library/re.html}, [Accessed: 2024-11-17]}
}

@misc{pcre2,
  author = {P. Hazel},
  title = {Perl-compatible Regular Expressions},
  howpublished = {Online},
  version = {10.43},
  month = {February},
  year = {2024},
  note = {Available: \url{https://www.pcre.org/}, [Accessed: 2024-11-17]}
}

@misc{js_v8_regexp,
  author = {V8 JavaScript Engine},
  title = {V8 regular expression source code},
  howpublished = {Online},
  year = {2024},
  note = {Available: \url{https://github.com/v8/v8/tree/main/src/regexp}, [Accessed: 2024-11-17]}
}

@misc{java,
  author = {OpenJDK},
  title = {JDK 17 API Documentation: Regular Expressions},
  howpublished = {Online},
  year = {2024},
  note = {Available: \url{https://docs.oracle.com/en/java/javase/17/docs/api/java.base/java/util/regex/package-summary.html}, [Accessed: 2024-11-17]}
}

@misc{dotnet,
  author = {Microsoft},
  title = {.NET 8.0.110 regular expressions in System.Text.RegularExpressions},
  howpublished = {Online},
  year = {2024},
  note = {Available: \url{https://learn.microsoft.com/en-us/dotnet/api/system.text.regularexpressions?view=net-8.0}, [Accessed: 2024-11-17]}
}

@INPROCEEDINGS{rengar,
  author={Wang, Xinyi and Zhang, Cen and Li, Yeting and Xu, Zhiwu and Huang, Shuailin and Liu, Yi and Yao, Yican and Xiao, Yang and Zou, Yanyan and Liu, Yang and Huo, Wei},
  booktitle={2023 IEEE Symposium on Security and Privacy (SP)}, 
  title={Effective {ReDoS} Detection by Principled Vulnerability Modeling and Exploit Generation}, 
  year={2023},
  volume={},
  number={},
  pages={2427-2443},
  doi={10.1109/SP46215.2023.10179328}
}

@inproceedings{regexLinguaFranca,
  author       = {James C. Davis and
                  Louis G. Michael IV and
                  Christy A. Coghlan and
                  Francisco Servant and
                  Dongyoon Lee},
  editor       = {Marlon Dumas and
                  Dietmar Pfahl and
                  Sven Apel and
                  Alessandra Russo},
  title        = {Why aren't regular expressions a lingua franca? an empirical study
                  on the re-use and portability of regular expressions},
  booktitle    = {Proceedings of the {ACM} Joint Meeting on European Software Engineering
                  Conference and Symposium on the Foundations of Software Engineering,
                  {ESEC/SIGSOFT} {FSE} 2019, Tallinn, Estonia, August 26-30, 2019},
  pages        = {443--454},
  publisher    = {{ACM}},
  year         = {2019},
  url          = {https://doi.org/10.1145/3338906.3338909},
  doi          = {10.1145/3338906.3338909},
  bibsource    = {dblp computer science bibliography, https://dblp.org}
}

@inproceedings{Davis19,
  author       = {James C. Davis},
  editor       = {Marlon Dumas and
                  Dietmar Pfahl and
                  Sven Apel and
                  Alessandra Russo},
  title        = {Rethinking Regex engines to address {ReDoS}},
  booktitle    = {Proceedings of the {ACM} Joint Meeting on European Software Engineering
                  Conference and Symposium on the Foundations of Software Engineering,
                  {ESEC/SIGSOFT} {FSE} 2019, Tallinn, Estonia, August 26-30, 2019},
  pages        = {1256--1258},
  publisher    = {{ACM}},
  year         = {2019},
  url          = {https://doi.org/10.1145/3338906.3342509},
  doi          = {10.1145/3338906.3342509},
  bibsource    = {dblp computer science bibliography, https://dblp.org}
}

@inproceedings{TuronovaHHLVV22,
  author       = {Lenka Turo\v{n}ov{\'{a}} and
                  Luk{\'{a}}\v{s} Hol{\'{\i}}k and
                  Ivan Homoliak and
                  Ond\v{r}ej Leng{\'{a}}l and
                  Margus Veanes and
                  Tom{\'{a}}\v{s} Vojnar},
  editor       = {Kevin R. B. Butler and
                  Kurt Thomas},
  title        = {Counting in Regexes Considered Harmful: Exposing {ReDoS} Vulnerability
                  of Nonbacktracking Matchers},
  booktitle    = {31st {USENIX} Security Symposium, {USENIX} Security 2022, Boston,
                  MA, USA, August 10-12, 2022},
  pages        = {4165--4182},
  publisher    = {{USENIX} Association},
  year         = {2022},
  url          = {https://www.usenix.org/conference/usenixsecurity22/presentation/turonova},
  bibsource    = {dblp computer science bibliography, https://dblp.org}
}

@article{FREYDENBERGER20191,
title = {Deterministic regular expressions with back-references},
journal = {Journal of Computer and System Sciences},
volume = {105},
pages = {1-39},
year = {2019},
issn = {0022-0000},
doi = {https://doi.org/10.1016/j.jcss.2019.04.001},
url = {https://www.sciencedirect.com/science/article/pii/S0022000018301818},
author = {Dominik D. Freydenberger and Markus L. Schmid},
}

@article{Antimirov96,
  author       = {Valentin M. Antimirov},
  title        = {Partial Derivatives of Regular Expressions and Finite Automaton Constructions},
  journal      = {Theor. Comput. Sci.},
  volume       = {155},
  number       = {2},
  pages        = {291--319},
  year         = {1996},
  url          = {https://doi.org/10.1016/0304-3975(95)00182-4},
  doi          = {10.1016/0304-3975(95)00182-4},
  bibsource    = {dblp computer science bibliography, https://dblp.org}
}

@article{BarriereP24,
  author       = {Aur{\`{e}}le Barri{\`{e}}re and
                  Cl{\'{e}}ment Pit{-}Claudel},
  title        = {Linear Matching of {JavaScript} Regular Expressions},
  journal      = {Proc. {ACM} Program. Lang.},
  volume       = {8},
  number       = {{PLDI}},
  pages        = {1336--1360},
  year         = {2024},
  url          = {https://doi.org/10.1145/3656431},
  doi          = {10.1145/3656431},
  bibsource    = {dblp computer science bibliography, https://dblp.org}
}

@incollection{Aho90,
  author       = {Alfred V. Aho},
  editor       = {Jan van Leeuwen},
  title        = {Algorithms for Finding Patterns in Strings},
  booktitle    = {Handbook of Theoretical Computer Science, Volume {A:} Algorithms and
                  Complexity},
  pages        = {255--300},
  publisher    = {Elsevier and {MIT} Press},
  year         = {1990},
  bibsource    = {dblp computer science bibliography, https://dblp.org}
}

@inproceedings{Terauchi25,
  author       = {Tachio Terauchi},
  title        = {On {DoS} Vulnerability of Regular Expressions, with and Without Backreferences},
  booktitle    = {38th {IEEE} Computer Security Foundations Symposium, {CSF} 2025, Santa
                  Cruz, CA, USA, June 16-20, 2025},
  pages        = {190--204},
  publisher    = {{IEEE}},
  year         = {2025},
  url          = {https://doi.org/10.1109/CSF64896.2025.00011},
  doi          = {10.1109/CSF64896.2025.00011},
  bibsource    = {dblp computer science bibliography, https://dblp.org}
}

@article{ChidaT23,
  author       = {Nariyoshi Chida and
                  Tachio Terauchi},
  title        = {On Lookaheads in Regular Expressions with Backreferences},
  journal      = {{IEICE} Trans. Inf. Syst.},
  volume       = {106},
  number       = {5},
  pages        = {959--975},
  year         = {2023},
  url          = {https://doi.org/10.1587/transinf.2022edp7098},
  doi          = {10.1587/TRANSINF.2022EDP7098},
  bibsource    = {dblp computer science bibliography, https://dblp.org}
}

@article{VarataluVE25,
  author       = {Ian Erik Varatalu and
                  Margus Veanes and
                  Juhan P. Ernits},
  title        = {{RE{\#}}: High Performance Derivative-Based Regex Matching with Intersection,
                  Complement, and Restricted Lookarounds},
  journal      = {Proc. {ACM} Program. Lang.},
  volume       = {9},
  number       = {{POPL}},
  pages        = {1--32},
  year         = {2025},
  url          = {https://doi.org/10.1145/3704837},
  doi          = {10.1145/3704837},
  bibsource    = {dblp computer science bibliography, https://dblp.org}
}

@inproceedings{CeskaHHKLMMSV19,
  author       = {Milan Ceska and
                  Vojtech Havlena and
                  Luk{\'{a}}s Hol{\'{\i}}k and
                  Jan Korenek and
                  Ondrej Leng{\'{a}}l and
                  Denis Matousek and
                  Jir{\'{\i}} Matousek and
                  Jakub Semric and
                  Tom{\'{a}}s Vojnar},
  title        = {Deep Packet Inspection in {FPGAs} via Approximate Nondeterministic Automata},
  booktitle    = {27th {IEEE} Annual International Symposium on Field-Programmable Custom
                  Computing Machines, {FCCM} 2019, San Diego, CA, USA, April 28 - May
                  1, 2019},
  pages        = {109--117},
  publisher    = {{IEEE}},
  year         = {2019},
  url          = {https://doi.org/10.1109/FCCM.2019.00025},
  doi          = {10.1109/FCCM.2019.00025},
  bibsource    = {dblp computer science bibliography, https://dblp.org}
}

@inproceedings{MatousekMK18,
  author       = {Denis Matousek and
                  Jir{\'{\i}} Matousek and
                  Jan Korenek},
  title        = {High-Speed Regular Expression Matching with Pipelined Memory-Based
                  Automata},
  booktitle    = {26th {IEEE} Annual International Symposium on Field-Programmable Custom
                  Computing Machines, {FCCM} 2018, Boulder, CO, USA, April 29 - May
                  1, 2018},
  pages        = {214},
  publisher    = {{IEEE} Computer Society},
  year         = {2018},
  url          = {https://doi.org/10.1109/FCCM.2018.00048},
  doi          = {10.1109/FCCM.2018.00048},
  bibsource    = {dblp computer science bibliography, https://dblp.org}
}

@inproceedings{NogamiT25,
  author       = {Taisei Nogami and
                  Tachio Terauchi},
  editor       = {Pawel Gawrychowski and
                  Filip Mazowiecki and
                  Michal Skrzypczak},
  title        = {Efficient Matching of Some Fundamental Regular Expressions with Backreferences},
  booktitle    = {50th International Symposium on Mathematical Foundations of Computer
                  Science, {MFCS} 2025, August 25-29, 2025, Warsaw, Poland},
  series       = {LIPIcs},
  volume       = {345},
  pages        = {81:1--81:19},
  publisher    = {Schloss Dagstuhl - Leibniz-Zentrum f{\"{u}}r Informatik},
  year         = {2025},
  url          = {https://doi.org/10.4230/LIPIcs.MFCS.2025.81},
  doi          = {10.4230/LIPICS.MFCS.2025.81},
  bibsource    = {dblp computer science bibliography, https://dblp.org}
}

@inproceedings{BhuiyanCB0S25,
  author       = {Masudul Hasan Masud Bhuiyan and
                  Berk {\c{C}}akar and
                  Ethan H. Burmane and
                  James C. Davis and
                  Cristian{-}Alexandru Staicu},
  title        = {SoK: {A} Literature and Engineering Review of Regular Expression Denial
                  of Service (ReDoS)},
  booktitle    = {Proceedings of the 20th {ACM} Asia Conference on Computer and Communications
                  Security, {ASIA} {CCS} 2025, Hanoi, Vietnam, August 25-29, 2025},
  pages        = {1659--1675},
  publisher    = {{ACM}},
  year         = {2025},
  url          = {https://doi.org/10.1145/3708821.3733912},
  doi          = {10.1145/3708821.3733912},
  bibsource    = {dblp computer science bibliography, https://dblp.org}
}

@inproceedings{HassanALDS23,
  author       = {Sk Adnan Hassan and
                  Zainab Aamir and
                  Dongyoon Lee and
                  James C. Davis and
                  Francisco Servant},
  title        = {Improving Developers' Understanding of Regex Denial of Service Tools
                  through Anti-Patterns and Fix Strategies},
  booktitle    = {44th {IEEE} Symposium on Security and Privacy, {SP} 2023, San Francisco,
                  CA, USA, May 21-25, 2023},
  pages        = {1238--1255},
  publisher    = {{IEEE}},
  year         = {2023},
  url          = {https://doi.org/10.1109/SP46215.2023.10179442},
  doi          = {10.1109/SP46215.2023.10179442},
  bibsource    = {dblp computer science bibliography, https://dblp.org}
}

@inproceedings{BarlasDD22,
  author       = {Efe Barlas and
                  Xin Du and
                  James C. Davis},
  title        = {Exploiting Input Sanitization for Regex Denial of Service},
  booktitle    = {44th {IEEE/ACM} 44th International Conference on Software Engineering,
                  {ICSE} 2022, Pittsburgh, PA, USA, May 25-27, 2022},
  pages        = {883--895},
  publisher    = {{ACM}},
  year         = {2022},
  url          = {https://doi.org/10.1145/3510003.3510047},
  doi          = {10.1145/3510003.3510047},
  bibsource    = {dblp computer science bibliography, https://dblp.org}
}

@phdthesis{Davis20a,
  author       = {James C. Davis},
  title        = {On the Impact and Defeat of Regular Expression Denial of Service},
  school       = {Virginia Tech, Blacksburg, VA, {USA}},
  year         = {2020},
  url          = {http://hdl.handle.net/10919/98593},
  bibsource    = {dblp computer science bibliography, https://dblp.org}
}

@inproceedings{LiuZM21,
  author       = {Yinxi Liu and
                  Mingxue Zhang and
                  Wei Meng},
  title        = {Revealer: Detecting and Exploiting Regular Expression Denial-of-Service
                  Vulnerabilities},
  booktitle    = {42nd {IEEE} Symposium on Security and Privacy, {SP} 2021, San Francisco,
                  CA, USA, 24-27 May 2021},
  pages        = {1468--1484},
  publisher    = {{IEEE}},
  year         = {2021},
  url          = {https://doi.org/10.1109/SP40001.2021.00062},
  doi          = {10.1109/SP40001.2021.00062},
  bibsource    = {dblp computer science bibliography, https://dblp.org}
}

@inproceedings{SuHLCG24,
  author       = {Weihao Su and
                  Hong Huang and
                  Rongchen Li and
                  Haiming Chen and
                  Tingjian Ge},
  editor       = {Davide Balzarotti and
                  Wenyuan Xu},
  title        = {Towards an Effective Method of {ReDoS} Detection for Non-backtracking
                  Engines},
  booktitle    = {33rd {USENIX} Security Symposium, {USENIX} Security 2024, Philadelphia,
                  PA, USA, August 14-16, 2024},
  publisher    = {{USENIX} Association},
  year         = {2024},
  url          = {https://www.usenix.org/conference/usenixsecurity24/presentation/su-weihao},
  bibsource    = {dblp computer science bibliography, https://dblp.org}
}

@inproceedings{ParoliniM22,
  author       = {Francesco Parolini and
                  Antoine Min{\'{e}}},
  editor       = {Yamine A{\"{\i}}t Ameur and
                  Florin Craciun},
  title        = {Sound Static Analysis of Regular Expressions for Vulnerabilities to
                  Denial of Service Attacks},
  booktitle    = {Theoretical Aspects of Software Engineering - 16th International Symposium,
                  {TASE} 2022, Cluj-Napoca, Romania, July 8-10, 2022, Proceedings},
  series       = {Lecture Notes in Computer Science},
  volume       = {13299},
  pages        = {73--91},
  publisher    = {Springer},
  year         = {2022},
  url          = {https://doi.org/10.1007/978-3-031-10363-6\_6},
  doi          = {10.1007/978-3-031-10363-6\_6},
  bibsource    = {dblp computer science bibliography, https://dblp.org}
}

@inproceedings{KirrageRT13,
  author       = {James Kirrage and
                  Asiri Rathnayake and
                  Hayo Thielecke},
  editor       = {Javier L{\'{o}}pez and
                  Xinyi Huang and
                  Ravi S. Sandhu},
  title        = {Static Analysis for Regular Expression Denial-of-Service Attacks},
  booktitle    = {Network and System Security - 7th International Conference, {NSS}
                  2013, Madrid, Spain, June 3-4, 2013. Proceedings},
  series       = {Lecture Notes in Computer Science},
  volume       = {7873},
  pages        = {135--148},
  publisher    = {Springer},
  year         = {2013},
  url          = {https://doi.org/10.1007/978-3-642-38631-2\_11},
  doi          = {10.1007/978-3-642-38631-2\_11},
  bibsource    = {dblp computer science bibliography, https://dblp.org}
}

@inproceedings{WustholzOHD17,
  author       = {Valentin W{\"{u}}stholz and
                  Oswaldo Olivo and
                  Marijn J. H. Heule and
                  Isil Dillig},
  editor       = {Axel Legay and
                  Tiziana Margaria},
  title        = {Static Detection of {DoS} Vulnerabilities in Programs that Use Regular
                  Expressions},
  booktitle    = {Tools and Algorithms for the Construction and Analysis of Systems
                  - 23rd International Conference, {TACAS} 2017, Held as Part of the
                  European Joint Conferences on Theory and Practice of Software, {ETAPS}
                  2017, Uppsala, Sweden, April 22-29, 2017, Proceedings, Part {II}},
  series       = {Lecture Notes in Computer Science},
  volume       = {10206},
  pages        = {3--20},
  year         = {2017},
  url          = {https://doi.org/10.1007/978-3-662-54580-5\_1},
  doi          = {10.1007/978-3-662-54580-5\_1},
  bibsource    = {dblp computer science bibliography, https://dblp.org}
}

@inproceedings{WeidemanMBW16,
  author       = {Nicolaas Weideman and
                  Brink van der Merwe and
                  Martin Berglund and
                  Bruce W. Watson},
  editor       = {Yo{-}Sub Han and
                  Kai Salomaa},
  title        = {Analyzing Matching Time Behavior of Backtracking Regular Expression
                  Matchers by Using Ambiguity of {NFA}},
  booktitle    = {Implementation and Application of Automata - 21st International Conference,
                  {CIAA} 2016, Seoul, South Korea, July 19-22, 2016, Proceedings},
  series       = {Lecture Notes in Computer Science},
  volume       = {9705},
  pages        = {322--334},
  publisher    = {Springer},
  year         = {2016},
  url          = {https://doi.org/10.1007/978-3-319-40946-7\_27},
  doi          = {10.1007/978-3-319-40946-7\_27},
  bibsource    = {dblp computer science bibliography, https://dblp.org}
}

@inproceedings{Shen000ML18,
  author       = {Yuju Shen and
                  Yanyan Jiang and
                  Chang Xu and
                  Ping Yu and
                  Xiaoxing Ma and
                  Jian Lu},
  editor       = {Marianne Huchard and
                  Christian K{\"{a}}stner and
                  Gordon Fraser},
  title        = {{ReScue}: crafting regular expression {DoS} attacks},
  booktitle    = {Proceedings of the 33rd {ACM/IEEE} International Conference on Automated
                  Software Engineering, {ASE} 2018, Montpellier, France, September 3-7,
                  2018},
  pages        = {225--235},
  publisher    = {{ACM}},
  year         = {2018},
  url          = {https://doi.org/10.1145/3238147.3238159},
  doi          = {10.1145/3238147.3238159},
  bibsource    = {dblp computer science bibliography, https://dblp.org}
}

@inproceedings{McLaughlinPSKV22,
  author       = {Robert McLaughlin and
                  Fabio Pagani and
                  Noah Spahn and
                  Christopher Kruegel and
                  Giovanni Vigna},
  editor       = {Kevin R. B. Butler and
                  Kurt Thomas},
  title        = {Regulator: Dynamic Analysis to Detect {ReDoS}},
  booktitle    = {31st {USENIX} Security Symposium, {USENIX} Security 2022, Boston,
                  MA, USA, August 10-12, 2022},
  pages        = {4219--4235},
  publisher    = {{USENIX} Association},
  year         = {2022},
  url          = {https://www.usenix.org/conference/usenixsecurity22/presentation/mclaughlin},
  bibsource    = {dblp computer science bibliography, https://dblp.org}
}

@article{Freydenberger13,
  author       = {Dominik D. Freydenberger},
  title        = {Extended Regular Expressions: Succinctness and Decidability},
  journal      = {Theory Comput. Syst.},
  volume       = {53},
  number       = {2},
  pages        = {159--193},
  year         = {2013},
  url          = {https://doi.org/10.1007/s00224-012-9389-0},
  doi          = {10.1007/S00224-012-9389-0},
  bibsource    = {dblp computer science bibliography, https://dblp.org}
}

@book{EsparzaB23,
  author={Javier Esparza and
          Michael Blondin},
  title={Automata Theory: An Algorithmic Approach},
  isbn={9780262048637},
  publisher={MIT Press},
  year={2023}
}

@inproceedings{KongYCGHM022,
  author       = {Lingkun Kong and
                  Qixuan Yu and
                  Agnishom Chattopadhyay and
                  Alexis Le Glaunec and
                  Yi Huang and
                  Konstantinos Mamouras and
                  Kaiyuan Yang},
  editor       = {Ranjit Jhala and
                  Isil Dillig},
  title        = {Software-hardware codesign for efficient in-memory regular pattern
                  matching},
  booktitle    = {{PLDI} '22: 43rd {ACM} {SIGPLAN} International Conference on Programming
                  Language Design and Implementation, San Diego, CA, USA, June 13 -
                  17, 2022},
  pages        = {733--748},
  publisher    = {{ACM}},
  year         = {2022},
  url          = {https://doi.org/10.1145/3519939.3523456},
  doi          = {10.1145/3519939.3523456},
  bibsource    = {dblp computer science bibliography, https://dblp.org}
}

@article{techrep,
  author      = {Vojtěch Havlena and
                 Lukáš Holík and
                 Ondřej Lengál and
                 Jan Vašák and
                 Sabína Gulčíková},
  title        = {Towards Efficient Matching of Regexes with Backreferences
                  using Register Set Automata (Technical Report)},
  journal      = {CoRR},
  volume       = {2205.12114},
  year         = {2026},
  url          = {https://arxiv.org/abs/2205.12114},
  eprinttype    = {arXiv},
  eprint       = {2205.12114},
}

@software{artifact,
  author       = {Vojtěch Havlena and
                  Lukáš Holík and
                  Ondřej Lengál,
                  Jan Vašák and
                  Sabína Gulčíková},
  title        = {Artifact for the {PLDI}'26 paper "Towards Efficient Matching of Regexes with Backreferences using Register Set Automata"},
  month        = april,
  year         = 2026,
  publisher    = {Zenodo},
  doi          = {10.5281/zenodo.19223981},
  url          = {https://doi.org/10.5281/zenodo.19223981}
}

@online{rsamatch,
 author = {Jan Vašák},
 year = {2026},
 title = {{\texttt{rsamatch}}},
 url = {https://github.com/VeriFIT/RegisterSetAutomata},
  urldate = {2026-04-15},
}
